%% file: testing.tex
\documentclass[11pt]{article}
\usepackage{amsmath,amssymb,amsfonts,amsthm,epsfig}
\usepackage[usenames,dvipsnames]{xcolor}

\usepackage{bm,xspace}

\usepackage{cancel}
\usepackage{fullpage}
\usepackage{liyang}
\usepackage{framed}
\usepackage{verbatim}
\usepackage{enumitem}
\usepackage{array}
\usepackage{multirow}
\usepackage{afterpage}
\usepackage{mathrsfs}

\usepackage[normalem]{ulem}

\newcommand{\rnote}[1]{\footnote{{\bf \color{red}Rocco}: {#1}}}

\newcommand{\noi}{\boldeta}
\usepackage{caption}

\usepackage{todonotes}

\makeatletter
\newtheorem*{rep@theorem}{\rep@title}
\newcommand{\newreptheorem}[2]{
\newenvironment{rep#1}[1]{
 \def\rep@title{#2 \ref{##1}}
 \begin{rep@theorem}\itshape}
 {\end{rep@theorem}}}
\makeatother
\theoremstyle{plain}

\makeatletter
\newenvironment{proofof}[1]{\par
  \pushQED{\qed}%
  \normalfont \topsep6\p@\@plus6\p@\relax
  \trivlist
  \item[\hskip\labelsep
\emph{    Proof of #1\@addpunct{.}}]\ignorespaces
}{%
  \popQED\endtrivlist\@endpefalse
}
\makeatother

\newcommand{\ignore}[1]{}

\def\colorful{1}
\def\adinclude{1}

\ifnum\adinclude=1
\newcommand{\adin}[1]{{#1}}
\newcommand{\adout}[1]{}
\fi

\ifnum\adinclude=0
\newcommand{\adin}[1]{}
\newcommand{\adout}[1]{{#1}}
\fi

\ifnum\colorful=1

\newcommand{\red}[1]{{\color{red} {#1}}}

\fi
\ifnum\colorful=0

\newcommand{\red}[1]{{{#1}}}

\fi

\usepackage{boxedminipage}

\newreptheorem{theorem}{Theorem}
\newtheorem*{theorem*}{Theorem}
\newreptheorem{lemma}{Lemma}
\newreptheorem{proposition}{Proposition}
\newtheorem*{noclaim*}{Claim}

\newcommand{\ion}{\boldsymbol{\eta}}

\newcommand{\gf}[1]{\mathsf{GF}(#1)}
\newcommand{\gfp}[2]{\mathsf{GF}^{#2}(#1)}
\newcommand{\yes}{\mathrm{yes}}
\newcommand{\no}{\mathrm{no}}
\newcommand{\bdd}{\mathsf{B}}
\newcommand{\midd}{\mathrm{middle}}

\newcommand{\normale}[1]{{N(0,#1)}}

\newcommand{\taua}{\theta}

\newcommand{\Poi}{\mathrm{Poi}}

\usepackage{chngpage}

\usepackage{dsfont}

\renewcommand{\N}{\mathds{N}}
\renewcommand{\R}{\mathds{R}}

\newcommand{\Bern}{\mathrm{Bern}}

\newcommand{\freq}{\mathrm{freq}}

\newcommand{\Noise}{{\cal NOISE}}

\newcommand{\cum}{\kappa}
\newcommand{\ov}{\overline}
\newcommand{\Sym}{\mathrm{Sym}}

\newcommand{\ol}[1]{\overline{#1}}

\begin{document}

\title{
Testing noisy linear functions for sparsity}

\author{\hspace{-35pt}Xue Chen\thanks{{\tt xue.chen1@northwestern.edu}.}  \\
\hspace{-35pt} \small{\sl Northwestern University} \and Anindya De\thanks{{\tt anindyad@cis.upenn.edu}. Supported by NSF grant CCF 1926872} \\ \small{\sl University of Pennsylvania}
 \and Rocco A. Servedio\thanks{{\tt rocco@cs.columbia.edu}. Supported by NSF grants CCF-1814873, IIS-1838154, CCF-1563155, and by the Simons Collaboration on Algorithms and Geometry.}   \\ \small{\sl Columbia University}\vspace*{10pt}
}

\date{\today}

\maketitle

\input{abstract}

\thispagestyle{empty}

\newpage
\setcounter{page}{1}

\input{intro}

\input{techniques}

\input{prelim}

\input{moment_cum_est}

\input{general_algorithm}

\input{gap_nonzero_cumulant}

\input{sym_algorithm}

\input{D-must-be-iid}

\input{noise-known}

\input{Gaussian_lower_bound}

\input{Cumulant_lower_bound}

\bibliography{allrefs}{}
\bibliographystyle{alpha}

\appendix

\input{appen_proofs}

\input{noise-is-necessary}

\end{document}

%% file: abstract.tex

\begin{abstract}
We consider the following basic inference problem:   there is an unknown  high-dimensional vector $w \in \mathbb{R}^n$, and an algorithm is given access to labeled pairs $(x,y)$ where $x \in \mathbb{R}^n$ is a measurement and 
$y = w \cdot x + \mathrm{noise}$. What is the complexity of deciding whether the target vector $w$ is (approximately) $k$-sparse? The recovery analogue of this problem --- given the promise that $w$ is sparse, find or approximate the vector $w$ --- is the famous \emph{sparse recovery} problem, with a rich body of work in signal processing, statistics, and computer science. 


We study the {decision} version of this problem (i.e.~deciding whether the unknown $w$ is $k$-sparse) from the vantage point of \emph{property testing}. 
Our focus is on answering the following high-level question: when is it possible to efficiently \emph{test} whether the unknown target vector $w$ is sparse versus far-from-sparse using a number of samples which is \emph{completely independent} of the dimension $n$?
\adin{\ignore{Suppose}We consider the natural seting in which $x$ is drawn from a i.i.d.~product distribution ${\cal D}$ over $\mathbb{R}^n$ and the $\mathrm{noise}$ process is independent of the input $x$. As our main result, we give a general algorithm 
which solves the above-described testing problem using a number of samples which is completely independent of the ambient dimension $n$, as long as ${\cal D}$ is not a Gaussian. In fact, our algorithm is \emph{fully noise tolerant}, in the sense that 
for an arbitrary $w$, it approximately computes the distance of $w$ to the closest $k$-sparse vector. 
  To complement this algorithmic result, we show that weakening any of our conditions\ignore{on either ${\cal D}$  or the $\mathrm{noise}$}\ignore{\rnote{I ignore'd out ``on either ${\cal D}$  or the $\mathrm{noise}$'' here; is this okay?}}
makes it information-theoretically impossible for \emph{any} algorithm to solve the testing problem with fewer than essentially $\log n$ samples. Thus our conditions essentially characterize when it is possible to test noisy linear functions for sparsity with constant sample complexity.

\ignore{The main algorithmic tool here is to} 

Our algorithmic approach is based on relating the cumulants of the output distribution (i.e.~of $y$) with elementary power sum symmetric polynomials in $w$ and using the latter to\ignore{decide on} measure the sparsity of $w$.  This approach crucially relies on a theorem of Marcinkiewicz from probability theory. In fact, to obtain effective sample complexity bounds with our approach, we prove a new finitary version of Marcinkiewicz's theorem.  This involves extending the complex analytic arguments used in the original proof with results about the distribution of zeros of entire functions. 
 }

\adout{As our main result, we give a general algorithm which, under certain conditions on the input distribution over  $x$ and the noise process, solves the above-described testing problem using a number of samples which is completely independent of the ambient dimension $n$.  We complement this algorithmic result by showing that weakening any of our conditions makes it information-theoretically impossible for \emph{any} algorithm to solve the testing problem with fewer than essentially $\log n$ samples. Thus our conditions essentially characterize when it is possible to test noisy linear functions for sparsity with constant sample complexity.}

\end{abstract}

%% file: intro.tex

\section{Introduction}

This paper addresses a basic data analysis problem from the perspective of \emph{property testing}.  To motivate our work, we begin by considering the following simple and fundamental inference problem: For independent uniform strings $\bx \sim \{-1,1\}^n$ and Gaussian noise $\ion \sim \normale{0.1}$, an algorithm gets access to labeled samples of the form $(\bx, \by)$ where $\by  = w \cdot \bx + \ion$ and $w$ is some fixed but unknown unit vector in $\mathbb{R}^n$. The task of recovering $w$ 
from these noisy samples is an instance of the standard \emph{linear regression} problem, which is of course very well studied in computer science, econometrics,  and statistics (see e.g.~\cite{gelman2006data, greene2003econometric} or many other references).  
As is well known, $\Theta(n)$ samples are both necessary and sufficient to recover $w$ (to within a small constant error), and the ordinary least squares algorithm is a computationally efficient algorithm  which achieves this sample complexity. 

Now suppose that the algorithm is promised that $w$ is \emph{$k$-sparse}, i.e.~it has only $k$ non-zero entries. In this case, the influential line of work on \emph{compressive sensing}  shows that much better sample complexities and running times can be achieved. In particular, the breakthrough work of Candes, Romberg and Tao~\cite{CRT06} shows that using just $m=O(k \log n)$ samples and running in time $\poly(m,n)$, it is possible to (approximately) recover the $k$-sparse vector $w$. Observe that when $k$ is small (like a constant), this is an exponential improvement over the sample complexity achieved by standard linear regression. We further note that by results such as \cite{PriceWoodruff12, ba2010lower}, the bound of $O(k \log n)$ samples is essentially tight, and that compressive sensing algorithms  are applicable for more general choices of the distribution of $\bx$ and the noise $\ion$ (see the survey by Candes~\cite{candes2006compressive}). 

In this paper we consider a natural \emph{decision} analogue of the above problem: the algorithm has access to the same type of $(\bx,\by=w\cdot \bx + \ion)$ samples as above, but it is \emph{not} promised that the target vector $w$ is $k$-sparse. Rather, the task of the algorithm now is to distinguish between the cases that (i) the target vector $w$ is $k$-sparse, versus (ii) the target vector $w$ is $\epsilon$-far from every $k$-sparse vector $w'$ (for some appropriate notion of ``far").  
Using algorithms from compressive sensing, it is straightforward to obtain an algorithm with $m=O(k \log n)$ sample complexity and $\poly(m,n)$ runtime. But can one do much better? In particular, it is \emph{a priori} conceivable that there is an algorithm for this decision problem\footnote{This is in contrast with the recovery problem, as shown by lower bounds such as \cite{PriceWoodruff12, ba2010lower} mentioned above.} with sample complexity \emph{completely independent of the ambient dimension $n$}. Do such ``ultra-efficient'' algorithms in fact exist?

As a corollary of our main algorithmic result, we give an affirmative answer to this question.  Our result implies that in the above setting, it is indeed possible to distinguish between $w$ which (i) is $k$-sparse, versus (ii) is $\epsilon$-far in $\ell_2$ distance from all $k$-sparse vectors $w'$, with an $m=O_{k,\epsilon}(1)$ sample  complexity that is completely independent of $n$.  In fact, we achieve much more: our algorithm can handle a broad range of distributions of $\bx$, and in essentially the same sample complexity we can approximate the distance from $w$ to the closest $k$-sparse vector.  Thus we can essentially determine the ``fit" of the best $k$-sparse vector using only $m=O_{k,\epsilon}(1)$ samples.  (The running time of our algorithm is $\poly(m,n)$.)

\subsection{Motivation:  Property testing}

Before describing our main results in more detail, we recall a line of work on \emph{property testing of functions} which strongly motivates our study. 
In the standard property testing framework, an algorithm is given access to an unknown function $f$ via an oracle $\mathcal{O}$. For a property $\mathcal{P}$ of functions, the goal of a property  testing algorithm for $\mathcal{P}$ is to make as few queries to $\mathcal{O}$ as possible and distinguish (with success probability, say, $9/10$) between the cases that (i) the function $f$ has the property $\mathcal{P}$, versus (ii) the function $f$ is at least $\epsilon$-far in Hamming distance from every function $g$ with property $\mathcal{P}$. 
As a well-known example of this framework, the seminal work of \cite{BLR93} showed that when $\mathcal{P}$ is the property of being $\gf{2}$-linear, $f$ is any function from $\gfp{2}{n}$ to $\gf{2}$, and the oracle $\mathcal{O}$ is a black-box oracle for $f$, then there is an algorithm with query complexity $O(1/\epsilon)$. We refer the reader to books and surveys such as \cite{Ron:08testlearn, Ron:10FNTTCS, PropertyTestingICS,  Goldreich17book} which give an overview of the nearly three decades of work in this area. 

An often-sought-after ``gold standard" for property testing algorithms, that can (perhaps surprisingly) be achieved for many problems, is an algorithm with \emph{constant query complexity}, i.e.~a query complexity that only depends on the error parameter $\epsilon$ and is completely independent of the ambient dimension $n$. This, for example, is the case with $\gf{2}$-linearity testing~\cite{BLR93}, low-degree testing~\cite{gemmell1991self}, junta testing~\cite{fischer2004testing}, and other problems. Indeed, there are grand conjectures (and partial results towards them) which seek to characterize all such properties $\mathcal{P}$ which can be tested with a constant number of queries to a black-box oracle (see e.g.~\cite{KaufmanSudan:08, bhattacharyya2013testing, BhattacharyyaGS15}). 

In this spirit, we explore the question of whether (and when), given noisy labeled samples of the form $(\bx, \by)$ where $\by = w\cdot \bx + \ion$, we can test $k$-sparsity of $w$ with a number of samples that only depends on $k$ and $\eps$, and is independent of $n$. Before describing our precise model, we point out an important difference between our model and much work on property testing.   
In the standard model of property testing of functions described above, it is usually assumed that the algorithm can make black-box queries to the unknown function; in contrast, in our model, the algorithm only has ``passive'' access to random samples.  Obtaining dimension-independent guarantees when given only sample access can be quite challenging; for example, the sample complexity of testing $\gf{2}$ linearity in this model is $\Theta(n)$ samples~\cite{goldreich2016sample} whereas as stated above only $O(1/\eps)$ queries are required by the \cite{BLR93} result.
We refer the reader to~\cite{chen2017sample, kearns2000testing, goldreich2016sample} for some property testing results in the ``sample-based'' model.

\subsection{The problem we consider}   \label{sec:model}

In order to describe the algorithmic problem that we consider in more detail, let us define the notion of \emph{distance to $k$-sparsity}. Given a nonzero vector $w \in \R^n$, we define its distance to $k$-sparsity to be
\begin{equation}~\label{def:distance}
\dist(w,k\text{-sparse}):= \min_{w' \in \R^n: \ w' \text{~is $k$-sparse}} {\frac {\|w - w'\|_2}{\|w\|_2}};
\end{equation}
this is equivalent to the fraction of the $2$-norm of $w$ that comes from the coordinates that are not among the $k$ largest-magnitude ones.
Note that when $w$ is a unit vector, then $\dist(w,k\text{-sparse})$ is the same as the $\ell_2$ distance between $w$ and the closest $k$-sparse vector.

\paragraph{Basic model:} We are now ready to describe our model. We are given access to independent labeled examples of the form $(\bx,\by)$ where $\bx \in \R^n$ and $\by \in \R.$  In each such labeled example $\bx$ is drawn from some distribution ${\cal D}$ over $\R^n$ and the label value $\by$ is a noise-corrupted version of $w \cdot \bx$ for some unknown target vector $w \in \R^n$.  In particular, $\by = w \cdot \bx + \ion$, 
where $\ion$ is drawn from some noise distribution (which is independent of $\bx$).
The goal is to distinguish between the following two cases: (i) $w$ is a $k$-sparse vector (meaning that it has at most $k$ nonzero coordinates), versus (ii) $w$ is \emph{$\eps$-far} from being $k$-sparse (meaning that $\dist(w,k\text{-sparse}) \geq \eps$).  Thus, we are considering a promise problem, or equivalently any output is okay in the intermediate case in which $w$ is not $k$-sparse but is $\eps$-close to being $k$-sparse.   
We refer to this problem as \emph{(non-robust) $k$-sparsity testing}. 

Our algorithms will in fact solve a robust version of this problem:  in the same model as above, for any given $\epsilon>0$, our algorithms will in fact approximate the value of $\dist(w,k\text{-sparse})$ to within an additive $\pm \epsilon$. We refer to this problem as \emph{noise tolerant $k$-sparsity testing} (see Parnas, Ron and Rubinfeld~\cite{PRR06}); it is immediate that any algorithm for this noise-tolerant version  immediately implies an algorithm with the same complexity for non-robust $k$-sparsity testing. In fact, while our main algorithmic result is for the noise tolerant problem, our lower bounds (which we describe later) are for the non-robust version (which \emph{a fortiori} makes them applicable to the noise-tolerant version). 

\paragraph{Our desideratum: constant-sample testability.} \label{remark:fixedk-noqueries}
As is the case for similar-in-spirit property testing problems such as $k$-junta testing \cite{FKRSS03,Blaisstoc09,Bshouty19}, we view $k$ as a parameter which is fixed relative to $n$, and our main goal is to obtain a \emph{constant-sample} tester, i.e.~ a testing algorithm for which the number of samples used is $O_{k,\eps}(1)$ completely independent of $n$.  As stressed earlier (unlike the junta testing problem or many other problems studied in Boolean function property testing), our testing algorithms are \emph{not} allowed to ``actively'' make queries --- their only source of information about $w$ is access to the i.i.d.~samples $(\bx,\by)$ that are generated as described above.

\subsection{Our algorithmic results}

Informally speaking, our main positive result says that for a broad class of input distributions ${\cal D}$, if the parameters of the noise are provided then there is a testing algorithm with $O_{k,\eps}(1)$ sample complexity independent of $n$.
Here is a qualitative statement of our main result  (\Cref{thm:general_alg} gives a more precise version). We start with a description of the algorithmic guarantee for non-robust $k$-sparsity testing.

\begin{theorem} [Qualitative statement of main result]
\label{thm:main-informal}
Fix any random variable $\bX$ over $\R$ which has variance 1, finite moments of every order\footnote{It will be clear from our proofs that having finite moments of all orders is a stronger condition than our algorithm actually requires; we state this stronger condition here for simplicity of exposition.},
 and is not Gaussian (i.e.~its total variation distance from every Gaussian is nonzero).  For any $n$, let ${\cal D}$ be the product distribution over $\R^n$ whose marginals are each distributed according to $\bX$, i.e.~${\cal D} \equiv \bX^n$. Let $\ion$ be a random variable corresponding to a noise distribution over $\R$ which is such that all its moments are finite. \ignore{For $w \in \R^n$ let ${\cal P}_w$ be the distribution over $\R^n \times \R$ defined as follows:  a draw of $(\bx,\by)$ from ${\cal P}_w$ is obtained by drawing $\bx \sim {\cal D}$, independently drawing $\bz \sim \Noise$, and setting $\by \leftarrow w \cdot \bx + \bz.$}

Then there is an algorithm (depending on ${\cal D}$ and $\ion$) with the following property: for any $w \in \R^n$ with\footnote{It may be helpful to think of $C$ as being a large absolute constant, but we establish our results for general $C$.} 
 $1/C \leq \|w\|_2 \leq C$,  given $k, \eps$,  and access to independent samples $(\bx, \by = w \cdot \bx + \boldeta)$ where each $\bx \sim {\cal D}$,\ignore{from ${\cal P}_w$,}
\begin{itemize}
\item if $w$ is $k$-sparse then with probability $9/10$ the algorithm outputs ``$k$-sparse''; and
\item if $w$ is $\eps$-far from $k$-sparse then with probability $9/10$ the algorithm outputs ``far from $k$-sparse.''
\end{itemize}

The number $m$ of samples used by the algorithm depends only on $C, k,\eps, \bX$ and $\boldeta$; in particular, it is independent of $n$.  
We will refer to such an algorithm as an \emph{$\eps$-tester for $k$-sparsity under ${\cal D}$ and $\boldeta$ with sample complexity $m$.}
\end{theorem}

\paragraph{Tolerant testing:}  \label{remark:tolerant}
As mentioned earlier, our algorithmic guarantees are in fact, much stronger. Namely, under the same conditions on ${\cal D}$ and $\boldeta$ as above, 
the algorithm in Theorem~\ref{thm:main-informal}, with high probability, in fact computes $\dist(w,k\text{-sparse})$ to an additive $\pm \eps$. Thus for  ${\cal D}$ and $\boldeta$ as above, this shows that noise tolerant $k$-sparsity testing can be done with a constant number of samples.

\begin{remark} [Explicit bounds and sharper quantitative bounds for ``benign'' distributions] \label{remark:quantitative-preview}
\Cref{thm:main-informal} shows that for every non-Gaussian random variable $\bX$ the corresponding testing problem has a constant-sample algorithm, but it does not give a uniform upper bound on sample complexity that holds for all non-Gaussian distributions.  (Indeed, no such uniform upper bound on sample complexity can exist; see  \Cref{rem:refinement} and \Cref{sec:refinement}, specifically \Cref{thm:refinement}, for an elaboration of this point.)  However, if the background random variable $\bX$ is supported on a bounded set, say $[-\bdd, \bdd]$, then it is in fact possible to get an explicit uniform upper bound on the sample complexity (which is a tower of height $O(k)$). We do this by proving a new finitary version of a theorem due to J.~Marcinkiewicz  \cite{Marcinkiewicz39} from probability theory.  This involves extending the complex analytic arguments used in the original proof; prior to this work, to the best of our knowledge no finitary analogue of the Marcinkiewicz theorem was known \cite{Neeman-comm, Bryc-comm, Janson-comm}. We give this proof in \Cref{sec:gap}.

 Going beyond \Cref{thm:main-informal}, we show that for a large class of ``benign'' distributions (which includes the uniform distribution over $[0,1]$, any product distribution over $\bits$, and many others), a different and simpler algorithm  provides a uniform upper bound on sample complexity, which is roughly $(k/\eps)^{O(k)}$. (See \Cref{sec:symmetric} for a detailed statement and proof of this result.) 
\end{remark}

\subsection{Lower bounds:  Qualitative optimality of our algorithmic results} \label{sec:discussion}
\subsubsection{On the role of noise and its independence of the data points.} We begin by addressing the role of noise in our model. Without noise corrupting the labels, when the background random variable $\bX$ is continuous, even the recovery problem will admit a simple algorithm which uses only $k+1$ samples (see \Cref{sec:noise-is-necc} for an elaboration on this point). Thus, all of our positive results are for settings in which the labels are corrupted by noise. On the other hand, some of our lower bounds are for problem variants in which the labels are noise-free; this of course only makes the corresponding lower bounds stronger.

Secondly, our  model (described in \Cref{sec:model}) requires that the distribution of the noise $\ion$ is independent of the distribution of $\bx$. It is easy to see that if the noise process corrupting the label $\by$ of a labeled example $(\bx,\by)$ is allowed to depend on $\bx$, then it is possible for the noise to perfectly simulate $k$-sparsity when the target vector is far from $k$-sparse or vice versa. In this situation no algorithm, even with infinite sample complexity, can succeed in testing $k$-sparsity.  Thus throughout this work we assume that the noise $\ion$ in each labeled example is independent of the example $\bx$.

\subsubsection{Necessity of the conditions in our algorithmic result.}

There are three main requirements in the conditions of \Cref{thm:main-informal} which may give pause to the reader.  First, the distribution ${\cal D}$ must be an i.i.d.~product distribution: the $n$ coordinate marginal distributions are not only independent, they are identically distributed according to some single univariate random variable $\bX$. Second, certain parameters (various cumulants) of the noise distribution must be provided to the testing algorithm.  And finally, the underlying random variable $\bX$ is not allowed to be a Gaussian distribution.  

While these may seem like restrictive requirements, it turns out that each one is in fact \emph{necessary} for constant-sample testability. We give three different lower bounds which show, roughly speaking, that if any of these requirements is relaxed then finite-sample testability with no dependence on $n$ is information-theoretically impossible --- in fact, in each case the testing problem becomes essentially as difficult as the sparse recovery problem, requiring $\tilde{\Omega}(\log n)$\footnote{The notation ``$\tilde{\Omega}(\cdot)$'' hides factors polylogarithmic in its argument, so $\tilde{\Omega}(\log n)$ means $\Omega({\frac {\log n}{\poly(\log \log n)}}).$} samples.

Our first lower bound shows that even if ${\cal D}$ is allowed to be a product distribution in which half the coordinates are one simple integer-valued distribution (a Poisson distribution) and the other half are a different simple integer-valued distribution (a Poisson distribution with a different parameter), then at least $\Omega({\frac {\log n}{\log \log n}})$ samples may be required. This lower bound holds even if no noise is allowed. The proof is given in \Cref{sec:iid}:

\begin{theorem} [${\cal D}$ must be i.i.d.] \label{thm:iid}
Let ${\cal D}$ be the product distribution $(\Poi(1))^{n/2} \times
(\Poi(100))^{n/2}$.  Then even if there is no noise (i.e.~ the noise distribution $\boldeta$ is identically zero), any algorithm which is an $(\eps=0.99)$-tester for 1-sparsity under ${\cal D}$ must have sample complexity $m = \Omega({\frac {\log n}{\log \log n}}).$
\end{theorem}

Our second lower bound shows that even if only two ``known'' possibilities are allowed for the noise distribution, then for ${\cal D}=\bX^n$ where $\bX$ is a simple ``known'' integer-valued underlying univariate random variable, at least $\Omega({\frac {\log n}{\log \log n}})$ many samples may be required.  The proof is given in \Cref{ap:two-noise}:

\begin{theorem} [The noise distribution $\boldeta$ must be known]
\label{thm:two-noise}
Let ${\cal D}$ be the i.i.d.~product distribution ${\cal D}=(\Poi(1))^n$.  Suppose that the noise distribution $\boldeta$ is unknown to the testing algorithm but is promised to be either $\Poi(1)$ or $\Poi(100)$.  Then any $(\eps=0.99)$-tester for 1-sparsity under ${\cal D}$ and the unknown noise distribution $\boldeta \in \{\Poi(1),\Poi(100)\}$ must have sample complexity $m = \Omega({\frac {\log n}{\log \log n}}).$ 
\end{theorem}

Finally, our third (and most technically involved) lower bound says that if the underlying univariate random variable $\bX$ is allowed to be a Gaussian, then even if the noise is Gaussian at least $\Omega(\log n)$ samples are required. The proof is given in \Cref{sec:gaussian-lb}:

\begin{theorem} [${\cal D}$ cannot be a Gaussian] \label{thm:gaussian-lb}
Let ${\cal D}$ be the standard $N(0,1)^n$ $n$-dimensional Gaussian distribution and let $\boldeta$ be distributed as $N(0,c^2)$ where $c>0$ is any constant. Then the sample complexity of any $(\eps=0.99)$-tester for $1$-sparsity under ${\cal D}$ and $\boldeta$ is $\Omega(\log n).$
\end{theorem}

\subsection{Related work}

We view this paper as lying at the confluence of several strands of research in theoretical computer science. As mentioned earlier, a strong motivation for our algorithmic desiderata comes from property testing. In particular, our $k$-sparsity testing question is in some sense akin to the well-studied problem of \emph{junta testing}, i.e., distinguishing between functions $f: \{\pm 1\}^n \rightarrow \{\pm 1\}$ which depend on at most $k$ coordinates versus those which are $\epsilon$-far from every such function.
There is a very rich line of work on junta testing, see e.g.~\cite{FKR+:04, Blaisstoc09, Blais10survey, Bshouty19, blais2019tolerant, ChocklerGutfreund:04} and other works. However, we note that all these papers (and other junta testing papers of which we are aware) assume query access to the unknown function $f$, whereas in our work we only assume a much weaker form of access, namely noisy labeled random samples.

A second strand of work is from compressive sensing.  Here the results of \cite{CRT06} and related works such as  \cite{candes2008restricted, candes2007sparsity} (as well as many other papers) give computationally efficient algorithms to (approximately) recover a sparse vector $w$ given labeled samples of the form $\{(\bx^{(i)}, w \cdot \bx^{(i)} + \noi)\}_{i=1}^T$ with sample complexity $T=O(k \log n)$. On one hand, such a sample complexity does not meet our core algorithmic desideratum of being independent of $n$. On the other hand, the algorithmic guarantee in \cite{CRT06} holds as long as the matrix formed by $ \bx^{(1)}, \ldots, \bx^{(T)}$ satisfies the so-called \emph{restricted isometry property} (see~\cite{candes2006compressive} for more details), which is a significantly more general condition than ours.  It is natural to wonder if an analogue of \Cref{thm:main-informal} can be obtained if ${\cal D}$ satisfies the weaker condition of being such that randomly drawn samples from ${\cal D}$ satisfy the restricted isometry property with high probability.  The answer to this question is negative; in particular, \Cref{thm:iid} gives an example of a distribution ${\cal D}$ for which $\tilde{\Omega}(\log n)$ samples are necessary for testing $k$-sparsity, but it is easy to show that randomly drawn samples from this distribution satisfies the restricted isometry property with high probability. 

Finally, another related line of work is given by Kong and Valiant~\cite{kong2018estimating}, who considered a setting in which an algorithm gets labeled samples of the form $(\by, w \cdot \bx + \ion)$, where $\ion$ is an \emph{unknown} distribution independent of $\bx$ and $w$ is a general (non-sparse) $n$-dimensional vector. The task of the algorithm is to estimate the variance of $\ion$ or equivalently, $\Vert w \Vert_2$; they view such a result as \emph{estimating how much of the data, i.e., $\by$, is explained by the linear part $w \cdot \bx$.} While learning $w$ itself requires $\Theta(n)$ samples (essentially the same as linear regression), their main result is that $\Vert w \Vert_2$ can be estimated with a sublinear number of samples. In particular, if the distribution of $\bx$ is isotropic, then the sample complexity required for this is only $O(\sqrt{n})$. In light of \Cref{thm:main-informal} and the results of \cite{kong2018estimating}, it is natural to ask whether there is a non-trivial estimator for noise in our setting when the target vector $w$ is assumed to be $k$-sparse. However, 
\Cref{thm:two-noise} essentially answers this in the negative, showing that if the magnitude of the noise is unknown, then any estimator must require  $\tilde{\Omega}(\log n)$ samples even for $1$-linearity testing. On the other hand,  $O(\log n)$ samples suffice for recovering the target $w$ (and hence the magnitude of the noise) when $k$ is a constant.

%% file: techniques.tex

\section{Our techniques and a more detailed overview of our results}
\label{sec:techniques}

\subsection{Our algorithmic techniques:  analysis based on cumulants}

Both of our algorithms for testing sparsity make essential use of the \emph{cumulants} of the one-dimensional coordinate marginal random variable $\bX$. For any integer $\ell \geq 0$ and any real random variable $\bX$, the $\ell$-th cumulant of $\bX$, denoted $\cum_\ell(\bX)$, is defined in terms of the first $\ell$ moments of $\bX$, and, like the moments of $\bX$, it can be estimated using independent draws from $\bX$ (see \Cref{def:cumulants} for a formal definition of cumulants.) However, cumulants enjoy a number of attractive properties which are not shared by moments and which are crucial for our analysis.  

There are two key properties, both very simple.  First, cumulants are \emph{additive} for independent random variables:
\[
\text{If~$\bX,\bY$ are independent, then~}\cum_\ell(\bX+\bY) =
\cum_\ell(\bX) + \cum_\ell(\bY).
\]
Second, cumulants are \emph{homogeneous}:
\[
 \text{For all~$c \in \R$, it holds that~}\cum_\ell(c\bX) = c^\ell \cdot \cum_\ell(\bX).
\]

We now explain the key idea of why additivity and homogeneity of cumulants are useful for the algorithmic problem we consider.  These properties directly imply that if a distribution ${\cal D}$ over $\R^n$ has coordinate marginals that are i.i.d.~according to a random variable $\bX$, then for $\bx \sim {\cal D}$ and $\by = w \cdot \bx + \ion$, we have that
\[
\cum_\ell(\by) = \cum_\ell(\ion) +  \cum_\ell(\bX) \cdot \sum_{i=1}^n w_i^\ell.
\]
It follows that if the $\ell$-th cumulants of $\ion$ and of $\bX$ are known and the $\ell$-th cumulant $\cum_\ell(\bX)$ of $\bX$ is not too small, then from an estimate of $\cum_\ell(\by)$ (which can be obtained from samples) it is possible to obtain an estimate of the power sum $\sum_{i=1}^n w_i^\ell$.  By doing this for $k$ suitable different (even) values of $\ell$, provided that the cumulants $\cum_\ell(\bX)$ are not too small, it is possible to estimate the magnitudes of the $k$ largest-magnitude coordinates of $w$.  These estimates can be shown to yield the desired information about whether or not $w$ is (close to) $k$-sparse.

The argument sketched in the previous paragraph explains, at least at an intuitive level, why it is possible to test for $k$-sparsity \emph{if the random variable $\bX$ has $k$ nonzero cumulants}\footnote{(of even order, though as we will see, a simple trick lets us sidestep this requirement)}.  But why will every non-Gaussian random variable $\bX$ (as described in \Cref{thm:main-informal}) satisfy this property, and why does \Cref{thm:main-informal} exclude Gaussian distributions?  The second of these questions has a very simple answer so we address it first:  it is well known that for any normal distribution $\bX \sim N(\mu,\sigma^2)$, the first two cumulants are $\cum_1(\bX)=\mu$, $\cum_2(\bX)=\sigma^2$, and all other cumulants are zero. It follows that indeed our algorithmic approach cannot be carried out for normal distributions.\footnote{Recall that by our lower bound \Cref{thm:gaussian-lb}, this is not a failing of our particular algorithm sketched above but an inherent difficulty in the testing problem. \Cref{thm:gaussian-lb} shows that no algorithm can test $k$-sparsity with a sample complexity that is $o(\log n)$ when the underlying distribution is normal.}  The answer to the first question comes from a deep result in probability theory due to J.~Marcinkiewicz:

\begin{theorem} [Marcinkiewicz's theorem \cite{Marcinkiewicz39,Marcink,bryc}] \label{thm:Marcink_theorem}
If $\bX$ is a random variable that has a finite number of nonzero cumulants, then $\bX$ must be a normal random variable (and $\bX$ has at most two nonzero cumulants).
\end{theorem}

It follows that if $\bX$ is not a normal distribution, then it must have infinitely many nonzero cumulants, and hence the algorithmic approach sketched above can be made to work for testing $k$-sparsity under $\bX^n$. Details of the estimation procedure and of the analysis of the overall general algorithm are provided in \Cref{sec:est_cum_moment} and \Cref{sec:proof-of-main-informal} respectively.

\medskip

\subsection{A structural result on cumulants:  nonzero cumulants cannot be ``spaced too far apart''} \label{sec:structural-overview}

As described above, our main positive result on testing for sparsity under a product distribution $\bX^n$ uses a sequence of orders $i_1, i_2, \dots, i_k$ such that the corresponding cumulants $\cum_{i_j}(\bX)$ are all nonzero.  Since the running time of our algorithm depends directly on $i_k$, it is natural to ask how large is this value.  Recall that Marcinkiewicz's theorem ensures that for any non-Gaussian distribution there indeed must exist nonzero cumulants of infinitely many orders $i_1, i_2, \dots$, but it gives no information about how far apart these orders may need to be.  Thus we are motivated to investigate the following question:  given a real random variable $\bX$, how large can the gap in orders be between consecutive nonzero cumulants?  This is a natural question which, prior to our work, seems to have been completely unexplored.

In \Cref{sec:gap} we give a first result along these lines, by giving an explicit upper bound on the gap between nonzero cumulants for random variables with \emph{bounded support} (see \Cref{thm:distance_non_zero_cumulants}).  This theorem establishes that for any real random variable $\bX$ with unit variance and support bounded in $[-\bdd,\bdd]$, given any positive integer $\ell$ there must be a value $j \in [\ell+1,(4\bdd)^{O(\ell)}]$ such that the $j$-th cumulant $\kappa_j(\bX)$ has magnitude at least $\big| \kappa_j(\bX) \big| \ge 2^{-(4\bdd)^{O(\ell)}}$.  Like Marcinkiewicz's theorem, the proof of our \Cref{thm:distance_non_zero_cumulants} uses complex analytic arguments, specifically results about the {distribution of zeros of entire functions, the Hadamard factorization of entire functions and the Hadamard Three-Circle Theorem.}

\subsection{A more efficient algorithm for ``nice'' distributions}

In addition to the general positive result described above,\ignore{
The above gives an intuitive sketch of the reasoning behind \Cref{thm:main-informal}.  We} we also give a refined result, showing that a significantly better sample complexity can be achieved for distributions which are ``nice'' in the sense that they have $k+1$ consecutive even cumulants $\kappa_2, \kappa_4,\dots, \kappa_{2k+2}$ that are all (noticeably) nonzero.  This is achieved via a different algorithm; like the previously described general algorithm, it uses (estimates of) the power sums $\sum_{i=1}^n w_i^\ell$, but it uses these power sums in a different way, by exploiting some basic properties of symmetric polynomials.  The first $k+1$ power sums $\sum_{i=1}^n w_i^2,$ $\sum_{i=1}^n w_i^4$, $\dots$ are used to estimate the $(k+1)$-st elementary symmetric polynomial 
\begin{equation} \label{eq:elementary}
\sum_{1 \leq i_1 < i_2 < \cdots < i_{k+1} \leq n} w_{i_1} w_{i_2} \cdots w_{i_{k+1}}.
\end{equation}
The value of \eqref{eq:elementary} will clearly be zero if $w$ is $k$-sparse, and it can be shown that it will be ``noticeably far from nonzero'' if $w$ is far from $k$-sparse.  These ideas can be converted into a testing algorithm; see \Cref{sec:symmetric} for details.

\subsection{Our lower bounds and lower bound techniques}
\label{sec:techniques-lb}

The lower bounds of \Cref{thm:iid} and \Cref{thm:two-noise} both crucially exploit the well known additivity property of the Poisson distribution:  for $a,b > 0$, we have that $\Poi(a)+\Poi(b) = \Poi(a+b)$. 
To see why this  is useful for lower bounds, let us explain the high-level idea that underlies~\Cref{thm:iid}. For intuition, first imagine that rather than receiving pairs $(\bx,\by) \in \R^n \times \R$, instead the testing algorithm is only given the output value $\by$ from each pair.  Then by the additivity of the Poisson distribution, it would be information-theoretically impossible to distinguish between (i) the case in which $\by$ is a sum of $100$ coordinates each of which is distributed as $\Poi(1)$ (and hence the target vector $w$ is $0.99$-far from being 1-sparse), versus (ii) the case in which $\by$ is a single coordinate distributed as $\Poi(100)$ (and hence the target vector $w$ is 1-sparse).  Of course, in our actual testing scenario things are not so simple because the testing algorithm does receive the coordinates $\bx_1,\dots,\bx_n$ of each example $(\bx,\by)$ along with the value of $\by$, and this provides additional useful information.  Our proof establishes that this additional information is essentially useless unless $\tilde{\Omega}(\log n)$ samples are provided. Roughly speaking, this is because with $n/2$ coordinates distributed as $\Poi(1)$ and $n/2$ coordinates distributed as $\Poi(100)$, there are ``too many possibilities'' of each sort ((i) and (ii) above) for the $\bx$'s to provide useful distinguishing information until this many samples have been received.

The lower bound of \Cref{thm:two-noise} is based on similar ideas.  Now all $n$ coordinates are identically distributed as $\Poi(1)$, but the noise may be distributed either as $\Poi(1)$ or as $\Poi(100)$. As above, if only the output values $\by = w \cdot \bx + \ion$ were available to the tester, it would be impossible to distinguish between (i$'$) the target vector $w$ is 1-sparse and the noise is $\Poi(100)$, versus (ii$'$) the target vector is 100-sparse and the noise is $\Poi(1)$, since in both cases the distribution of $\by$ is $\Poi(101).$  The formal proof is by a reduction to \Cref{thm:iid}.


Finally, we turn to the lower bound of \Cref{thm:gaussian-lb}, which states that $\Omega(\log n)$ samples are required for the testing problem if the distribution ${\cal D}$ is $N(0,1)^n$ and the noise distribution $\ion$ is normally distributed as $N(0,c^2)$. The high level idea is that it is difficult to distinguish between the following two distributions over pairs $(\bx,\by)$:

\begin{itemize}

\item First distribution (no-distribution): in each draw of $(\bx,\by)$ from the no-distribution, each $\bx_j$ is an independent $N(0,1)$ random variable, and $\by$ is an $N(0,1+c^2)$ normal random variable which is completely independent of all of the $\bx_j$'s;

\item Second distribution (yes-distribution):  there is a fixed but unknown uniform random coordinate $\bi \in [n]$, and in each draw of $(\bx,\by)$ from the yes-distribution, each $\bx_j$ is an independent $N(0,1)$ random variable and $\by = \bx_{\bi} + N(0,c^2)$.

\end{itemize}

Similar to the first paragraph of this subsection, since the sum of a draw from $N(0,1)$ plus an independent draw from $N(0,c^2)$ is a draw from $N(0,1+c^2)$, if only the output value $\by$ from each pair were given to a tester then it would be information-theoretically impossible to distinguish between the two distributions described above.  And similar to the discussion in that paragraph, the idea that animates our lower bound proof here is that the additional information (the $\bx_1,\dots,\bx_n$-coordinates of each sample) available to the testing algorithm is essentially useless unless $\Omega(\log n)$ samples are provided.  As before, roughly speaking, this is because there are ``too many possibilities'' (for which coordinate might be the unknown hidden $\bi \sim [n]$ in the second distribution) for the $\bx$-components of the samples to provide useful distinguishing information until $\Omega(\log n)$ many samples have been received.  The formal argument uses Bayes' rule to analyze the optimal distinguishing algorithm (corresponding to a maximum likelihood approach) and employs the Berry-Esseen theorem to make these intuitions precise.

\begin{remark} \label{rem:refinement}
We note here that we also give a quantitative refinement of \Cref{thm:gaussian-lb}.  Since for a Gaussian random variable $\bX$ all cumulants $\cum_\ell(\bX), \ell>2$, are zero, we may informally view \Cref{thm:gaussian-lb} as saying that if the cumulants of $\bX$ are zero then the number of samples required to test for linearity under $\bX^n$ may be arbitrarily large (going to infinity as $n$ does). This intuitively suggests that if the cumulants of $\bX$ are ``small'' then ``many'' samples should be required to test for linearity under $\bX^n$.  In~\Cref{sec:refinement} we make this intuition precise: building on~\Cref{thm:gaussian-lb}, we show (roughly speaking) that if the cumulants of a random variable $\bX$ are at most $\gamma$, then at least $1/\gamma$ samples are required for testing linearity under $\bX^n$ and Gaussian noise.
See~\Cref{thm:refinement}  for a precise statement and proof.
\end{remark}

\subsection{Directions for future work}

Our results suggest a number of directions for future work; we touch on a few of these below.  

Within the linearity testing framework that this paper considers, it would be interesting to gain a more quantitatively precise understanding of the sample complexity required to test linearity.  A natural specific question here is the following:  let $\bX$ be a simple random variable such as $\bX = $ uniform on $\{-1,1\}$ or $\bX = $ uniform on $[0,1]$. For these specific distributions, what is the optimal dependence on $k$ for the $k$-sparsity testing question that we have considered?  It would be interesting to determine whether or not an exponential dependence on $k$ is required. 

Another natural quantitative question arises from our results in \Cref{sec:gap}.
\Cref{thm:distance_non_zero_cumulants} implies an explicit ``tower-type'' upper bound on the minimum value $i_k$ such that a random variable $\bX$ as above must have at least $k$ nonzero cumulants in $\{1,\dots,i_k\}$. It would be interesting to obtain sharper quantitative bounds or bounds that hold under relaxed conditions on the random variable $\bX$.

Finally, another intriguing potential direction is to look beyond linearity and attempt to identify other contexts in which sparsity is testable with a constant sample complexity independent of $n$.  A concrete first goal along these lines is to investigate the sparsity testing question when $(\bx,\by)$ is distributed as $\by = \phi(w \cdot \bx) + \textrm{noise}$ for various natural transfer functions $\phi$ such as the probit function or the logistic function.

\subsection{Notational conventions}

Given a vector $w \in \R^n$ we write $\|w\|_\ell$ to denote the $\ell$-norm of $w$, i.e.~$\|w\|_\ell = \left(\sum_{i=1}^n w_i^\ell\right)^{1/\ell}.$  For a nonzero vector $w \in \R^n$ where $n > k$, the vector $w$'s \emph{distance from being $k$-sparse} is 

\[
\dist(w,k\text{-sparse}):= \min_{w' \in \R^n: \ w' \text{~is $k$-sparse}} {\frac {\|w - w'\|_2}{\|w\|_2}}.
\]
Equivalently, if the entries of $w$ are sorted by magnitude so that $|w_{i_1}| \geq \cdots \geq |w_{i_n}|$, the distance of $w$ from being $k$-sparse is
\[
{\frac {\sqrt{w_{i_{k+1}}^2 + \cdots + w_{i_n}^2}}{\|w\|_2}}.
\]

For a random variable $\bZ$, we write $m_{\ell}(\bZ)$ to denote its $\ell$th raw moment, i.e., $\E[\bZ^{\ell}]$.

%% file: prelim.tex

\section{Preliminaries:  Facts about Cumulants} \label{sec:prelims}

In this section we recall some basic facts about cumulants which we will use extensively.

\begin{definition} \label{def:cumulants}
The \emph{cumulants} of $\bX$ are defined by the cumulant generating function $K(t)$, which is the natural logarithm of the moment generating function $M(t) = \E[e^{t \bX}]$:
\[
K(t)=\ln \E[e^{t\bX}].
\] Equivalently, $e^{K(t)}= \E[e^{t\bX}]$. For $\ell > 0$ the cumulants of $\bX$, which are denoted $\cum_\ell(\bX)$, are the coefficients in the Taylor expansion of the cumulant generating function about the origin:
\[
K(t) = \sum_{\ell=1}^\infty \cum_\ell(\bX) {\frac {t^\ell}{\ell!}}.
\]
Equivalently, $\cum_\ell(\bX)=K^{(\ell)}(0)$.
\end{definition}

One useful property of cumulants is additivity for independent random variables, which follows as an easy consequence of the definition:

\begin{fact} \label{fact:additive}
If $\bX$ and $\bY$ are independent random variables then $\cum_\ell(\bX+\bY)=\cum_\ell(\bX)+\cum_\ell(\bY)$.
\end{fact}
\begin{corollary}\label{cor:sym_var_cumulant}
For any random variable $\bX$, the value of $\cum_\ell(\bX-\bX)$ is zero when $\ell$ is odd and is $2 \cdot \cum_\ell(\bX)$ when $\ell$ is even.
\end{corollary}

Looking ahead, all of our algorithms will work by estimating cumulants of the real random variable $\by$ which is distributed as $\by = w \cdot \bx + \ion$ where $\bx \sim \bX^n$ and $\ion$ is independently drawn from a noise distribution.  By \Cref{fact:additive} and \Cref{cor:sym_var_cumulant}, we can (and do) assume throughout the analysis of our algorithms that $\bX$ and $\ion$ are both symmetric distributions. This is because we can always combine two independent draws $(\bx_1,\by_1)$ and $(\bx_2,\by_2)$ with $\by_i=w \cdot \bx_i + \ion$ into one draw $(\frac{\bx_1-\bx_2}{\sqrt{2}}, \frac{\by_1-\by_2}{\sqrt{2}}),$ such that the new marginal distribution $\frac{\bX-\bX}{\sqrt{2}}$ and noise distribution $\frac{\ion-\ion}{\sqrt{2}}$ are both symmetric and have the same variances as before combining.

Another useful property is $\ell$-th order homogeneity of the $\ell$-th cumulant:
\begin{fact} \label{fact:homogeneous}
For any $c \in \R$ and any $\ell \in \N$, we have $\cum_\ell(c\bX)=c^\ell \cum_\ell(\bX).$
\end{fact}

Let $m_\ell(\bX)$ denote the $\ell$th moment $\E[\bX^\ell]$ of a random variable $\bX$.
There is a one-to-one mapping between the first $n$ moments and the first $n$ cumulants which can be derived by relating coefficients in the Taylor series expansions of the cumulant and moment generating functions~\cite{Barndorff}:

\begin{fact}\label{fact:transfer_moments_cumulants}
Let $\bX$ be a random variable with mean zero.  Then
\begin{equation}\label{eq:degree_l_cum}
\cum_\ell(\bX)=m_\ell(\bX)-\sum_{j=1}^{\ell-1} {\ell-1 \choose j-1} \cum_j(\bX) \cdot m_{\ell-j}(\bX).
\end{equation}
Cumulants can be expressed in terms of moments and vice-versa:
\begin{equation}\label{eq:moments_cumulants}
m_\ell(\bX)=\sum_{k=1}^\ell B_{\ell,k}\bigg( \cum_1(\bX),\ldots,\cum_{\ell-k+1}(\bX) \bigg)
\end{equation}
and
\begin{equation}\label{eq:cumulants_moments}
\cum_\ell(\bX)=\sum_{k=1}^\ell (-1)^{k-1} (k-1)! B_{\ell,k}\bigg( m_1(\bX),\ldots,m_{\ell-k+1}(\bX) \bigg	),
\end{equation}
where $B_{\ell,k}$ are incomplete \emph{Bell polynomials},
$$
B_{\ell,k}(x_1,\ldots,x_{\ell-k+1})=\sum \frac{\ell!}{j_1!\cdots j_{\ell-k+1}!} \left(\frac{x_1}{1} \right)^{j_1} \cdots \left(\frac{x_{\ell-k+1}}{(\ell-k+1)!}\right)^{j_{\ell-k+1}},
$$
where the summation is over all non-negative sequences $(j_1,\ldots,j_{\ell-k+1})$ that satisfy
$$
j_1+\cdots+j_{\ell-k+1}=k \text{ and } j_1+2j_2 + \cdots + (\ell-k+1)j_{\ell-k+1}=\ell.
$$
\end{fact}

\Cref{eq:degree_l_cum} can be used to give a upper bound on $\cum_\ell(\bX)$ in terms of the moments of $\bX$:

\begin{claim}\label{clm:up_bound_cumulant}
For any random variable $\bX$ with mean zero and any even $\ell$, we have $|\cum_\ell(\bX)| \le m_\ell(\bX) \cdot e^{\ell} \cdot \ell!$.
\end{claim}

\begin{remark}
When $\bX$ is the random variable that is uniform over $\{0,1,\ldots,C\}$, the $\ell$-th cumulant is $\cum_\ell(\bX)=\frac{\Bern(\ell)}{\ell} \cdot (C^\ell - 1)$ where $\Bern(\ell)$ is the Bernoulli number of order $\ell$ which has an asymptotic growth as $(\frac{\ell/2}{\pi e})^\ell$ \cite{Swanson2017}. This simple example shows that the dominant $\ell!$ term in \Cref{clm:up_bound_cumulant} is essentially best possible.
\end{remark}

We defer the proof of~\Cref{clm:up_bound_cumulant} to~\Cref{sec:append_proofs}.

%% file: moment_cum_est.tex

\section{Estimating moments of the weight vector $w$ using moments and cumulants}\label{sec:est_cum_moment}

\ignore{
}




Throughout this section $\bX$ will denote a real random variable with mean zero, unit variance, and finite moments of all orders, and $w \in \R^n$ will be a vector that is promised to have $\|w\|_2 \in [1/C,C]$. 
The main result of this section is the following theorem, which shows that it is possible to estimate norms of the vector $w$ given access to noisy samples of the form $(\bx, \by=w\cdot \bx + \ion)$ where $\bx \sim \bX^n$:

\begin{theorem}\label{thm:number_samples_Lk_w}
Let $\bX$ be a symmetric real-valued random variable with variance 1 and finite moments of all orders, and let $\ion$ be a symmetric real-valued random variable with finite moments of all orders.\ignore{
}
There is an algorithm (depending on $\bX$ and $\ion$)\footnote{As will be clear from the proof, the algorithm only needs to ``know'' $\bX$ and $\ion$ in the sense of having sufficiently accurate estimates of certain cumulants.} with the following property:  Let $w \in \R^n$ be any (unknown) vector with $\|w\|_2 \in [1/C,C]$.
Given any $\eps,\delta>0$ and any even integer $\ell$ such that $|\cum_\ell(\bX)| \geq \tau$, the algorithm takes as input $m=\poly( \ell!,m_{2\ell}(\bX)+m_{2\ell}(\ion),1/(\delta\eps),1/\tau, C^{\ell})$ many independent random samples where each $\bx^{(i)}\sim \bX^n$ and each $\bz^{(i)}=w \cdot \bx^{(i)} + \ion.$  It outputs an estimate $M_\ell$ of $\sum_{i=1}^n w_i^\ell$, the $\ell$-th power of the $\ell$-norm of the vector $w$, which with probability at least $1-\delta$ satisfies
\[
\left| M_\ell - \|w\|_\ell^\ell \right| \le \eps.
\]
\end{theorem}

We will also use the following result on estimating the moments of $w \cdot \bX + \ion$:

\begin{lemma}\label{lemma:number_samples_k_moments}
Let $\bX$ be a symmetric real-valued random variable with mean zero, variance 1, and finite moments of all orders, and let $\ion$ be a symmetric real-valued random variable with finite moments of all orders. 

There is an algorithm (depending on $\bX$ and $\ion$) with the following property:  Let $w \in \R^n$ be any (unknown) vector with $\|w\|_2 \in [1/C,C]$. Given any $\eps$ and $\delta$ and any even integer $\ell$, the algorithm takes as input $m=\poly( \ell!,m_{2\ell}(\bX)+m_{2\ell}(\ion),1/(\delta\eps), C^{\ell})$ many independent random samples where each $\bx^{(i)}\sim \bX^n$ and each $\bz^{(i)}=w \cdot \bx^{(i)} + \ion.$  It outputs an estimate $\wt{m}_{\ell}(\bZ)$, which with probability at least $1-\delta$ satisfies
\[
\left| \wt{m}_{\ell}(\bZ) - \E[|w \cdot \bx^{(i)} + \ion|^{\ell}] \right| \le \eps.
\]
\end{lemma}


We defer the proof of \Cref{thm:number_samples_Lk_w} and \Cref{lemma:number_samples_k_moments} to \Cref{sec:append_proofs}.

%% file: general_algorithm.tex

\section{General Testing Algorithm: Proof of \Cref{thm:main-informal}} \label{sec:proof-of-main-informal}

The main result of this section is \Cref{thm:general_alg}, which is a more precise version of \Cref{thm:main-informal}. Roughly speaking, it says that there is a constant-sample tolerant tester for $k$-sparsity for any non-Gaussian distribution.  

\begin{theorem} [Detailed statement of main result: tolerant tester for non-Gaussian distributions]
\label{thm:general_alg}
\ignore{\rnote{Two things about this theorem statement.  First, what is the exact definition we want to use for $w$ being $\eps$-far from $k$-sparse. Is it that we always assume $\|w\|=1$?  Or do we allow any scaling of $w$? I had thought we would do the latter (that's how some of the lower bounds are structured and that's what's stated in the intro) but this section seems to be using a definition which always assumes that $\|w\|=1.$  (Maybe a setup where we assume that $\|w\|=\Theta(1)$ is a middle ground we could take?) {\color{brown} Xue: I agree with the latter. For general $\|w\|=\Theta(1)$, we shall be able to estimate the variance of $w \cdot \bX^n + \ion$ very accurately and rescale $\|w\|$ to $1 \pm \eps/100$.}
}

}
Fix any real random variable $\bX$ which has variance one and finite moments of every order, and is not a Gaussian distribution (i.e.~its total variation distance from every Gaussian distribution is nonzero).  Let $\ion$ be any real random variable with finite moments of every order. 

There is a tolerant testing algorithm with the following properties:
Let $0 \le c < s \le 1$ be any given completeness and soundness parameters, and let $w$ be any vector (unknown to the algorithm) with $1/C \leq \|w\|_2 \leq C.$  
The algorithm is given $c,s,\eps,k,C$ and access to independent samples $(\by = w \cdot \bx + \ion)$ where each $\bx \sim \bX^n$. Its sample complexity is 
\[
m=\poly\parens*{ \ell_1!,m_{2\ell_1}(\bX)+m_{2\ell_1}(\ion),1/\delta_1^{\ell_1},1/\tau,C^{\ell_1} },
\] where $\tau = \min_{i \in [k]}\{|\cum_{\ell_i}(\bX)|\}$ and $\{\ell_i\}_{i \in [k]},\{\delta_i\}_{i \in [k]}$ are as defined below. The algorithm satisfies the following:

\begin{itemize}

\item if $\dist(w,k\text{-sparse}) \leq c$ then with probability at least $9/10$ the algorithm outputs ``yes;'' and

\item if $\dist(w,k\text{-sparse}) \geq s$ then with probability at least $9/10$ the algorithm outputs ``no.''

\end{itemize}
\ignore{
Given any $c$ and $s$ with $1 \ge s >c \ge 0$, for vectors $w \in \mathbb{R}^{n}$ with $\|w\|_2 \in [1/10, 10]$, the algorithm described below is an \red{$c$-versus-$s$ tolerant tester} for $k$-linearity under ${\cal D}$ and $\ion$.  Its sample complexity is $m=\poly\parens*{ \ell_k!,m_{2\ell_k}(\bX)+m_{2\ell_k}(\ion),1/\delta_k^{\ell_k},1/\tau }$, where \red{$\tau = \min_{i \in [k]}\{|\cum_{\ell_i}(\bX)|\}$} and $\ell_i,\ell_k,\delta_k$ are as defined below. 
}

Furthermore, if the random variable $\bX$ is supported in $[-\bdd,\bdd]$ for some constant $\bdd$, then the sample complexity of the tolerant tester (as a function of $k$) \ignore{described in \Cref{thm:general_alg}} is bounded by a tower function of height $O(k)$.
\end{theorem}



\ignore{




}


\ignore{
}

We begin by stating the algorithm:

\begin{enumerate}
\item First, recall that as stated earlier, we may assume that $\bX$ and $\ion$ are both symmetric. We rescale all samples by a factor of $C$ so that $\|w\|_2 \in [1/C^2,1]$ in our subsequent analysis. 

We fix $\eps=\frac{s^2 - c^2}{2C^4}$ and apply \Cref{lemma:number_samples_k_moments} with $\ell=2$ to obtain an estimate $s_2$ of $\sum_{i=1}^n w_i^2=\|w\|_2^2=\E[|w \cdot \bX^n + \ion|^2]-\E[|\ion|^2]$ that is accurate to within additive accuracy $\eps/4$ (with probability 0.99). 

\item Set a sequence of error parameters $\delta_1<\delta_2<\ldots<\delta_k$ and natural numbers (orders of cumulants) $\ell_1>\ell_2>\cdots>\ell_k$ with the following properties:
\begin{enumerate}

\item 
$\delta_k =\eps/(12k)$ and $\ell_k \geq 100/\delta_k^3$ is even \ignore{\rnote{The statement of \Cref{clm:estimation_error} required $\ell$ to be even (though I am not totally sure why) so I added it in here}};

\item For $i=k-1,k-2,\ldots,1,$  $\delta_{i} = (\delta_{i+1}/5 \ell_{i+1})^{\ell_{i+1}}/(2k)$ and $\ell_{i} \geq 100/\delta_i^3$ is even;

\item For each $i \in [1,k]$ the $\ell_i$-th cumulant $\cum_{\ell_i}(\bX)$ of $\bX$ is nonzero.

\end{enumerate}

\item For $j=1,\dots,k$: run the algorithm of \Cref{thm:number_samples_Lk_w} to obtain an estimate $M_{\ell_j}$ which satisfies $\big| M_{\ell_j} - \|w\|_{\ell_j}^{\ell_j} \big| \le (\delta_j/5\ell_j)^{\ell_j}/(2k)$ with failure probability at most $1/(20k).$   Set $\wt{w}_j=\min\braces*{1,\abs*{M_{\ell_j} - \sum_{i=1}^{j-1} \wt{w}_i^{\ell_j} }^{1/\ell_j}}$.

(The intuition is that at the $j$-th iteration of this step, the algorithm computes an estimate $\wt{w}_j$ of the magnitude of the $j$-th largest magnitude coordinate in the weight vector $w$.)

\item If $\sum_{i=1}^k \wt{w}_i^2 < (1-\frac{s^2-c^2}{2}) \cdot s_2$, \ignore{\rnote{What I have at the end of the proof suggests this inequality maybe should be ``$\sum_{i=1}^k \wt{w}_i^2 < \red{(1 - {\frac {3c^2} 2} + \frac{^2}{2})} \cdot s_2$.''}} output ``No,'' and otherwise output ``Yes.''
\end{enumerate}

A remark is in order regarding condition 2(c) above.  Recall that by Marcinkiewicz's theorem \cite{Marcink,bryc}, since $\bX$ is not a Gaussian distribution it must have infinitely many nonzero cumulants. (This is where we use the assumption that $\bX$ is not Gaussian; indeed if $\bX$ were Gaussian then $\tau$ as defined in the theorem statement would be zero.)  Hence a sequence of orders $\ell_1 > \cdots > \ell_k$ satisfying conditions 2(a), 2(b) and 2(c) must indeed always exist.

\medskip

To analyze the algorithm we will use the following lemma, which shows that a good estimate of $\|w\|_\ell^\ell$ yields a good estimate of $\|w\|_\infty$:

\begin{lemma}\label{lem:dist_infty_norm}
Given any vector $w$ with $\|w\|_2^2 \le 1$ and $\delta>0,$ let $\ell \geq 100/\delta^3$ be even and let $M_{\ell}$ satisfy $\big|M_{\ell} - \|w\|_{\ell}^{\ell} \big| \le (\frac{\delta}{5})^{\ell}/2$. Then $\big| M_{\ell}^{1/{\ell}}-\|w\|_{\infty} \big| \le \delta$.
\end{lemma}
We defer the proof of \Cref{lem:dist_infty_norm} to \Cref{sec:proof_infty_norm} and use it to prove \Cref{thm:general_alg}.

\begin{proofof}{\Cref{thm:general_alg}}
Without loss of generality we assume that the coordinates of $w$ satisfy $w_1 \ge w_2 \ge \cdots \ge w_n \ge 0$. We use induction to prove that $|\wt{w}_j-w_j| \le \delta_j/\ell_j$ for all $j=1,\dots,k.$

For the base case $j=1$, we have that the difference between $M_{\ell_1}$ and $\|w\|_{\ell_1}^{\ell_1}$ has magnitude at most $(\delta_1/5\ell_1)^{\ell_1}/2k$, so we can apply \Cref{lem:dist_infty_norm}.\ignore{\rnote{Why does Step 2 need $\big| M_{\ell_j} - \|w\|_{\ell_j}^{\ell_j} \big| \le (\delta_j/5\ell_j)^{\ell_j}/k$ instead of just $\big| M_{\ell_j} - \|w\|_{\ell_j}^{\ell_j} \big| \le (\delta_j/5\ell_j)^{\ell_j}$ --- at least in this $j=1$ case it seems like it doesn't need the final ``$/k$''?}} 
The value of $\tilde{w}_1$ as defined in Step~3 is $M_{\ell_1}^{1/\ell_1}$, 
so \Cref{lem:dist_infty_norm} gives that $|\wt{w}_1 - w_1| \le \delta_1/\ell_1$.

For the inductive step, we assume that the claimed bound holds for all $\wt{w}_1,\ldots,\wt{w}_{j-1}$, and we will apply \Cref{lem:dist_infty_norm} to bound the distance between $w_j$ and $\wt{w}_j$. We bound the error between $M_{\ell_j} - \sum_{i=1}^{j-1} \wt{w}_i^{\ell_j}$ and $\|w\|_{\ell_j}^{\ell_j} - \sum_{i=1}^{j-1} w_i^{\ell_j}$ by
\begin{align*}
& |M_{\ell_j} - \|w\|_{\ell_j}^{\ell_j}| + \sum_{i=1}^{j-1} |\wt{w}_i^{\ell_j} - w_i^{\ell_j}| \\
\le & (\delta_j/5 \ell_j)^{\ell_j}/2k + \sum_{i=1}^{j-1} |\wt{w}_i - w_i| \cdot \ell_j \tag{using $0 \le \wt{w}_i,w_i \le 1$}\\
\le & (\delta_j/5 \ell_j)^{\ell_j}/2k + \sum_{i=1}^{j-1} (\delta_i/\ell_i) \cdot \ell_j\\
\le & (\delta_j/5 \ell_j)^{\ell_j}/2k + \sum_{i=1}^{j-1} \delta_i \le (\delta_j/5 \ell_j)^{\ell_j}/2.
\end{align*}
In the third step of our algorithm, we use $M_{\ell_j}-\sum_{i=1}^{j-1} \wt{w}_i^{\ell_j}$ as an estimation of $\|(w_j,w_{j+1},\ldots,w_n)\|_{\ell_j}^{\ell_j}$. The above calculation shows the error of this estimation is $(\delta_j/5 \ell_j)^{\ell_j}/2$. Thus applying \Cref{lem:dist_infty_norm} to $\abs*{M_{\ell_j} - \sum_{i=1}^{j-1} \wt{w}_i^{\ell_j} }^{1/\ell_j}$ and $\|(w_j,w_{j+1},\ldots,w_n)\|_{\infty}$ with its ``$\delta$'' parameter being $\delta_j/\ell_j$, we get that $|\wt{w}_j-w_j| \le \delta_j/\ell_j$. This concludes the inductive proof. \ignore{\rnote{Can we add a bit more explanation of this last sentence? I guess the idea is that we're applying \Cref{lem:dist_infty_norm} with its ``$\delta$'' being $\delta_j/\ell_j$, but it's not totally clear to me how knowing that the error between $M_{\ell_j} - \sum_{i=1}^{j-1} \wt{w}_i^{\ell_j}$ and $\|w\|_{\ell_j}^{\ell_j} - \sum_{i=1}^{j-1} w_i^{\ell_j}$ is at most $ (\delta_j/5 \ell_j)^{\ell_j}/2$ gives us the condition that \Cref{lem:dist_infty_norm} needs, which I guess the way we are applying it is that $\big|M_{\ell_j} - \|w\|_{\ell_j}^{\ell_j} \big| \le (\frac{\delta}{5\ell_j})^{\ell_j}/2$.} }

With this upper bound on each $|\wt{w}_j-w_j|$ in hand, we can infer that\ignore{\rnote{The \red{2} below seems to follow from some absolute scaling that is being assumed on $w$, i.e that its 2-norm is 1 or at least that every entry is at most 1 or so}
}
\begin{equation} \label{eq:useful-bound}
\sum_{i=1}^k | \wt{w}_i^2 - w_i^2| = \sum_{i=1}^k |\wt{w}_i - w_i| \cdot |\wt{w}_i + w_i| \leq \sum_{i=1}^k (\delta_i/\ell_i) \cdot 3 \le \eps/4,
\end{equation}
where the first inequality uses $\|w\|_2 \leq 1$ and the closeness of each $\wt{w}_i$ to $w_i$ to upper bound $|\wt{w}_i + w_i| \leq 3.$

We now use \Cref{eq:useful-bound} to establish correctness of our algorithm.
We first consider the ``yes'' case in which $\dist(w,k\text{-sparse}) \leq c$. In this case, we have that \ignore{\rnote{The following inequality had been $\sum_{i=1}^k w_i^2 - \eps/4 = (c^2) \cdot \|w\|_2^2 -\eps/4$, but (a) I think it is to be an inequality, right? and (b) I think by the definition of being $c$-close to sparse it should be $1-c^2$ instead of $c^2$. Let me know if I'm wrong about this?}}
 $\sum_{i=1}^k \wt{w}_i^2 \ge \sum_{i=1}^k w_i^2 - \eps/4 \ge (1-c^2) \cdot \|w\|_2^2 -\eps/4$. Since $s_2 = \|w\|_2^2 \pm \eps/4$, we have \[
\sum_{i=1}^k |\wt{w}_i|^2 \ge (1-c^2) (s_2 - \eps/4) - \eps/4 \ge (1-c^2) \cdot s_2 - \eps/2 
\]
Furthermore, since $\|w\|_2 \in [1/C^2,1]$ shows $\|w\|_2^2 \in [1/C^4, 1]$, given $\eps=\frac{s^2 - c^2}{2C^4}$, we have 
\[
s_2 \ge \|w\|_2^2 - \eps/4 \ge 1/C^4 - 1/(8C^4) \ge 2/(3C^4) > \eps/(s^2-c^2)
\] and we can simplify our lower bound on $\sum_{i=1}^k \wt{w}_i^2$ to
\[
(1-c^2) s_2 - \eps/2 > (1-c^2) s_2 - \frac{s^2-c^2}{2} \cdot s_2 = 
\parens*{1 - {\frac {s^2 - c^2} 2} } \cdot s_2,
\]
from which we see that the algorithm is correct in the ``yes''-case.

Similarly, in the ``NO" case, we have $\sum_{i=1}^k \wt{w}_i^2 < \parens*{1 - {\frac {s^2 - c^2} 2} } \cdot s_2$. \ignore{\rnote{I guess this red now needs to be $\red{
\parens*{1 - {\frac {3c^2} 2} +  \frac{s^2}{2}} \cdot s_2}$.  Shall we write down a verification of this part? Or should we just say that the argument is  similar?}}
 This proves the assertions made in the two bulleted statements of the theorem.

Finally, when $\bX$ is supported in $[-\bdd,\bdd]$, we apply \Cref{thm:distance_non_zero_cumulants} to upper bound $\ell_i$: given any $\ell_{i+1}$ and $\delta_{i+1}$, for $t=100 k^3/(\delta_{i+1}/5 \ell_{i+1})^{3 \ell_{i+1}}$, there always exists $\ell_i \in [t, (4\bdd)^{O(t)}]$ with $\kappa_{\ell_i}(X) \ge 2^{-(4\bdd)^{O(t)}}$. Thus $\tau$ is also lower bounded by $2^{-(4\bdd)^{O(\ell_1)}}$.

\end{proofof}

\subsection{Proof of \Cref{lem:dist_infty_norm}}\label{sec:proof_infty_norm}

For convenience we assume throughout this subsection that $w_1 \ge w_2 \ge \cdots \ge w_n \geq 0$ in the vector $w$.

\begin{fact}\label{fact:approx_ell_infty}
If $\|w\|_2 \le 1$, then $\|w\|_\ell^\ell$ is always between $w_1^\ell$ and $w_1^{\ell-2}$ for any $\ell \ge 3$.
\end{fact}
\begin{proof}
$w_1^\ell \le \|w\|_\ell^\ell=\sum_{i=1}^n w_i^\ell \leq w_1^{\ell-2} \sum_{i=1}^n w_i^2 \le w_1^{\ell-2}$.
\end{proof}

\begin{proofof}{\Cref{lem:dist_infty_norm}}
Recall that by assumption we have $w_1=\|w\|_{\infty} \le 1$. Let $\taua$ denote $M_\ell^{1/\ell}$ and $\Delta \le (\delta/5)^{\ell}/2$ denote the error such that $M_\ell = \|w\|_{\ell}^{\ell} \pm \Delta$. We consider two cases based on the size of $w_1$:

\begin{enumerate}
\item The first case is that $w_1\le \delta/5$. In this case we upper bound $\taua$ by
\begin{align*}
& (\|w\|_\ell^\ell + \Delta)^{1/\ell} \\
\le & (w_1^{\ell-2}+\Delta)^{1/\ell}
\tag{using the upper bound from \Cref{fact:approx_ell_infty} on $\|w\|_\ell^\ell$ } \\
\le & ((\delta/5)^{\ell-2}+(\delta/5)^{\ell}/2)^{1/\ell} \\
\le & 2^{1/\ell} \cdot (\delta/5)^{(\ell-2)/\ell} \\
\le & 2^{1/\ell} \cdot (5/\delta)^{2/\ell} \cdot \delta/5 \tag{using the fact $\ell=100/\delta^3$}\\
\le & 2\delta/5.
\end{align*}
So we have that $\big| \taua - w_1 \big| \leq 2 \delta/5 + w_1$, which is at most $3 \delta/5$ by the assumption of $w_1$ and the Lemma.

\item The second case is that $w_1> \delta/5$. In this case we first bound $w_1-\taua$ by
\begin{align*}
w_1-(\|w\|_\ell^\ell - \Delta)^{1/\ell} & \le w_1 - (w_1^{\ell} - \Delta)^{1/\ell} \tag{using the lower bound from \Cref{fact:approx_ell_infty} on $\|w\|_\ell^\ell$}\\
& = w_1 - w_1 ( 1 - \frac{\Delta}{w_1^\ell})^{1/\ell} \\
& \le w_1 - w_1 (1 - 2 \frac{\Delta}{\ell \cdot w_1^\ell}) \tag{using $(1-x)^{1/\ell} \ge 1-2x/\ell$ when $x \le 1/2$} \\
& = 2 w_1 \cdot \frac{\Delta}{\ell \cdot w_1^\ell}.
\end{align*}
Then we bound $\taua-w_1$ by
\begin{align*}
(\|w\|_\ell^\ell + \Delta)^{1/\ell} - w_1 & \le (w_1^{\ell-2} + \Delta)^{1/\ell} - w_1 \tag{using the upper bound from \Cref{fact:approx_ell_infty} on $\|w\|_\ell^\ell$}\\
& \le (w_1^{\ell-2} + \Delta)^{1/\ell} - (w_1^{\ell}+\Delta)^{1/\ell} + (w_1^{\ell}+\Delta)^{1/\ell} - w_1 \\
& \le (w_1^\ell + \Delta)^{1/\ell} \cdot \left( \parens*{\frac{w_1^{\ell-2}+\Delta}{w_1^{\ell}+\Delta}}^{1/\ell} - 1 \right) + w_1 \parens*{1+\frac{\Delta}{w_1^\ell}}^{1/\ell} - w_1\\
& \le (w_1^\ell+\Delta)^{1/\ell} \cdot \left( \parens*{1+\frac{w_1^{\ell-2}(1-w_1^2)}{w_1^\ell+\Delta}}^{1/\ell} - 1 \right) + w_1 \parens*{1 + \frac{\Delta}{\ell \cdot w_1^\ell}} - w_1 \tag{using $(1+x)^{1/\ell} \le 1+x/\ell$} \\
& \le (w_1^\ell + \Delta)^{1/\ell} \cdot \frac{w_1^{\ell-2}}{\ell(w_1^\ell + \Delta)} + w_1 \frac{\Delta}{\ell \cdot w_1^\ell}.
\tag{using $(1+x)^{1/\ell} \le 1+x/\ell$ again} 
\end{align*}
We combine the above two bounds to get that
\[
|\taua-w_1| \le \ignore{w_1 \cdot \frac{\Delta}{\ell \cdot w_1^\ell}+} (w_1^\ell + \Delta)^{1/\ell} \cdot \frac{w_1^{\ell-2}}{\ell(w_1^\ell + \Delta)} + 3 w_1 \frac{\Delta}{\ell \cdot w_1^\ell}.
\]
Plugging in our bounds on $w_1$ and $\Delta \le (\delta/5)^{\ell}/2$ into this inequality, this is at most
\begin{align*}
(w_1^{\ell} + \Delta)^{1/\ell} \cdot \frac{1}{\ell \cdot w_1^2} + 3 \frac{\Delta}{\ell \cdot w_1^{\ell-1}} & \le (w_1 + \Delta^{1/\ell}) \cdot \frac{1}{\ell \cdot w_1^2} + 3 \frac{\Delta}{\ell w_1^{\ell-1}} \tag{using $(x+y)^{1/\ell} \le x^{1/\ell}+y^{1/\ell}$}\\
& \le \frac{1}{\ell \cdot \delta/5} + \frac{\delta/5}{\ell \cdot (\delta/5)^2} + \frac{3 \cdot (\delta/5)^{\ell}/2}{\ell \cdot (\delta/5)^{\ell-1}} \tag{since in this case $w_1 \ge \delta/5$} \\
& \le \frac{\delta}{2}. \tag{using $\ell \geq 100/\delta^3$}
\end{align*} 
\qedhere
\end{enumerate}
\end{proofof}

%% file: gap_nonzero_cumulant.tex

\section{Bounding the gap between non-zero cumulants} \label{sec:gap}


\bigskip

The result of Marcinkiewicz (\Cref{thm:Marcink_theorem}) shows that any non-Gaussian random variable $\bX$  has an infinite number of non-zero cumulants. However, this result is not constructive and leaves open two obvious questions: 
\begin{enumerate}
\item Suppose $\kappa_\ell(\bX) \not =0$. What can we say about $\arg\min_{\ell'>\ell} \kappa_{\ell'}(\bX) \not =0$? In other words, how many consecutive zero cumulants can $\bX$ have following the non-zero cumulant $\kappa_{\ell}(\bX)$? 
\item Merely having a non-zero cumulant $\kappa_{\ell'}(\bX)$ is not sufficient for us; since our results depend on the magnitude of the non-zero cumulants, we would also like a lower bound on the magnitude of $\kappa_{\ell'}(\bX)$ (where $\ell'$ is as defined above). Can we get such a lower bound on $\kappa_{\ell'}(\bX)$? 
\end{enumerate}
The main result of this section is to give an \emph{effective} answer to both these questions when the random variable $\bX$ has bounded support. To the best of our knowledge (and based on conversations with experts~\cite{Neeman-comm, Bryc-comm, Janson-comm}), previously no such effective bound was known for gaps between non-zero cumulants. 

Before stating our result, we note that for any real random variable $\bX$ the random variable $\bY = \bX - \bX'$ (where $\bX'$ is an independent copy of $\bX$) is (i) symmetric and (ii) has $\kappa_{\ell}(\bY) = (1+(-1)^{\ell}) \kappa_{\ell}(\bX)$. Thus for the purposes of this section, it suffices to restrict our attention to symmetric random variables and even-numbered cumulants.


\begin{theorem}\label{thm:distance_non_zero_cumulants}
Given any $\ell$ and any symmetric random variable $\bX$ with unit variance and support $[-\bdd, \bdd ]$, {where $\bdd \geq 1$,} there exists $\ell'=(\Theta(1)\cdot \bdd^4 \log \bdd)^{\ell}$ such that $\big| \kappa_j(\bX) \big| \ge 2^{-\ell'}$ for some $j \in (\ell,\ell']$.
\end{theorem}

Before delving into the formal proof of this theorem, we give a high-level overview. Recall that the cumulant generating function and moment generating function of $\bX$ are defined respectively as
\[
K_{\bX}(z) = \sum_{j \ge 1} \frac{\kappa_j(\bX)}{j!} z^j; \ \ M_{\bX}(z) = \E[e^{z\bX}].
\]
The first main ingredient (Claim~\ref{clm:zero_roots}) is that  the function $M_{\bX}(z)$ has a root in the complex disc of radius $O(\bdd^3)$ centered at the origin. The proof of this is somewhat involved and uses a range of ingredients such as bounding the number of zeros of entire functions and the Hadamard factorization theorem.

Now, suppose it were the case that $|\kappa_{j}(\bX)| \leq 2^{-\ell'}$ for all $j \in (\ell,\ell']$ for a sufficiently large $\ell'$. We consider the ``truncated" function $P_{\ell}(z)$  
$$
P_{\ell}(z) = \sum_{j=1}^{\ell} \frac{\kappa_j(\bX)}{j!} z^j. 
$$
Observe that while $K_{\bX}(z)$ is not necessarily well defined everywhere, (i) it is easy to show that it is well defined in the open disc of radius $1/(e\bdd)$ (call this set $\mathcal{B}$); (ii) the function $P_{\ell}(z)$ is an entire function. Further, 
since $\kappa_{j}(\bX)$ is assumed to have very small magnitude for all $j \in (\ell,\ell']$, it is not difficult to show that  $P_{\ell}(z)$ and $K_{\bX}(z)$ are close to each other
in $\mathcal{B}$. 
 Using 
$e^{K_{\bX}(z)} = M_{\bX}(z)$ in $\mathcal{B}$ (since both are well-defined), we infer that $e^{P_{\ell}(z)}$ and $M_{\bX}(z)$ are also close to each other in $\mathcal{B}$. In other words, the function $h(z) := e^{P_{\ell}(z) - M_{\bX}(z)}$ is close to zero in $\mathcal{B}$.  

Finally, we observe that $e^{P_{\ell}(z)}$ has no zeros in $\mathbb{C}$ and in fact, we can show that it has relatively large magnitude within a ball of radius $O(\bdd^3)$. Using the first ingredient that $M_{\bX}(z)$ has a zero in this disc, we derive that the maximum of $|h(z)|$ is large in a disc of radius $O(\bdd^3)$. However, since $h$ is an entire function, once $\ell'$ is sufficiently large, this contradicts the fact that $h(z)$ is close to zero in $\mathcal{B}$ (this uses Hadamard's three circle theorem). This 
finishes the proof.

\begin{proofof}{Theorem~\ref{thm:distance_non_zero_cumulants}}
Towards a contradiction, fix $\zeta = 2^{-\ell'}$ and let us assume that $|\kappa_{j}(\bX)| < \zeta$ for $j \in (\ell, \ell']$. Let us consider the moment generating function $M_{\bX}: \mathbb{C} \rightarrow \mathbb{C}$ defined 
by $
M_{\bX}(z) = \mathbf{E}[e^{z \bX}]. 
$ From the fact that the random variable $\bX$ is bounded in $[-\bdd,\bdd ]$, it follows that the function $M_{\bX}$ is an entire function (i.e., holomorphic over all of $\mathbb{C}$). Next, consider the cumulant generating function $K_{\bX}: \mathbb{C} \rightarrow \mathbb{C}$ defined as
\[
K_{\bX}(z) = \sum_{j \ge 1} \frac{\kappa_j(\bX)}{j!} z^j. 
\]
From \Cref{clm:up_bound_cumulant}, we know that $|\kappa_{j}(\bX)| \le \bdd^j \cdot e^{j} \cdot j!$. Define the open disc $\mathcal{B} = \{z: |z| < 1/(e \bdd )\}$ and observe that the right hand side series is absolutely convergent in $\mathcal{B}$ and hence $K_{\bX}$ is holomorphic in $\mathcal{B}$. We recall from the definition of cumulants that for  $z \in \mathcal{B}$, 
$
e^{K_{\bX}(z)}  = M_{\bX}(z). 
$ 

We will need the following claim about the roots of $M_{\bX}$:  
\begin{claim}\label{clm:zero_roots}
For any symmetric random variable $\bX$ with unit variance and support $[-\bdd,\bdd]$ where $\bdd \geq 1$, there exists $z_0$ with $|z_0| \le 200 \bdd^3$ such that $M_{\bX}(z_0) = \E[e^{z_0 X}]=0$.
\end{claim}
We defer the proof of \Cref{clm:zero_roots} to \Cref{sec:proof_zero_roots}. 
 Let us define $P_{\ell}(z)$ to be the polynomial obtained by truncating the cumulant generating function Taylor series expansion to degree $\ell$, so 
 $P_{\ell}(z)= \sum_{1 \le j \le \ell} \frac{\kappa_j(\bX)}{j!} z^j.$
 We now define the function $g: \mathbb{C} \rightarrow \mathbb{C}$ as
\begin{equation}~\label{eq:def-g}
g(z) = e^{P_{\ell}(z)} - \mathbf{E}[e^{z\bX}]. 
\end{equation}
Observe that $g$ is an entire function. The following claim lower bounds the magnitude of $g$ on the point $z_0$ defined above:
\begin{claim}\label{clm:lower_bound_mag}
Let $z_0$ be the complex number satisfying $\E[e^{z_0 \bX}]=0$ in \Cref{clm:zero_roots}. Then we have
$$ |g(z_0)| \ge e^{-2 \cdot (200 e)^\ell \cdot \bdd^{4\ell}}.$$
\end{claim}
\begin{proof}
We have 
\[
|P_{\ell}(z_0)| \le \sum_{j=1}^\ell \frac{|\kappa_j(\bX)|}{j!} |z_0|^j \le \sum_{j=1}^\ell e^j \cdot \bdd^j \cdot |z_0|^j \le 2 (e\bdd)^\ell \cdot (200 \bdd^3)^\ell.
\]
The first inequality is just a triangle inequality whereas the second inequality uses \Cref{clm:up_bound_cumulant}. Since $\E[e^{z_0 X}]=0$, we get that
$
|g(z_0)| \ge e^{-|P_{\ell} (z_0)|} \ge e^{-2 \cdot (200 e)^\ell \cdot \bdd^{4\ell}}. 
$
\end{proof}
We now recall the Hadamard three-circle theorem. 
\begin{theorem}[Hadamard three-circle theorem]
Let $0<r_1<r_2<r_3$ and let $h$ be an analytic function on the annulus $\{z \in \C: |z| \in [r_1,r_3]\}$. Let $M_{h}(r)$ denote the maximum of $h(z)$ on the circle $|z|=r$. Then,
$$
\ln \frac{r_3}{r_1} \ln M_h(r_2) \le \ln \frac{r_3}{r_2} \ln M_h(r_1) + \ln \frac{r_2}{r_1} \ln M_h(r_3).
$$
\end{theorem}
We are now ready to finish the proof of \Cref{thm:distance_non_zero_cumulants}.  The proof uses the following claim:

\begin{claim} \label{claim:zast}
There is a point $z_\ast$ satisfying
\begin{eqnarray}~\label{eq:lb-g}
|z_{\ast} | \le \frac{1}{2 e \bdd} \ \textrm{ and } 
|g(z_{\ast})| \ge e^{- 12 \left((400 e)^{\ell} \cdot \bdd^{4\ell} \cdot \ln(400 e \bdd) \right)}. 
\end{eqnarray}
\end{claim}
\begin{proof}
Recall from \Cref{clm:zero_roots} that the point $z_0$ satisfies $|z_0| \le 200 \bdd^3$ and $M_{\bX}(z_0) = \E[e^{z_0 X}]=0$.  There are two cases:

\begin{enumerate}

\item  If $|z_0| \le \frac{1}{2e \bdd}$: In this case, we set $z_\ast=z_0$. By definition, it satisfies the first condition in \eqref{eq:lb-g} and using \Cref{clm:lower_bound_mag}, it satisfies the second condition. 

\item  If $|z_0| > \frac{1}{2e \bdd}$: Set $r_1 = \frac{1}{2e\bdd}$, $r_2 = |z_0|$ and $r_3= er_2$.
For $g$ as defined in \eqref{eq:def-g}, from \Cref{clm:lower_bound_mag}, we have 
\begin{equation}~\label{eq:lbr2}
M_g(r_2)  \ge |g(z_0) | = |e^{P_{\ell}(z_0)} - \mathbf{E}[e^{z_0 \bX}]|
\geq
e^{-2 \cdot (200 e)^\ell \cdot \bdd^{4\ell}}. \end{equation}
On the other hand, consider any point $z_3$ such that $|z_3| = r_3$. We have that 
\begin{eqnarray*}
|g(z_3)|&\le&  |e^{P_\ell(z_3)}|+ |\E[e^{z_3 \bX}]|  
\le e^{\sum_{j=1}^\ell (e\bdd)^j \cdot |z_3|^j} + \int_{-\bdd}^{\bdd} \Pr[\bX=x] \cdot e^{|z_3| \cdot x} \mathrm{d} x  \\ 
&\le& e^{\sum_{j=1}^\ell (e\bdd)^j \cdot r_3^j} + \int_{-\bdd}^{\bdd} \Pr[\bX=x] \cdot e^{r_3 \cdot x} \mathrm{d} x \le e^{2 (e\bdd)^{\ell} \cdot r_3^{\ell}},
\end{eqnarray*}
and hence
\begin{eqnarray}~\label{eq:ubr3}
|M_g(r_3)| \le e^{2 (e\bdd)^{\ell} \cdot r_3^{\ell}}. 
\end{eqnarray}
Now observe that the function $g$ defined in \eqref{eq:def-g} is an entire function. Consequently, 
using $r_3/r_2=e$,  we can apply the Hadamard three circle Theorem to $g$ to obtain 
\begin{align*}
\ln M_g(r_1) &\ge \ln \frac{r_3}{r_1} \ln M_g(r_2) - \ln \frac{r_2}{r_1} \ln M_g(r_3)  \\ 
&\ge  - \bigg(2 \cdot (200 e)^\ell \cdot \bdd^{4\ell} \bigg) \ln \frac{r_3}{r_1} - \bigg( 2 (e\bdd)^\ell \cdot (er_2)^\ell \bigg) \ln \frac{r_2}{r_1} \tag{applying \eqref{eq:lbr2}, \eqref{eq:ubr3}}\\
&\ge -\bigg( 3 (400 e)^\ell \cdot \bdd^{4\ell} \ln (400e \bdd^4) \bigg).
\end{align*}
The last inequality uses that $r_2 \le 200 \bdd^3$ (from \Cref{clm:zero_roots}). This implies the existence of a point $z_{\ast}$ satisfying  \eqref{eq:lb-g} and concludes the proof of \Cref{claim:zast}. \qedhere
\end{enumerate}
\end{proof}

Continuing with the proof of \Cref{thm:distance_non_zero_cumulants}, observe that the Taylor expansion for $K_{\bX}(z)$ (at $z=0$) converges  absolutely in $\mathcal{B}$. Thus $K_{\bX}(z)$ is holomorphic in $\mathcal{B}$ and is given by its Taylor expansion.
 Since $z_{\ast} \in \mathcal{B}$, recalling our initial assumption that the $(\ell+1)$-th through $\ell'$-th cumulants all have magnitude at most $\zeta$, we have that 
\begin{align}
|K_{\bX}(z_{\ast}) - P_{\ell}(z_{\ast})| &\le \sum_{j=\ell+1}^{\ell'} \frac{\zeta \cdot |z_{\ast}|^{j}}{j!}  + \sum_{j>\ell'} \frac{|\kappa_j(\bX)| \cdot |z_{\ast}|^{j}}{j!} \nonumber \\
&\le \sum_{j=\ell+1}^{\ell'} \frac{\zeta}{j! \cdot (2 e \bdd)^j}  + \sum_{j>\ell'} \frac{\bdd^j \cdot e^j \cdot j! }{(2e\bdd)^j \cdot j!}  \tag{using  \Cref{clm:up_bound_cumulant} and \eqref{eq:lb-g} }\nonumber  \\
&\le \frac{2\zeta}{\ell! \cdot (2e\bdd)^{\ell}}+ 2^{-\ell'} \le 2^{-\ell' +1}.\label{eq:upperbound-diff}
\end{align}
Since $K_{\bX}(z)$ is holomorphic in $\mathcal{B}$, so is $M_{\bX}(z) = e^{K_{\bX}(z)}$, and we have that
\begin{align}
\big| M_{\bX}(z_{\ast}) - e^{P_{\ell}(z_\ast)} \big| &= \big| e^{K_{\bX}(z_\ast)}- e^{P_{\ell}(z_\ast)} \big| =\big |e^{P_\ell(z_\ast)} \big| \cdot \big| e^{K_{\bX}(z_\ast) - P_\ell(z_\ast) }-1 \big|  \nonumber \\
&\le  |e^{P_\ell(z_\ast)} \big| \cdot 2^{-\ell' +2},  \label{eq:upperbound-diff2}
\end{align}
where the last inequality is by \Cref{eq:upperbound-diff}.
However, applying \Cref{clm:up_bound_cumulant} and recalling that $|z_\ast| \le 1/(2e\bdd)$, we also have  
\[
\big| P_\ell(z_\ast) \big| \le  \sum_{j=1}^{\ell}\frac{|z_\ast|^j\kappa_j(\bX)}{j!} \le \sum_{j=1}^{\ell}\frac{|z_\ast|^j \cdot \bdd^j \cdot e^{j}\cdot  j!}{j!} \le 1.
\]
Plugging this back into \eqref{eq:upperbound-diff2}, we get 
\[
\big| M_{\bX}(z_{\ast}) - e^{P_{\ell}(z_\ast)} \big| \leq 4e \cdot 2^{-\ell'}. 
\]
However, recalling that $g(z) = M_{\bX}(z) - e^{P_{\ell}(z)}$, this contradicts \eqref{eq:lb-g} with room to spare provided that, say,
$
\ell' >  50 \left((400 e)^{\ell} \cdot \bdd^{4\ell} \cdot \ln(400 e \bdd) \right).
$ 
This finishes the proof of \Cref{thm:distance_non_zero_cumulants}
\end{proofof}

\subsection{Proof of \Cref{clm:zero_roots}}\label{sec:proof_zero_roots}
We start by showing that the function $M_{\bX}(z)$ must necessarily decay along the line $\{z: \mathsf{Re}(z)=0\}$ close to the origin. 
\begin{claim}\label{clm:decrease_value}
For the symmetric random variable $\bX$, which has unit variance and is supported on $[-\bdd,\bdd]$ {where $\bdd \geq 1$}, there exists $\alpha_\ast \in \mathbb{R}$ such that $|\alpha_\ast| \le 3$,  
$M_{\bX}(i\alpha_\ast) \in \mathbb{R}$ and $|M_{\bX}(i\alpha_\ast)| \le 1- \frac{1}{2\bdd^2}$. 
\end{claim}
\begin{proof}
First of all, note that  by symmetry of $\bX$, $M_{\bX}(i\alpha) = \E [e^{i \alpha \bX}]$ is necessarily real-valued for any $\alpha \in \mathbb{R}$. Next, choose $F>1$ (we will fix its exact value soon). We have
\begin{eqnarray}
\int_{\alpha=-F}^{F} M_{\bX}(2\pi i \alpha) d\alpha &=& \int_{\alpha=-F}^{F} \int_{x=-\bdd}^{\bdd} \Pr[\bX=x] \cdot e^{2\pi i \alpha x} dx d\alpha \nonumber \\ 
&=& \int_{x=-\bdd}^{\bdd} \Pr[\bX=x] \int_{\alpha=-F}^{F} e^{2\pi i \alpha x} d \alpha dx  \nonumber \\
&\leq& \int_{|x| \le 1/2} 2\Pr[\bX=x] \cdot F dx + \int_{|x| > 1/2} \Pr[\bX=x] \cdot \frac{\sin(2\pi F x)}{\pi x} dx. \label{eq:integral-moment-around-0} 
\end{eqnarray}
Now, observe that 
\[
1=\E[|\bX|^2] \le \Pr[|\bX| > 1/2] \cdot \bdd^2 + (1-\Pr[|\bX| > 1/2]) \cdot \frac{1}{4}.
\]
Thus, we obtain  
\begin{equation}~\label{eq:Markov2}
\Pr[|\bX| > 1/2] \ge \frac{3}{4\big(\bdd^2 - \frac{1}{4} \big)} \ge \frac{3}{4\bdd^2}. 
\end{equation}
Likewise, observe that $\sin(2\pi Fx)/(\pi x)$ always has magnitude at most $2/\pi$ for $|x| > 1/2$. Plugging \eqref{eq:Markov2} and this back into \eqref{eq:integral-moment-around-0}, and using $F > 1$, we have that
\begin{eqnarray*}
\int_{\alpha=-F}^{F} M_{\bX}(2\pi i \alpha) d\alpha \leq \bigg( 1- \frac{3}{4\bdd^2}\bigg) \cdot 2F + \frac{3}{4\bdd^2} \cdot \frac{2}{\pi}. 
\end{eqnarray*}
This implies that there is a point $\alpha_{\ast} \in [-F,F]$ such that 
$$
M_{\bX}(2\pi i \alpha_{\ast}) \le 1 -\frac{3}{4\bdd^2} + \frac{3}{4\bdd^2 \pi F}. 
$$
Plugging in $F=3$, we get the claim.

\ignore{
Next, observe that $M_{\bX}(z)$ is an entire function  and that $\log |M_{\bX}(z)| \le C \cdot |z|$. Thus, the order of  $M_{\bX}(z)$ is $1$. The next lemma bounds the number of zeros of $M_{\bX}(z)$ in a ball of bounded radius. 
}

\end{proof}
Observe that $M_{\bX}(z)$ is an entire function and thus is well defined on all of $\mathbb{C}$. 
The next lemma bounds the number of zeros of $M_{\bX}(z)$ in a ball of radius $R$. This is essentially the same  as  the first part of Theorem~2.1 in \cite{SteinShakarchi}, though the bound given there is asymptotic whereas we need a precise quantitative bound. 

\begin{claim}\label{clm:bound_zeros}
For $M_{\bX}(z)$ as defined above and $r>0$, let $n(R)$ denote the number of zeros of $M_{\bX}(z)$ contained in the ball $\{|z|: |z| \le R\}$ (counting multiplicities). Then $n(R) \le e \bdd R$. 
\end{claim}
Before proceeding with the proof of \Cref{clm:bound_zeros}, we recall a useful ingredient, namely, Jensen's formula (see~Theorem 1.1, Section 5 in \cite{SteinShakarchi}):

\begin{theorem} [Jensen's formula] \label{thm:Jensen-formula}
Let $h$ be an analytic function in a region of $\mathbb{C}$ which contains the closed disc $\mathbf{D} = \{z: |z| \le R\}$. Suppose $h(0) \not =0$ and $h$ does not have zeros on the boundary $\partial \mathbf{D} = \{z: |z| = R\}$.
 Then 
\[
\int_0^1 \ln \left|h (R e^{2 \pi \bi t}) \right| \mathrm{d} t = \ln |h(0)| + \sum_{z: |z| < R, \  h(z)=0} \ln \frac{R}{|z|},
\]
where the summation on the right hand side counts the roots of $h$ with multiplicity.
\end{theorem}
\begin{proofof}{\Cref{clm:bound_zeros}}
First since $M_{\bX}(z)$ is an entire function,
its zeros are isolated. Thus, by perturbing $R$ infinitesimally, we can assume that $M_{\bX}(z)$ has no zeros on $\partial \mathbf{D}$. Further, an immediate consequence of the Jensen's formula is that the number of zeros of an analytic function in $\mathbf{D}$ must be finite. 
To see this, let $R'>R$ and apply 
Jensen's formula on the cicle of radius $R'$. By Jensen's formula, it follows that $\sum_{z: |z| < R', \  h(z)=0} \ln \frac{R'}{|z|}$ is finite which implies that the number of zeros in ${\mathbf{D}}$ has to be finite. 
 For any radius $R_{\ast}$,
let us now enumerate the zeros of $M_{\bX}(z)$ that lie within the disc $\{z \in \C: |z| \leq R_\ast\}$ as $z_1, \ldots, z_{n(R_\ast)}$ such that $|z_1| \le |z_2| \ldots \le |z_{n(R_\ast)}|$. Then, 
\begin{equation}	~\label{eq:zeros-integral}			
\sum_{i=1}^{n(R_\ast)}\ln \frac{R_\ast}{|z_i|} = \sum_{i=1}^{n(R_\ast)-1} i \cdot \ln \frac{|z_{i+1}|}{|z_i|} + n(R_\ast) \cdot \ln \frac{R_\ast}{|z_{n(R_\ast)}|} = \int_{0}^{R_\ast} n(r) \frac{\mathrm{d} r}{r}.
\end{equation}
The last equality simply follows by observing that since $n(r)$ is finite in $[0,R_\ast]$, hence we can split the integral on the right hand side at the points of discontinuity of $n(r)$. Next, we have 
\begin{equation}~\label{eq:zeros-integral2}
n(R)  \le n(R) \int_{R}^{eR} \frac{\mathrm{d} r}{r} \le \int_R^{eR} n(r) \frac{\mathrm{d} r}{r}\le \int_0^{eR} n(r) \frac{\mathrm{d} r}{r}  =  \sum_{i=1}^{n(eR)}\ln \frac{eR}{|z_i|}. 
\end{equation}
In the above, the first three inequalities follow by definition while the last equality is an application of
\eqref{eq:zeros-integral} with $R_\ast= eR$. Finally, 
by definition, $M_{\bX}(0)=1$ and $\ln |M_{\bX}(z)| \le \bdd|z|$. Using these two facts with \eqref{eq:zeros-integral2} and \Cref{thm:Jensen-formula}, we get 
\[
n(R)   \le \sum_{i=1}^{n(eR)}\ln \frac{eR}{|z_i|} = \int_{0}^1 \ln \left|M_{\bX} (eR e^{2 \pi \bi t}) \right| \mathrm{d} t \le e\bdd R. 
\]
This finishes the proof of \Cref{clm:bound_zeros}. 
\end{proofof}
\begin{corollary}\label{cor:sum_roots}
Let $M_{\bX}(z) = \E [e^{z \bX}]$ as defined earlier, and let $\alpha>1$, $R_{\ast}>0$ be such that $M_{\bX}(z)$ has no roots in the ball $\{z: |z| \le R_{\ast}\}$. Then, 
\[
\sum_{z: M_{\bX}(z)=0} \frac{1}{|z|^{\alpha}} \le \frac{\alpha (\alpha-1) \cdot e\bdd}{R_\ast^{\alpha-1}}.
\]
\end{corollary}
\begin{proof}
It follows from \Cref{clm:bound_zeros} that the number of roots of $M_{\bX}(z)$ is countable. Let us enumerate these roots as $z_1, z_2, \ldots $.  We have that
\[
\sum_{z: M_{\bX}(z)=0} \frac{1}{|z|^{\alpha}} = \sum_{i} \frac{1}{|z_i|^{\alpha}} 
= \sum_i \alpha \int_{|z_i|}^{\infty} \frac{1}{r^{1+\alpha}}\mathrm{d} r
= \alpha \int_0^{\infty} \frac{n(r)}{r^{1+\alpha}} \mathrm{d} r.
\]
From the assumption $n(r)=0$ for all $r \le R_\ast$, we can use \Cref{clm:bound_zeros} to upper bound the right hand side as 
\[
\alpha \int_{R_\ast}^{\infty} \frac{e\bdd r}{r^{1+\alpha}} \mathrm{d} r = \alpha \cdot e\bdd \int_{R_\ast}^{\infty} \frac{1}{r^{\alpha}} \mathrm{d} r \le \frac{\alpha (\alpha-1) e \bdd}{R_\ast^{\alpha-1}}. \qedhere
\]
\end{proof}

The last ingredient we will need to prove \Cref{clm:zero_roots} is the Hadamard factorization theorem (see Theorem 5.1, Section 5 in \cite{SteinShakarchi}):
 
\begin{theorem}~\label{thm:Hadamard}
Let $h$ be an entire function that is of order $1$ (i.e., $\log |h(z)| = O(|z|)$). If $h(0) \not =0$, then there exist $A, A' \in \C$ such that
\[
h(z) = e^{Az + A'} \prod_{n\ge 1} \bigg(1 -\frac{z}{z_n}\bigg) e^{\frac{z}{z_n}},  
\]
where $z_1, z_2, \ldots$ are the roots of $h(z)$. 
\end{theorem}
\begin{proofof}{\Cref{clm:zero_roots}}
As $M_{\bX}(z)$ is an entire function of order one, we can use \Cref{thm:Hadamard} to express it as 
\[
M_{\bX}(z) = e^{Az + A'} \prod_{n\ge 1} \bigg(1 -\frac{z}{z_n}\bigg) e^{\frac{z}{z_n}},  
\]
where $z_1, z_2, \ldots$ are the roots of $M_{\bX}(z)$. We first recall that $M_{\bX}(0)=1$, and hence $A'=0$. We next observe that since $M_{\bX}(z)$ is a symmetric function, if $z_n$ is a root then so is $-z_n$ (and with the same multiplicity). Together with the symmetry of $M_{\bx}(z)$, this implies that the coefficient $A$ of $z$ appearing in the exponent is also zero.  
Next, we observe that $M_{\bX}(z)$ cannot have 
any root on the real line. Thus, if we define $\Omega_1= \{z : \mathsf{Re}(z)>0\}$, then the right hand side of the above equation simplifies to \[
M_{\bX}(z) = \prod_{z_i \in \Omega_1: M_{\bX}(z_i)=0} \bigg(1 -\frac{z^2}{z_i^2}\bigg).  
\]

Now, suppose that $M_{\bX}(z)$ does not have any zeros in a ball of radius $R_{\ast}$ around the origin. Then, for any $z$ such that $|z| \le R_{\ast}$, the above gives that 
\[
|M_{\bX}(z)| \ge   \prod_{z_i \in \Omega_1: M_{\bX}(z_i)=0} \bigg(1 -\frac{|z|^2}{|z_i|^2}\bigg) \ge 1 - \sum_{z_i \in \Omega_1: M_{\bX}(z_i)=0} \frac{|z|^2}{|z_i|^2}. 
\]
Applying \Cref{cor:sum_roots} (with $\alpha=2$),  we have that
\[
|M_{\bX}(z)| \ge 1 - \frac{2e\bdd|z|^2}{R_{\ast}}. 
\]
Choosing $R_{\ast} = 72e\bdd^3$, we get that $|M_{\bX}(z)| \ge 1- |z|^2/36\bdd^2$ for all $|z| \le 72e\bdd^3$. In particular, for all $|z| \le 3$, $M_{\bX}(z) \ge 1-1/(4\bdd^2)$. This contradicts \Cref{clm:decrease_value}. {Thus, $M_{\bX}(z)$ has a root of magnitude at most $72e\bdd^3 \le 200 \bdd^3$.}
\end{proofof}

\ignore{Xue's old text -- Xue, delete this if you're happy with the last bit of calculation
\begin{claim}
If $g$ has no zero of magnitude at most $R_0$ for $R_0=10^4 C^2$, we have $|g(z)| > 1 - 0.5/C$ for all $z$ with $|z| \le 40$.
\end{claim}
We will use Hadamard's : for zeros $z_1,z_2,\ldots$ of an entire function $g$ with growth order 1 and $g(0) \neq 0$, there exists constant $A$ and $B$ such that
\[
g(z)=e^{Az+B} \prod_{n \ge 1} (1-z/z_n)e^{z/z_n} \textit{ for all } z.
\]
\begin{proof} 
Since $g(0)=1$, $B=0$. At the same time, $g(z)=g(-z)$ since $\bX$ is symmetric. Thus 
\[
g(z)^2 =\prod_{n} (1 - z^2 / z_n^2). 
\]
For any $|z| < \sqrt{R_0}$, we have
\[
|g(z)| \ge \prod_n ( 1 - |z^2|/z_n^2)^{1/2} \ge 1 - \frac{|z^2|}{2} \sum_{n} 1/|z_n|^2.
\]
By Corollary~\ref{cor:sum_roots}, we know $\sum_n 1/|z_n|^2 \le 2eC/R_0$ (for $\alpha=2$), which implies $|g(z)| \ge 1 - \frac{e C \cdot |z^2|}{R_0}$ for any $|z|  \le \sqrt{R_0}$.
\end{proof}

However, this contradicts with \Cref{clm:decrease_value} given $R_0=10^4 C^2$, which shows $f$ with $|f| \le 5$ and $|g(2 \pi \bi f)|\le 1-0.5/C$.}

\ignore{
\subsection{Proof of \Cref{clm:zero_roots} for integers}
We finish the proof of \Cref{clm:zero_roots} in this section.

\begin{proofof}{\Cref{clm:zero_roots}}
Since $X$ is symmetric, we could rewrite $\sum_{j=-C}^{C} \Pr[X=j] \cdot e^{j z}$ as a polynomial $Q(w)$ of $w=e^{-z}+e^{z}$:
$$
(e^{-kz}+e^{kz})=(e^{-z}+e^{z})^k - \sum_{i \in (-k,k): i-k \equiv 0 \mod 2} {k \choose i} \cdot (e^{-iz}+e^{-iz}).
$$
So the degree-$i$ coefficient in $Q(w)$ is bounded by 
$$\sum_{k \ge i:i-k \equiv 0 \mod 2} {k \choose i} \cdot \Pr[X=k] \le \sum_{k \ge i} {k \choose i} \Pr[X=k] \le {C \choose i}.
$$

The degree of $Q(w)$ is $C$, which indicates that there are $C$ complex roots $w_1,\ldots,w_{C}$. Let $w_1,\ldots,w_{C'}$ denote the number of non-zero roots. Since $Q$ is not a constant, $C' \ge 1$. Thus the coefficient of $w^{C-C'}$ in $Q$ equals $w_1 \cdot w_2 \cdots w_{C'}$, which shows 
$$|w_1| \cdot |w_2| \cdots |w_{C'}| \le {C \choose C-C'}={C \choose C'}.
$$ 
From all discussion above, we know there exists a non-zero complex root $w^*$ of $Q$ with $|w^*| \le C$.

We bound $|z_0|$ satisfying $e^{-z_0}+e^{z_0}=w^*$. For $z_0=r + \i \theta$, $e^{-z_0}=e^{-r} \cdot (\cos \theta - \i \sin \theta)$ with $|e^{-z_0}|=e^{-r}$ and $e^{z_0}=e^{r} \cdot (\cos \theta + \i \sin \theta)$ with $|e^{z_0}|=e^{r}$. Notice that when $r>6$, 
$$
|e^{-z_0}+e^{z_0}| \ge |e^{z_0}|-|e^{-z_0}| \ge e^r - e^{-r} \ge e^r/2.
$$ 
When $|w^{*}| > 10$, we have $r = \log |w^{*}| + O(1)$. Otherwise, we know $r=O(1)$.
\end{proofof}
However, this is contradicted with \Cref{clm:decrease_value}.
}

%% file: sym_algorithm.tex

\section{A more efficient tester for benign distributions}
\label{sec:symmetric}

While the sample complexity of \Cref{thm:main-informal} is independent of $n$, it is  quite large as a function of $k$ (of tower type), and it uses cumulants of very high order (again of tower type).  As mentioned in \Cref{remark:quantitative-preview}, in this section we show that if $\bX$ satisfies some mild conditions --- specifically, its first $k+1$ even cumulants are all nonzero --- then it is possible to give a much more efficient tester which uses only the first $k+1$ even cumulants.  The  sample complexity of this tester is only exponential in $k$:

\begin{theorem} [More efficient tester for distributions with nonzero even cumulants] \label{thm:test_sym_poly}
Fix any real random variable $\bX$ which has mean zero, variance one, and 
non-zero cumulants $|\kappa_4({\cal D})|,\ldots,|\kappa_{2k+2}({\cal D})| \ge \tau$.
Let $\ion$ be any real random variable with finite moments of order $2,4,\dots,4k+4$.\ignore{ Let $\ion$ be any real random variable with mean zero, variance 1, and finite moments of order $1,2,\dots,\red{4k+4}$.}
 
The algorithm described below is an $\eps$-tester for $k$-linearity under ${\cal D}$ and $\ion$, if $w$ is promised to satisfy $1/C \le \|w\|_2 \le C.$ Its sample complexity is 
\[
m=\poly(k!,m_{4k+4}({\cal D}),m_{4k+4}(\ion),(C/\eps)^k,1/\tau),
\]
where $\tau=\min\{\kappa_2(\bX),\kappa_4(\bX),\dots,\kappa_{2k+2}(\bX)\}.$
\end{theorem}

We remark that unlike \Cref{thm:general_alg}, the tester given by \Cref{thm:test_sym_poly} is not a tolerant tester.

The algorithm and its analysis use symmetric polynomials.
Let $\Sym_k(x_1,\ldots,x_n)$ denote the elementary symmetric polynomial $\sum_{S \in {[n] \choose k}} \prod_{i \in S}x_i$ of degree $k$ over variables $x_1,\ldots,x_n$.    The idea of the algorithm is extremely simple:  if $w$ is $k$-sparse then $\Sym_{k+1}(w_1^2,\dots,w_n^2)=0$, while if $w$ is far from $k$-sparse then $\Sym_{k+1}(w_1^2,\dots,w_n^2)$ must be bounded away from 0, and these two different cases can be distinguished using the tools developed in previous earlier sections.  

More formally, the analysis will use Newton's identity, which  gives us the following: for all $\ell \in \{1,\dots,n\}$, we have
\begin{equation}\label{eq:newton_sym}
\ell \cdot \Sym_\ell(w_1^2,\ldots,w_n^2) = \sum_{i=1}^\ell (-1)^{i-1} \cdot \Sym_{\ell-i}(w_1^2,\ldots,w_n^2) \cdot (\sum_{j \in [n]} w_j^{2i}).
\end{equation}

Now we describe the algorithm, which is very simple:

\begin{enumerate}
\item We first rescale all samples by a factor of $C$ so that $\|w\|_2 \in [1/C^2,1]$ in the following analysis. 

\item For $i=2,4,\dots,2k+2$, run the algorithm of \Cref{thm:number_samples_Lk_w} to obtain an estimate $M'_i$ of $\|w\|_i^i$ which is accurate to within an additive $\eps':=\frac{1}{3(k+1)} \cdot \frac{C^{-(4k+4)} \cdot \eps^{2k}}{(k+1)!}$.  Set $M_i=\min\{M'_i,1\}.$

(The rationale for taking the min is that since $\|w\|^2_2 \leq 1$, it must be the case that the true value of $\|w\|^i_i$ is at most 1.)

\item Set $\wt{S}_0=1$ and $\wt{S}_1=M_2$. Then for each $\ell=2,3,\ldots,k+1$, set $\wt{S}_\ell=\frac{1}{\ell} \cdot \sum_{i=1}^\ell (-1)^{i-1} \wt{S}_{\ell-i} \cdot M_{2i}$. 

($\wt{S}_\ell$ should be thought of as an estimate of $\Sym_\ell(w_1^2,\dots,w_n^2)$.)

\item If $\wt{S}_{k+1}>\frac{1}{2} \cdot \frac{C^{-(4k+4)} \cdot \eps^{2k}}{(k+1)!}$ then output ``far from $k$-sparse'' and otherwise output ``$k$-sparse.''
\end{enumerate}

\ignore{
}

Before we prove \Cref{thm:test_sym_poly}, we state and prove a technical result that is useful for the case in which $w$ is far from being $k$-sparse:

\begin{claim}\label{clm:lower_bound_sym_poly}
If $w$ is $\eps$-far from being $k$-sparse, then $\Sym_{k+1}(w_1^2,\ldots,w_n^2) \ge  \|w\|_2^{2k+2} \cdot  \frac{\eps^{2k}}{(k+1)!} $.
\end{claim}
\begin{proof}
Since $w$ is $\eps$-far from being $k$-sparse, $\sum_{i \notin S} w_i^2 > \eps^2 \cdot \|w\|_2^2$ for every subset $S \in {[n] \choose k}$.\ignore{\rnote{Is this exactly the definition exactly of $w$ (a unit vector) being $\eps$-far from $k$-sparse that we are using? Or is it that for every $k$-sparse $w'$ we have $\|w-w'\|\geq \eps$ --- in this case the ``$\eps$'' in this sentence would be an $\eps^2$ {\color{brown} Xue: Sorry, Rocco. It shall be the latter definition of $\eps^2$.}}}  We use induction to show that all $\ell \in [1,2,\ldots,k+1]$ have $\Sym_{\ell}(w_1^2,\ldots,w_n^2) \ge \|w\|_2^{2\ell} \cdot \eps^{2(\ell-1)}/\ell!$.

The base case is simple:  when $\ell=1$ we have $\Sym_1(w_1^2,\ldots,w_n^2)=\sum_i w_i^2=\|w\|_2^2$.

For the inductive step, we rewrite $\Sym_{\ell+1}(w_1^2,\dots,w_n^2)$ for $\ell \le k$ as follows:
\begin{align*}
\Sym_{\ell+1}(w_1^2,\ldots,w_n^2)&=\sum_{S \in {[n] \choose \ell+1} } \prod_{j \in S} w_j^2\\
&=\frac{1}{\ell+1} \sum_{T \in {[n] \choose \ell}} \prod_{j \in T} w_j^2 (\sum_{i \notin T} w_i^2)\\
& \ge \frac{1}{\ell+1} \sum_{T \in {[n] \choose \ell}} \prod_{j \in T} w_j^2 \cdot \eps^2 \cdot \|w\|_2^2 \\
& = \frac{ \eps^2 \cdot \|w\|_2^2}{\ell+1} \Sym_{\ell}(w_1^2,\dots,w_n^2),
\end{align*}
where the inequality follows from the first sentence of the proof.  From the induction hypothesis, this is at least $\|w\|_2^{2\ell+2} \cdot \frac{\eps^2}{\ell+1} \cdot \eps^{2(\ell-1)}/\ell!=\|w\|_2^{2\ell+2} \cdot \eps^{2\ell}/(\ell+1)!$.
\end{proof}

\begin{proofof}{\Cref{thm:test_sym_poly}}
We apply induction from $\ell=0$ to $k+1$ to bound the error between $\wt{S}_\ell$ and $\Sym_\ell(w_1^2,\ldots,w_n^2)$ by $\ell \cdot \eps'$. For brevity we simply write $\Sym_\ell$ for $\Sym_\ell(w_1^2,\ldots,w_n^2)$ for the rest of this section.

For the base cases $\ell=0$ and $\ell=1$, by definition we have that $\wt{S}_0$ and $\bigg| \Sym_1-\|w\|_2^2 \bigg| \le \eps'$.

For the inductive step, the error $|\wt{S}_\ell - \Sym_\ell|$ between $\wt{S}_\ell$ and $\Sym_\ell$ is 
\begin{align*}
& \frac{1}{\ell} \left| \sum_{i=1}^\ell (-1)^{i-1} \wt{S}_{\ell-i} \cdot M_{2i} - \sum_{i=1}^\ell (-1)^{i-1} \cdot \Sym_{\ell-i} \cdot \|w\|_{2i}^{2i} \right| \tag{def. of $\wt{S}_{\ell}$ and \Cref{eq:newton_sym}} \\
\le & \frac{1}{\ell} \sum_{i=1}^\ell \left| \wt{S}_{\ell-i} \cdot M_{2i} - \Sym_{\ell-i} \cdot \|w\|_{2i}^{2i} \right|\\
\le & \frac{1}{\ell} \sum_{i=1}^\ell \left( \left|\wt{S}_{\ell-i}-\Sym_{\ell-i} \right| \cdot M_{2i} + \Sym_{\ell-i} \cdot \left|M_{2i}-\|w\|_{2i}^{2i}\right| \right) \\
\le & \frac{1}{\ell} \sum_{i=1}^{\ell-1} \left|\wt{S}_{\ell-i}-\Sym_{\ell-i}\right| + \frac{1}{\ell} \sum_{i=1}^\ell \Sym_{\ell-i} \cdot \left| M_{2i}-\|w\|_{2i}^{2i}\right| \tag{using $M_{2i} \leq 1$}\\
\le & \frac{1}{\ell} \sum_{i=1}^{\ell-1} (\ell-i) \cdot \eps' + \frac{1}{\ell} \sum_{i=1}^\ell \frac{\eps'}{(\ell-i)!} \tag{inductive step and $\Sym_\ell \le 1/\ell!$}\\
\le & \frac{\ell(\ell-1)}{2 \ell} \cdot \eps' + \frac{e}{\ell} \cdot \eps' \le \ell \cdot \eps',
\end{align*}
where we note that for the penultimate inequality, we have $\Sym_\ell(w_1^2,\ldots,w_n^2)\le (\sum_j w_j^2)^\ell/\ell! = 1/\ell!$.

Thus we have established that $|\wt{S}_\ell - \Sym_\ell| \leq \ell \cdot \eps'$; we now use this to argue correctness of our algorithm. 
This is simple:
if $w$ is $k$-sparse then $\Sym_{k+1}(w_1^2,\ldots,w_n^2)=0$ and by the above bound we have $\wt{S}_{k+1} \le (k+1) \eps' \le \frac{1}{3} \cdot \frac{C^{-(4k+4)} \cdot \eps^{2k}}{(k+1)!}$. On the other hand, if $w$ is $\eps$-far from $k$-sparse then by \Cref{clm:lower_bound_sym_poly} we have that $\Sym_{k+1}(w_1^2,\ldots,w_n^2) \ge \frac{C^{-(2k+2)} \cdot \eps^{2k}}{(k+1)!}$ from, which implies that $\wt{S}_{k+1} \ge \frac{C^{-(4k+4)} \cdot \eps^{2k}}{(k+1)!} - (k+1) \eps' \ge \frac{2}{3} \cdot \frac{C^{-(4k+4)} \cdot \eps^{2k}}{(k+1)!}$, and the proof is complete.
\end{proofof}

%% file: D-must-be-iid.tex

\section{Proof of \Cref{thm:iid}: ${\cal D}$ must be i.i.d. for finite-sample testability} \label{sec:iid}

For ease of presentation we first prove the following variant of \Cref{thm:iid} which deals with ($\eps=0.5$)-testers rather than ($\eps=0.99$)-testers.  After the proof we describe the minor (but notationally cumbersome) changes to the proof which yield \Cref{thm:iid} as stated earlier.

\medskip

\noindent {\bf Variant of \Cref{thm:iid}.} (${\cal D}$ must be i.i.d. for finite-sample testability)\emph{.}
\emph{
Let ${\cal D}$ be the product distribution $(\Poi(1))^{n/2} \times
(\Poi(2))^{n/2}$.  Then even if there is no noise (i.e.~the noise distribution $\ion$ is identically zero), any algorithm which is an $(\eps=0.5)$-tester for 1-linearity under ${\cal D}$ must have sample complexity $m = \Omega({\frac {\log n}{\log \log n}}).$
}

\medskip

\begin{proof}
We consider two different distributions over the target vector $w$.
The first distribution, denoted ${\cal W}_{\yes}$, is uniform over the $n/2$ canonical basis vectors $\{e_{n/2+1},\dots,e_{n}\} \subset \R^n$ (recall that these are the coordinates whose corresponding $\bx_i$'s are distributed as $\Poi(2)$ under ${\cal D}$).  The second distribution, denoted ${\cal W}_{\no}$, is uniform over the multiset of $(n/2)^2$ many vectors $\{e_i + e_j\}_{1 \leq i ,j \leq n/2}$ (recall that for $i \leq n/2$ the random variable $\bx_i$ is distributed as $\Poi(1)$ under ${\cal D}$).  It is clear that every $w$ in the support of ${\cal W}_{\yes}$ is $1$-sparse and that a $1-{\frac {\Theta(1)}{n}}$ fraction of outcomes of $w$ in the support of ${\cal W}_{\no}$ are $0.5$-far from being $1$-sparse.

Let ${\cal P}_{\midd}$ be the following distribution over $(\ba,\bb)$ pairs in $\N^n \times \N$:  in a draw from ${\cal P}_{\midd}$, $\ba$ is drawn from ${\cal D}=(\Poi(1))^{n/2} \times (\Poi(2))^{n/2}$ and $\bb$ is independently drawn from $\Poi(2).$  Let $
(\overline{\ba},\overline{\bb}) := ((\ba^{(1)},\bb^{(1)}),\dots,(\ba^{(m)},\bb^{(m)}))$ be a sequence of $m$ pairs drawn independently from ${\cal P}_{\midd}.$

The following claim is crucial to the proof of the variant of \Cref{thm:iid}:

\begin{claim} \label{claim:noniid} Let $\star \in \{\no,\yes\}.$ Let $\bw$ be drawn from ${\cal W}_{\star}$ and let $(\overline{\bx},\overline{\by}) := ((\bx^{(1)},\by^{(1)}),\dots,(\bx^{(m)},\by^{(m)}))$ be a sequence of $m$ examples independently generated as follows for each $i$:  $\bx^{(i)}$ is drawn from ${\cal D}=(\Poi(1))^{n/2} \times (\Poi(2))^{n/2}$ and $\by^{(i)} \leftarrow \bw \cdot \bx.$   Then for some  $m = {\frac {c \log n}{\log \log n}}$ where $c>0$ is a sufficiently small absolute constant, the variation distance between $(\overline{\ba},\overline{\bb})$ and $(\overline{\bx},\overline{\by})$ is at most 0.01.
\end{claim}

We prove \Cref{claim:noniid} for $\star = \no$ below; the case when $\star = \yes$ follows by a similar (simpler) proof.  The variant of \Cref{thm:iid} follows directly from this claim and the triangle inequality.
\end{proof}

\subsection{Proof of \Cref{claim:noniid}.}

\ignore{
}
Since $\overline{\ba}=(\ba^{(1)},\dots,\ba^{(m)})$ and $\overline{\bx}=(\bx^{(1)},\dots,\bx^{(m)})$ are both distributed according to ${\cal D}^m$, \Cref{claim:noniid} follows directly from the following:

\begin{claim} \label{claim:noniid2}
With probability at least $0.999$ over a draw of $\overline{\bx}$ from ${\cal D}^m$, the distribution of $\overline{\by}$ conditioned on $\overline{\bx}$ (which we denote $\overline{\by}(\overline{\bx})$) has variation distance at most $0.0001$ from the distribution of $\overline{\bb}$ (which is simply $(\Poi(2))^m$).
\end{claim}

In the rest of this subsection we prove \Cref{claim:noniid2}.

For $j \in \N$ and $v \in \{1,2\}$, let us write $\Poi(v)(j)$ to denote $v^j e^{-v}/j!$, the probability weight that the $\Poi(v)$ distribution puts on the outcome $j$. 
For $\overline{j} = (j_1,\dots,j_m) \in \N^{m}$, let $\Poi(v)(\overline{j})$ denote $\prod_{\ell=1}^m \Poi(v)(j_{\ell})$, the probability that a sequence of $m$ independent draws from $\Poi(v)$ come out as $j_1,\dots,j_m$.

Given $x \in \N^{n}$ and $j \in \N$, let us write $\freq(x,j)$ to denote the fraction of the first $n/2$ coordinates of $x$ that have value $j$.  Similarly, given $\overline{x} \in \N^{m \times n}$ and $\overline{j} \in \N^m$, let $\freq(\overline{x},\overline{j})$ denote the fraction of the first $n/2$ columns in the matrix $\overline{x}$ which match $\overline{j}.$

We say that a matrix $\overline{x} \in \N^{m \times n}$ is \emph{good} if
\[
\|\freq(\overline{x},\cdot) - \Poi(1)(\cdot)\|_1 := 
\sum_{\overline{j} \in \N^m} |\freq(\overline{x},\overline{j}) - \Poi(1)(\overline{j})| \leq \delta := 0.00005.
\]

The following claim is where we use the fact that $m = {\frac {c \log n}{\log \log n}}$:
\begin{claim} \label{claim:egg}
A random $\overline{\bx} \sim {\cal D}^m$ is good with probability at least $0.999$.
\end{claim}
\begin{proof}
Since $\Poi(1)(j) = e^{-1}/j!$, we have that $\sum_{j\geq k} \Poi(1)(j) \leq {\frac 1 {k!}}$.  Taking $k = {\frac {C \log n}{\log \log n}}$ for a suitably large constant $C$ and doing a union bound over $i \in [m]$, we have that 
\begin{equation} \label{eq:light-part}
\sum_{\ol{j} \in \N^m \ : \ j_i \geq k \text{~for some~}i} \Poi(1)(\ol{j}) \leq \tfrac{\delta}{4} \quad \quad \text{(with room to spare).}
\end{equation}
We will argue below that with probability at least $0.999$, a random $\ol{\bx} \sim {\cal D}^m$ satisfies
\begin{equation} \label{eq:heavy-part}
\sum_{\ol{j} \in \{0,\dots,k-1\}^m} |\freq(\ol{\bx},\ol{j}) - \Poi(1)(\ol{j})| \leq \tfrac{\delta}{4}
\end{equation}
Together with \Cref{eq:light-part} and the fact that both $\freq(\ol{x},\cdot)$ and $\Poi(1)(\cdot)$ sum to 1, this implies that $\sum_{\ol{j} \in \N^m \ : \ j_i \geq k \text{~for some~}i} \freq(\ol{\bx})(\ol{j}) \leq \tfrac{\delta}{2}$, and this with \Cref{eq:light-part} and \Cref{eq:heavy-part} establishes \Cref{claim:egg}.

To establish \Cref{eq:heavy-part}, first observe that for a suitable choice of the (small) absolute constant $c$ in the definition of $m = {\frac {c \log n}{\log \log n}}$, the number of summands $\ol{j}$ in \Cref{eq:heavy-part} is at most $k^m < n^{0.001}.$ Fix a specific $\ol{j}$ that appears in \Cref{eq:heavy-part}.  For each $i \in [n/2]$, the probability that the $i$-th column of $\ol{\bx}$ contributes to $\freq(\ol{\bx},\ol{j})$ is precisely $\Poi(1)(\ol{j})$, and consequently the random variable $\freq(\ol{\bx},\ol{j})$ can be viewed as the observed frequency of heads in $n/2$ tosses of a coin which comes up heads each time with probability $\Poi(1)(\ol{j})).$ A standard Chernoff bound gives that the observed frequency in $n/2$ such coin tosses differs from the expected frequency by an additive $\pm n^{-1/4}$ with probability at most $2^{-\Theta(\sqrt{n})}$, so at most a $2^{-\Theta(\sqrt{n})}$ fraction of outcomes of $\ol{\bx}$ are ``bad'' for $\ol{j}$ (in the sense of contributing more than $n^{-1/4}$ to the sum in \Cref{eq:heavy-part}).  A union bound over all (at most $n^{0.001}$ many) $\ol{j}$'s that appear in the sum in \Cref{eq:heavy-part} completes the proof.\ignore{\rnote{We can think about sharpening the argument to get $m = c \log n$, but I don't think it is a high priority.}}
\end{proof}

Fix a good matrix $\overline{x}$, i.e.~one which satisfies $\sum_{\overline{j} \in \N^m} |\freq(\overline{x},\overline{j}) - \Poi(1)(\overline{j})| \leq {\delta}.$  For any
$\ell \in [m]$, the $\ell$-th coordinate of $\overline{\by}(\overline{x})$ takes value $i \in \N$ with probability $\sum_{j, j' \in \N: j + j' = i} \freq(x^{(\ell)},j) \cdot \freq(x^{(\ell)},j')$ (this follows from the definition of the ${\cal W}_{\no}$ distribution).  Similarly, the $m$-dimensional vector $\overline{\by}(\overline{x})$ takes value $\overline{i}
\in \N^m$ with probability $\sum_{\overline{j}, \overline{j'} \in \N: \overline{j} + \overline{j'} = \overline{i}} \freq(x^{(\ell)},\overline{j}) \cdot \freq(x^{(\ell)},\overline{j'}).$ 
Recalling that each coordinate of $\overline{\bb}$ is independently distributed as $\Poi(2)$, we have that
\begin{align}
\dtv(\overline{\by}(\overline{x}),\overline{\bb})
&=
\sum_{\overline{i} \in \N^m} \left|
\sum_{\overline{j},\overline{j'} \in \N^m: \overline{j} + \overline{j'} = \overline{i}}
\freq(\overline{x},\overline{j}) \cdot \freq(\overline{x},\overline{j'}) -
\Poi(2)(\overline{i})
\right| \nonumber \\
&=
\sum_{\overline{i} \in \N^m} \left|
\sum_{\overline{j},\overline{j'} \in \N^m: \overline{j} + \overline{j'} = \overline{i}}
\freq(\overline{x},\overline{j}) \cdot \freq(\overline{x},\overline{j'}) -
\Poi(1)(\overline{j}) \cdot \Poi(1)(\overline{j'})
\right| \label{eq:pickle}\\
&\leq
\sum_{\overline{j},\overline{j'} \in \N^m} |\freq(\overline{x},\overline{j})
\cdot \freq(\overline{x},\overline{j'}) - \Poi(1)(\overline{j})\cdot \Poi(1)(\overline{j'})| \nonumber \\
&\leq \sum_{\overline{j},\overline{j'} \in \N^m} \freq(\overline{x},\overline{j})
\cdot |\freq(\overline{x},\overline{j'}) - \Poi(1)(\overline{j'})|
+
\sum_{\overline{j},\overline{j'} \in \N^m} \Poi(1)(\overline{j'})
\cdot |\freq(\overline{x},\overline{j}) - \Poi(1)(\overline{j})|
\nonumber \\ 
&\leq \sum_{\overline{j}\in \N^m} \freq(\overline{x},\overline{j}) \cdot \delta
+
\sum_{\overline{j'} \in \N^m} \Poi(1)(\overline{j'}) \cdot \delta
 \leq 2{\delta} = 0.0001, \nonumber
\end{align}
where \Cref{eq:pickle} uses the identity $\sum_{\overline{j},\overline{j'} \in \N^m: \overline{j} + \overline{j'} = \overline{i}} \Poi(1)(\overline{j}) \cdot \Poi(1)(\overline{j'}) = \Poi(2)(\overline{i})$ (which is a direct consequence of the fact that the sum of two independent draws from a $\Poi(1)$ random variable is a $\Poi(2)$ random variable).
This concludes the proof of \Cref{claim:noniid2}. \qed

\medskip

We now discuss the modifications to the above proof which give \Cref{thm:iid} as originally stated (with $\eps=0.99$).  The definition of ${\cal W}_{\yes}$ is unchanged (note that now the last $n/2$ coordinates correspond to the $\Poi(100)$ distribution) but now ${\cal W}_{\no}$ is uniform over the multiset of $(n/2)^{100}$ many vectors $\{e_{i_1} + \cdots + e_{i_{100}}\}_{1 \leq i_1,\dots,i_{100} \leq n/2}$; now a $1-{\frac {\Theta(1)}{n}}$ fraction of outcomes of $w$ in the support of ${\cal W}_{\no}$ are $0.99$-far from being $1$-sparse.
The distribution $\ba$ is now drawn from ${\cal D}=(\Poi(1))^{n/2} \times (\Poi(100))^{n/2}$ and $\bb$ is independently drawn from $\Poi(100).$ In the chain of inequalities of which \Cref{eq:pickle} is a part, the inner sum of \Cref{eq:pickle} is now indexed by 100-tuples of vectors $\overline{j^{(1)}},\dots,\overline{j^{(100)}}$ which sum to $\overline{i}$, and in the RHS of that chain of inequalities we ultimately get $100\delta$ rather than $2 \delta.$ The proof can be completed along these lines with minor changes to the argument given above.

%% file: noise-known.tex

\section{Proof of \Cref{thm:two-noise}: the noise distribution must be known for finite-sample testability} \label{ap:two-noise}

As in \Cref{sec:iid}, we first prove the following variant which deals with ($\eps=0.5$)-testers and later indicate the changes necessary to get \Cref{thm:two-noise}:

\medskip

\noindent {\bf Variant of  \Cref{thm:two-noise}.} (The noise distribution $\ion$ must be known)\emph{.}
\emph{Let ${\cal D}$ be the i.i.d.~product distribution ${\cal D}=(\Poi(1))^n$.  Suppose that the noise distribution $\ion$ is unknown to the testing algorithm but is promised to be either $\Poi(1)$ or $\Poi(2)$.  Then any $(\eps=0.5)$-tester for 1-linearity under ${\cal D}$ and the unknown noise distribution $\ion \in \{\Poi(1),\Poi(2)\}$ must have sample complexity $m=\Omega({\frac {\log n}{\log \log n}}).$ 
}

\begin{proof}
The proof is similar to the proof of the Variant of \Cref{thm:iid}. We consider two different distributions over the target vector $w$.
The first distribution, denoted ${\cal W}_{\yes}$, is uniform over the $n$ canonical basis vectors $\{e_1,\dots,e_n\} \subset \R^n$.  The second distribution, denoted ${\cal W}_{\no}$, is uniform over the multiset of $n^2$ vectors $\{e_i + e_j\}_{1 \leq i, j \leq n}$. It is clear that every $w$ in the support of ${\cal W}_{\yes}$ is $1$-sparse and that a $1-{\frac {\Theta(1)}{n}}$ fraction of outcomes of $w$ in the support of ${\cal W}_{\no}$ are $0.5$-far from being $1$-sparse.

Let ${\cal P}_{\midd}$ be the following distribution over $(\ba',\bb')$ pairs in $\N^n \times \N$:  in a draw from ${\cal P}_{\midd}$, $\ba'$ is drawn from ${\cal D}=(\Poi(1))^n$ and $\bb'$ is independently drawn from $\Poi(3).$  Let $
(\overline{\ba'},\overline{\bb'}) := ((\ba'^{(1)},\bb'^{(1)}),\dots,(\ba'^{(m)},\bb'^{(m)}))$ be a sequence of $m$ pairs drawn independently from ${\cal P}_{\midd}.$

We will use the following claims:

\begin{claim} \label{claim:no-no-noise} Let $\bw$ be drawn from ${\cal W}_{\no}$ and let $(\overline{\bx'},\overline{\by'}) := ((\bx'^{(1)},\by'^{(1)}),\dots,(\bx'^{(m)},\by'^{(m)}))$ be a sequence of $m$ examples independently generated as follows for each $i$:  $\bx'^{(i)}$ is drawn from ${\cal D}=(\Poi(1))^n$ and $\by'^{(i)} \leftarrow \bw \cdot \bx' + \Poi(1)$.  Then for some $m = \Omega({\frac {\log n}{\log \log n}})$, the variation distance between $(\overline{\ba'},\overline{\bb'})$ and $(\overline{\bx'},\overline{\by'})$ is at most 0.01.
\end{claim}

\begin{claim} \label{claim:yes-no-noise} Let $\bw$ be drawn from ${\cal W}_{\yes}$ and let $(\overline{\bx'},\overline{\by'}) := ((\bx'^{(1)},\by'^{(1)}),\dots,(\bx'^{(m)},\by'^{(m)}))$ be a sequence of $m$ examples independently generated as follows for each $i$:  $\bx'^{(i)}$ is drawn from ${\cal D}=(\Poi(1))^n$ and $\by'^{(i)} \leftarrow \bw \cdot \bx' + \Poi(2)$.  Then for some $m = \Omega({\frac {\log n}{\log \log n}})$, the variation distance between $(\overline{\ba'},\overline{\bb'})$ and $(\overline{\bx'},\overline{\by'})$ is at most 0.01.
\end{claim}

We prove \Cref{claim:no-no-noise} below; \Cref{claim:yes-no-noise} follows by a similar (but simpler) proof.  \Cref{thm:two-noise} follows directly from these two claims and the triangle inequality.
\end{proof}

\subsection{Proof of \Cref{claim:no-no-noise}.}
\ignore{
} 
Since $\overline{\ba'}=(\ba'^{(1)},\dots,\ba'^{(m)})$ is distributed identically to $\overline{\bx'}=(\bx'^{(1)},\dots,\bx'^{(m)})$, \Cref{claim:no-no-noise} follows directly from the following:

\begin{claim} \label{claim:noniid2prime}
With probability at least $0.999$ over a draw of $\overline{\bx}$ from ${\cal D}^m$, the distribution of $\overline{\by'}$ conditioned on $\overline{\bx'}$ (which we denote $\overline{\by'}(\overline{\bx'})$) has variation distance at most $0.01$ from the distribution of $\overline{\bb'}$ (which is simply $(\Poi(3))^m$).
\end{claim}

\Cref{claim:noniid2prime} is established essentially by a reduction to \Cref{claim:noniid}.\ignore{

}
Recall the following standard fact about total variation distance:
\begin{fact} \label{fact:dtv-inequality}
Let $\bA,\bA'$ be two random variables and let $\bB$ be independent of $\bA$ and of $\bA'$.  Then $\dtv(\bA + \bB, \bA' + \bB) \leq \dtv(\bA,\bA').$
\end{fact}

We will apply \Cref{fact:dtv-inequality} as follows.  Define $(\overline{\bx''},\overline{\by''}):=((\bx''^{(1)},\by''^{(1)}),\dots,(\bx''^{(m)},\by''^{(m)}))$ to be the $m \times (n+1)$ matrix valued random variable distributed as follows:  to make a draw of $(\overline{\bx''},\overline{\by''})$, first draw $\bw \sim {\cal W}_{\no}$, then independently 
for each $i$ let $\bx''^{(i)}$ be drawn from ${\cal D}=(\Poi(1))^n$ and let $\by''^{(i)} \leftarrow \bw \cdot \bx''$. We take $\bA$ to be the random variable $(\overline{\bx''},\overline{\by''}).$
We take $\bB$ to be an $m \times (n+1)$ matrix valued random variable distributed as follows:  the first $n$ columns are identically $0^m$, and the $(n+1)$-st column is distributed as $(\Poi(1))^m$. Finally, we take $\bA'$ to be the $m \times (n+1)$ dimensional matrix valued random variable $(\overline{\bx''},(\Poi(2))^m)$ (i.e.~the first $n$ columns are distributed as $\overline{\bx''}$ above and the last column consists of $m$ independent draws from $\Poi(2)$). 

A trivial modification of the proof of the ($\star = \no$)-case of \Cref{claim:noniid} gives that $\dtv(\bA,\bA') \leq 0.01$, so by \Cref{fact:dtv-inequality} we get that $\dtv(\bA + \bB,\bA' + \bB) \leq 0.01$.  The random variable $\bA + \bB$ is distributed precisely as $(\overline{\bx'},\overline{\by'})$, and since $\Poi(1) + \Poi(2)$ is distributed as $\Poi(3)$, the random variable $\bA' + \bB$ is distributed precisely as $(\overline{\ba'},\overline{\bb'})$, so \Cref{claim:no-no-noise} is proved. \qed

\medskip

We now discuss the modifications to the above proof which give \Cref{thm:two-noise} as originally stated (with $\eps=0.99$).  The definition of ${\cal W}_{\yes}$ is unchanged but now ${\cal W}_{\no}$ is uniform over the multiset of $n^{100}$ many vectors $\{e_{i_1} + \cdots + e_{i_{100}}\}_{1 \leq i_1,\dots,i_{100} \leq n}$. In ${\cal P}_{\midd}$, the label coordinate $\bb'$ is now independently drawn from $\Poi(101)$. The statement of \Cref{claim:no-no-noise} is unchanged, and its proof changes in the obvious ways to accomodate the new definition of $\bb'$; in \Cref{claim:yes-no-noise} the ``$\Poi(2)$'' is replaced by ``$\Poi(100)$.'' The rest of the argument follows with the obvious changes.

%% file: Gaussian_lower_bound.tex

\section{Proof of \Cref{thm:gaussian-lb}: a lower bound for Gaussian distributions} \label{sec:gaussian-lb}

We prove a more general result which implies \Cref{thm:gaussian-lb} as a special case:

\begin{theorem} [Strengthening of \Cref{thm:gaussian-lb}]\label{thm:gaussian-lb-strengthened}
Let ${\cal D}$ be the standard $N(0,1)^n$ $n$-dimensional Gaussian distribution and let $\ion$ be $N(0,c^2)$ where $c>0$ is any constant.  Then for any $\tau>0$ and any $1 \leq k \leq n^{1-\tau}$, the sample complexity of any $(\eps=0.99)$-tester for $k$-sparsity under ${\cal D}$ and $\ion$ is $\Omega(\log n).$
\end{theorem}

Throughout this section we write $\ol{\bx}$ to denote a $t \times n$ matrix of $t$ examples $\bx^{(1)},\dots,\bx^{(t)}$ where each $\bx^{(j)}$ is independently distributed as $N(0,1)^n$. We write $\ol{\by}$ to denote an associated vector of labels $(\by^{(1)},\dots,\by^{(t)}) \in \R^t.$

We begin by observing that it suffices to prove the result for small constant noise rates $0 < c < 0.1$. This is because if there were a successful testing algorithm that worked in the presence of Gaussian noise at higher rates, such an algorithm could be used to test successfully in the presence of Gaussian noise at lower rates (simply by adding additional Gaussian noise to each label).  

We first state the main technical result of this section and a closely related corollary:

\begin{lemma}\label{lem:Gaussian_one_sparse}
Let $0 < c < 0.1$ and let $t={\frac {\log n}{10 \log 1/c}}.$ Let  ${\cal S}_{\yes}$ and ${\cal S}_{\no}$ denote the following two distributions over $t$-element samples $(\overline{\bx}, \overline{\by})$ where $\ov{\bx} \in \mathbb{R}^{t \times n}$ and $\ov{\by} \in \mathbb{R}^{t}$: 
\begin{enumerate}
\item In a draw from ${\cal S}_{\no}$, $\ov{\bx} \sim N(0,1)^{t \times n}$ and independently $\by \sim N(0,1+c^2)^{t}$.
\item A draw from ${\cal S}_{\yes}$ is obtained as follows: first draw $\overline{\bx} \sim N(0,1)^{t \times n}$.  Then draw a uniform $\bi \in [n]$, and output $(\overline{\bx},\overline{\by})$ where $\by^{(j)}=\bx^{(j)}_{\bi}+N(0,c^2)$ for every $j \in [t]$. 
\end{enumerate}
Then the statistical distance $\dtv({\cal S}_{\no},{\cal S}_{\yes})$ between ${\cal S}_{\no}$ and ${\cal S}_{\yes}$ is $o_n(1)$.
\end{lemma}

\begin{corollary}\label{cor:stat_dist_k_sparse}
Let $0 < c < 0.1$ and let $t={\frac {\log n/k}{10 \log 1/c}}.$  
For any sufficiently large $n$, given any sparsity $1 \leq k \leq n^{1-\tau}$, let ${\cal S}'_{\yes}(k)$ and ${\cal S}'_{\no}$ denote the following two distributions over $(\overline{\bx}, \overline{\by})$ where $\ov{\bx} \in \mathbb{R}^{t \times n}$ and $\ov{\by} \in \mathbb{R}^{t}$: 
\begin{enumerate}
\item ${\cal S}'_{\no}$ is identical to ${\cal S}_{\no}$;
\item A draw from ${\cal S}'_{\yes}(k)$ is obtained as follows: first draw $\overline{\bx} \sim N(0,1)^{t \times n}$.  Then draw a uniform $\bi \in \{0,1,\ldots,n/k-1\}$\footnote{We assume without loss of generality that $k$ divides $n$.}, and output $(\overline{\bx},\overline{\by})$ where $\by^{(j)}=\frac{\bx^{(j)}_{\bi \cdot k+1}+\cdots+\bx^{(j)}_{\bi \cdot k+k}}{\sqrt{k}}+N(0,c^2)$ for every $j \in [t]$. 
\end{enumerate} 
Then the statistical distance $\dtv({\cal S}'_{\no},{\cal S}'_{\yes}(k))$ between ${\cal S}'_{\no}$ and ${\cal S}'_{\yes}(k)$ is $o_n(1)$.
\end{corollary}

We finish the proof of \Cref{thm:gaussian-lb-strengthened} here and and defer the proofs of \Cref{lem:Gaussian_one_sparse} and \Cref{cor:stat_dist_k_sparse} to~\Cref{sec:one_sparse} and~\Cref{sec:proof_k_sparse} respectively.

\begin{proofof}{\Cref{thm:gaussian-lb-strengthened} using \Cref{cor:stat_dist_k_sparse}}
We apply \Cref{cor:stat_dist_k_sparse} twice, to sparsity $k$ and $100k$ separately. Since ${\cal S}_{\no}$ is the same in these two applications of \Cref{cor:stat_dist_k_sparse}, by the triangle inequality for total variation distance we get that $\dtv({\cal S}'_{\yes}(k),{\cal S}'_{\yes}(100k)) = o_n(1)$ 
for $t = \frac{\log n/(100k)}{10 \log 1/c}$. Since the target vector underlying ${\cal S}'_{\yes}(k)$ is $k$-sparse and the target vector underlying ${\cal S}'_{\yes}(100k)$ is $0.99$-far from being $k$-sparse, we get that no algorithm for ($\eps=0.99$)-testing $k$-sparsity can succeed given at most $t$ samples.
Since $k \leq n^{1-\tau}$ and $c = \Theta_n(1)$ the value of $t$ is $\Omega(\log n)$, and the theorem is proved.
\end{proofof}



\subsection{Proof of \Cref{lem:Gaussian_one_sparse}}\label{sec:one_sparse}
Recall that $\overline{\bx}=(\bx^{(1)},\ldots,\bx^{(t)})$ where each $\bx^{(i)} \sim N(0,1)^n$.
Throughout this section we have $0 < c< 0.1$, $t={\frac {\log n}{10 \log 1/c}}.$  

From the definition of ${\cal S}_{\no}$, we have that the pdf of ${\cal S}_{\no}$ at any point $(\ol{x},\ol{y}) \in \R^{t \times n} \times \R^t$ is
\begin{equation}\label{eq:distribution_no}
{\cal S}_{\no} (\overline{x},\overline{y}) = \prod_{i \in [t]} \left[ \left( \prod_{j=1}^n \frac{1}{\sqrt{2 \pi}} e^{- (x^{(i)}_j)^2/2} \right) \cdot \frac{1}{\sqrt{2 \pi (1+c^2)}} e^{-\frac{(y^{(i)})^2}{2(1+c^2)}} \right],
\end{equation}
and likewise  the pdf of ${\cal S}_{\yes}$ at $(\ol{x},\ol{y})$ is
\begin{equation}\label{eq:distribution_yes}
{\cal S}_{\yes}(\overline{x},\overline{y}) = \frac{1}{n} \sum_{\ell \in [n]} \left( \prod_{i \in [t]} \left[  \left( \prod_{j=1}^n \frac{1}{\sqrt{2 \pi}} e^{- (x^{(i)}_j)^2/2} \right) \cdot \frac{1}{\sqrt{2 \pi c^2}} e^{-\frac{\big(y^{(i)} - x^{(i)}_{\ell} \big)^2}{2 c^2}} \right] \right).
\end{equation}

We use the following claim to bound the statistical distance between ${\cal S}_{\no}$ and ${\cal S}_{\yes}$.
\begin{claim}\label{clm:adv_random_guess}
We have that
\begin{align*}
\underset{(\overline{\bx},\overline{\by}) \sim {\cal S}_{\no} }{\Pr}\left[ {\cal S}_{\no}(\overline{\bx},\overline{\by}) \ge{\cal S}_{\yes}(\overline{\bx},\overline{\by}) \right] &= {\frac 1 2} \pm \frac{1}{n^{\Omega(1)}}, \text{ and }\\
\underset{(\overline{\bx},\overline{\by}) \sim {\cal S}_{\yes} }{\Pr}\left[  {\cal S}_{\no}(\overline{\bx},\overline{\by}) <{\cal S}_{\yes}(\overline{\bx},\overline{\by})\right] &= {\frac 1 2} \pm \frac{1}{n^{\Omega(1)}}.
\end{align*}
\end{claim}

As a direct corollary of \Cref{clm:adv_random_guess}, we have:


\begin{corollary}\label{cor:probability_in_S}
Define the subset $S \subseteq \R^{t \times n} \times \R^t$,
$S:=\{ (\overline{\bx},\overline{\by})  : {\cal S}_{\no}(\ol{x},\ol{y})>
{\cal S}_{\yes}(\ol{x},\ol{y})\}$.
Then both
$\underset{(\overline{\bx},\overline{\by}) \sim {\cal S}_{\no} }{\Pr} \left[ (\overline{\bx},\overline{\by}) \in S \right]$ and $\underset{(\overline{\bx},\overline{\by}) \sim {\cal S}_{\yes} }{\Pr} \left[ (\overline{\bx},\overline{\by}) \in \overline{S} \right]$ are at most ${\frac 1 2} + \frac{1}{n^{\Omega(1)}}$.
\end{corollary}

Let us use \Cref{cor:probability_in_S} to finish the proof of \Cref{lem:Gaussian_one_sparse} here before proving \Cref{clm:adv_random_guess}:

\begin{proofof}{\Cref{lem:Gaussian_one_sparse} using \Cref{cor:probability_in_S}}
By the definition of statistical distance and of $S$, we have that $\dtv({\cal S}_{\no},{\cal S}_{\yes})$ is equal to
\begin{align*}
& \int_{\overline{x},\overline{y}} \bigg| {\cal S}_{\no}(\overline{x},\overline{y}) - {\cal S}_{\yes}(\overline{x},\overline{y}) \bigg| d \overline{x} d \overline{y} \\
= & \underset{(\overline{\bx},\overline{\by}) \sim {\cal S}_{\no}}{\Pr} \left[(\overline{\bx},\overline{\by}) \in S \right] - \underset{(\overline{\bx},\overline{\by}) \sim {\cal S}_{\yes}}{\Pr}\left[(\overline{\bx},\overline{\by}) \in S\right] + \underset{(\overline{\bx},\overline{\by}) \sim {\cal S}_{\yes}}{\Pr} \left[(\overline{\bx},\overline{\by}) \in \overline{S} \right] - \underset{(\overline{\bx},\overline{\by}) \sim {\cal S}_{\no}}{\Pr}\left[(\overline{\bx},\overline{\by}) \in \overline{S} \right] \\
= & \left(\underset{(\overline{\bx},\overline{\by}) \sim {\cal S}_{\no}}{\Pr} \left[(\overline{\bx},\overline{\by}) \in S \right] - 1/2\right) + 
\left(1/2 - \underset{(\overline{\bx},\overline{\by}) \sim {\cal S}_{\yes}}{\Pr}\left[(\overline{\bx},\overline{\by}) \in S\right]\right)\\
&  + \left(\underset{(\overline{\bx},\overline{\by}) \sim {\cal S}_{\yes}}{\Pr} \left[(\overline{\bx},\overline{\by}) \in \overline{S} \right] -1/2 \right)+ \left(1/2 - \underset{(\overline{\bx},\overline{\by}) \sim {\cal S}_{\no}}{\Pr}\left[(\overline{\bx},\overline{\by}) \in \overline{S} \right] \right).
\end{align*}
Because $\underset{(\overline{\bx},\overline{\by}) \sim {\cal S}_{\no}}{\Pr} \left[(\overline{\bx},\overline{\by}) \in S \right] + \underset{(\overline{\bx},\overline{\by}) \sim {\cal S}_{\no}}{\Pr}\left[(\overline{\bx},\overline{\by}) \in \overline{S} \right] = 1$, the first parenthesized term equals the fourth parenthesized term in the last line above, and similarly the second term equals the third term. We thus further simplify the above to
\[2 \left(\underset{(\overline{\bx},\overline{\by}) \sim {\cal S}_{\no}}{\Pr} \left[(\overline{\bx},\overline{\by}) \in S \right] - 1/2\right) + 2 \left( \underset{(\overline{\bx},\overline{\by}) \sim {\cal S}_{\yes}}{\Pr} \left[(\overline{\bx},\overline{\by}) \in \overline{S} \right] -1/2 \right),\]
which by~\Cref{cor:probability_in_S} is at most $\frac{4}{n^{\Omega(1)}}$.
\end{proofof}

In the rest of this subsection we prove~\Cref{clm:adv_random_guess}. Since $\overline{\bx}$ is generated the same way in ${\cal S}_{\no}$ and ${\cal S}_{\yes}$, let us simply write ${\cal S}(\ol{x})$ to denote the pdf of the marginal of either distribution (${\cal S}_{\no}$ or ${\cal S}_{\yes}$) over $\ol{x}$, i.e.
\[
{\cal S}(\ol{x}) = \prod_{i \in [t]}  \prod_{j=1}^n \frac{1}{\sqrt{2 \pi}} e^{- (x^{(i)}_j)^2/2}.
\]
We rewrite $\underset{(\overline{\bx},\overline{\by}) \sim {\cal S}_{\no} }{\Pr}\left[ {\cal S}_{\no}(\overline{\bx},\overline{\by}) \ge{\cal S}_{\yes}(\overline{\bx},\overline{\by}) \right]$ as

$$
\underset{\overline{\bx} \sim N(0,1)^{t\times n}, \overline{\by} \sim N(0,1+c^2)^t}{\Pr}\left[  {\cal S}(\ol{\bx}) \cdot \prod_{i \in [t]} \frac{1}{\sqrt{2 \pi (1+c^2)}} e^{-\frac{(\by^{(i)})^2}{2(1+c^2)}} \ge {\cal S}(\ol{\bx}) \cdot \frac{1}{n} \sum_{\ell \in [n]} \prod_{i \in [t]} \frac{1}{\sqrt{2 \pi c^2}} e^{-\frac{\big(\by^{(i)} - \bx^{(i)}_\ell \big)^2}{2 c^2}} \right].
$$
Thus it suffices to show that with high probability over $\overline{\by} \sim N(0,1+c^2)^t$, we have
\begin{equation}\label{eq:closeness}
\underset{\overline{\bx} \sim N(0,1)^{t\times n}}{\Pr} \left[ \prod_{i \in [t]} \frac{1}{\sqrt{2 \pi (1+c^2)}} e^{-\frac{(\by^{(i)})^2}{2(1+c^2)}} \ge \frac{1}{n} \sum_{\ell \in [n]} \prod_{i \in [t]} \frac{1}{\sqrt{2 \pi c^2}} e^{-\frac{\big( \by^{(i)} - \bx^{(i)}_\ell \big)^2}{2 c^2}} \right] \approx {\frac 1 2}.
\end{equation}
For notational  simplicity, given $\overline{w} \in \mathbb{R}^t$ and $\overline{y} \in \mathbb{R}^t$, we write $R(\overline{w},\overline{y})$ to denote $\underset{i \in [t]}{\prod} \frac{1}{\sqrt{2 \pi c^2}} e^{-\frac{\big( y^{(i)} - w^{(i)} \big)^2}{2 c^2}},$ so we can re-express the right hand side of the inequality in \Cref{eq:closeness} as
$$
\frac{1}{n} \sum_{\ell \in [n]} \prod_{i \in [t]} \frac{1}{\sqrt{2 \pi c^2}} e^{-\frac{\big( \by^{(i)} - \bx^{(i)}_\ell \big)^2}{2 c^2}} = \frac{1}{n} \sum_{\ell \in [n]} R\left( (\bx^{(1)}_\ell,\ldots,\bx^{(t)}_\ell),\overline{\by} \right).
$$

\begin{claim}\label{clm:Kolmogorov_dist}
With probability $1-1/\poly(n)$ over the outcome $\ol{y}$ of a random $\overline{\by} \sim N(0,1+c^2)^t$, the random variable $\frac{1}{n} \sum_{\ell \in [n]} R\left( \big( \bx^{(1)}_\ell,\ldots,\bx^{(t)}_\ell \big), \overline{y} \right)
$ (where the randomness is over $\overline{\bx} \sim N(0,1)^{t\times n}$) is $1/n^{\Omega(1)}$-close in Kolmogorov-Smirnov distance to a Gaussian random variable with mean $(2\pi(1+c^2))^{-t/2} \cdot e^{-\frac{\|\overline{y}\|_2^2}{2(1+c^2)}}$ and variance $\frac{1}{n} \cdot \Theta \bigg( (2 \pi \cdot c \sqrt{2 + c^2})^{-t} \cdot e^{-\frac{\|\overline{y}\|_2^2}{2 + c^2}} \bigg)$ .
\end{claim}

We finish the proof of~\Cref{clm:adv_random_guess} and defer the proof of~\Cref{clm:Kolmogorov_dist} to ~\Cref{sec:proof_claim_Kolmogorov}.

\begin{proofof}{\Cref{clm:adv_random_guess} using~\Cref{clm:Kolmogorov_dist}}
We say that $\ol{y}$ is \emph{good} if the $1/n^{\Omega(1)}$-closeness described in~\Cref{clm:Kolmogorov_dist} holds.
We have that
\begin{align}
& \underset{(\ol{\bx},\ol{\by})}{\Pr}\left[ \frac{1}{n} \sum_{\ell \in [n]} R\left( (\bx^{(1)}_\ell,\ldots,\bx^{(t)}_\ell), \overline{\by} \right) \le \prod_{i \in [t]} \frac{1}{\sqrt{2 \pi (1+c^2)}} e^{-\frac{(y^{(i)})^2}{2(1+c^2)}}\right] \nonumber\\
= & \underset{\ol{\by}}{\Pr}[\ol{\by} \text{ is good}] \cdot {\Pr}\left[ N\left({\frac {e^{-\frac{\|y\|_2^2}{2(1+c^2)}}}{(2\pi(1+c^2))^{t/2}}}, \frac{1}{n} \cdot \Theta \bigg( (2 \pi \cdot c \sqrt{2 + c^2})^{-t} \cdot e^{-\frac{\|\overline{y}\|_2^2}{2 + c^2}} \bigg) \right) \le {\frac {e^{-\frac{\|y\|_2^2}{2(1+c^2)}}}{(2\pi(1+c^2))^{t/2}}} \right] \label{eq:biggie}\\
& \pm \underset{\ol{\by}}{\Pr}[\ol{\by} \text{ is good}] \cdot 1/n^{\Theta(1)} \pm \underset{\ol{\by}}{\Pr}[\ol{\by} \text{ is not good}]\nonumber\\
= & 1/2 \pm 1/n^{\Theta(1)},\nonumber
\end{align}
where the last line used \Cref{clm:Kolmogorov_dist} and the fact that the mean of the Gaussian on the LHS of the inequality in \Cref{eq:biggie} equals the right hand side of \Cref{eq:biggie}, so the probability that the Gaussian exceeds its mean is exactly $1/2.$
\end{proofof}

\subsection{Proof of~\Cref{clm:Kolmogorov_dist}}\label{sec:proof_claim_Kolmogorov}
The high-level idea is to exploit the fact that $\frac{1}{n} \sum_{\ell \in [n]} R\left( \big( \bx^{(1)}_\ell,\ldots,\bx^{(t)}_\ell \big), \overline{y} \right)$ is a sum of independent random variables and establish closeness to a suitable Gaussian using the Berry-Esseen theorem.  To apply Berry-Esseen we need to show that the individual summand random variables are suitably ``nice'' and need to calculate the appropriate parameters (mean and variance) of the sum, since they determine the Gaussian to which the sum converges.   

We first calculate the expectation.
\begin{lemma}~\label{lem:expect1}
Fix any $z \in \R$ and any $c>0$. Then
\[
\underset{\bx \sim N(0,1)}{\E} \bigg[\frac{1}{c}  e^{-\frac{|z-\bx|^2}{2c^2}}\bigg]  = \frac{1}{\sqrt{1+c^2}} \cdot e^{-\frac{|z|^2}{2(1+c^2)}}
\]
\end{lemma} 
\begin{proof}
The proof is a computation:
\begin{eqnarray*}
\underset{\bx \sim N(0,1)}{\E} \bigg[\frac{1}{c}  e^{-\frac{|z-\bx|^2}{2c^2}}\bigg] &=& \int_x \frac{1}{\sqrt{2\pi}} e^{-\frac{x^2}{2}} \frac{1}{c}  e^{-\frac{|z-x|^2}{2c^2}} dx \\
&=& \frac{1}{\sqrt{2\pi} \cdot c} \int_x e^{-\frac{z^2}{2c^2}}
 \cdot e^{-x^2 \frac{(1+c^2)}{2c^2}}  \cdot e^{\frac{z \cdot x}{c^2}} dx \\ 
 &=& \frac{1}{\sqrt{1+c^2}} \cdot e^{-\frac{|z|^2}{2(1+c^2)}}  \int_x \frac{\sqrt{1+c^2}}{\sqrt{2\pi} c}  e^{-x^2 \frac{(1+c^2)}{2c^2}}  \cdot e^{\frac{z \cdot x}{c^2}}  \cdot e^{-z^2 \frac{1}{2(1+c^2)c^2}}dx \\
  &=& \frac{1}{\sqrt{1+c^2}} \cdot e^{-\frac{|z|^2}{2(1+c^2)}}  \int_x \frac{\sqrt{1+c^2}}{\sqrt{2\pi} c}  e^{-\frac{(1+c^2) \cdot \big( x-\frac{z}{1+c^2} \big)^2}{2c^2}} dx \\
 &=& \frac{1}{\sqrt{1+c^2}} \cdot e^{-\frac{|z|^2}{2(1+c^2)}}. \qedhere
\end{eqnarray*}
\end{proof}

\Cref{lem:expect1} implies the expectation bound stated in~\Cref{clm:Kolmogorov_dist}, i.e., 
\begin{equation} \label{eq:useme}
\Ex_{\ol{\bx} \sim N(0,1)^{t \times n}}\left[ \frac{1}{n} \sum_{\ell \in [n]} R\bigg( \big( \bx^{(1)}_\ell,\ldots,\bx^{(t)}_\ell \big), \ol{y} \bigg) \right] =  (2\pi(1+c^2))^{-t/2} \cdot e^{-\frac{\|\ol{y}\|_2^2}{2(1+c^2)}},
\end{equation}
since the $t$ rows $\bx^{(1)},\ldots,\bx^{(t)}$ of $\ol{\bx}$ are generated independently. 

Using~\Cref{lem:expect1}, we calculate the moments of $R(\bw,\ol{y})$, which we need to bound in order to  apply the Berry-Esseen theorem:

\begin{claim}\label{clm:bound_moments}
For any $\ol{y} \in \R^t, j \geq 1$, the $j$-th moment of $R(\bw,\ol{y})$ for a random $\bw \sim N(0,1)^t$ is
\[
\E\big[ R(\bw,\ol{y})^j \big]=(2 \pi)^{-jt/2} \cdot \frac{1}{(c^{j-1} \sqrt{j+c^2})^t} \cdot e^{-\frac{ \|\ol{y}\|_2^2 }{ 2 +2c^2/j }}.
\]
\end{claim}
\begin{proof}
We have
\begin{align*}
\E[R(\bw,\ol{y})^j]& = \E\left[ \prod_{i \in [t]} (1/\sqrt{2\pi})^j \frac{1}{c^j} e^{-\frac{\big( y^{(i)} - \bw^{(i)} \big)^2}{2 c^2} \cdot j} \right]\\
& = \prod_{i \in [t]} (1/\sqrt{2\pi})^j \frac{1}{c^{j-1} \cdot \sqrt{j}} \E\left[ \frac{1}{c/\sqrt{j}} e^{-\frac{\big( y(i) - \bw(i) \big)^2}{2 (c/\sqrt{j})^2}} \right]\\
& = \prod_{i \in [t]} (1/\sqrt{2\pi})^j \frac{1}{c^{j-1} \cdot \sqrt{j} \cdot \sqrt{1+c^2/j} } \cdot e^{- \frac{y_i^2}{2(1+c^2/j)}}\tag{by \Cref{lem:expect1}}\\
& = (2 \pi)^{-jt/2} \cdot\frac{1}{(c^{j-1} \sqrt{j + c^2})^t} \cdot e^{-\frac{\|\ol{y}\|_2^2}{2 + 2c^2/j}}. \qedhere
\end{align*}
\end{proof}

Next we use~\Cref{clm:bound_moments} to bound the precise expressions we will need to control for the Berry-Esseen theorem. In the rest of this subsection we fix $\ol{y} \in \R^t$ and write $\E[R]$ to denote the expectation $\underset{\bw \sim N(0,1)^t}{\E}[R(\bw,\ol{y})]$.  

\begin{corollary}\label{cor:bound_centered_2nd_4th}
For $0 < c < 0.1$ we have
\begin{equation} \label{eq:first}
\underset{\bw \sim N(0,1)^t}{\E}\left[ \bigg| R(\bw,\ol{y}) - \E[R] \bigg|^2 \right] =\Theta(1) \cdot {\frac 1 {(2 \pi c \sqrt{2 + c^2})^t}} \cdot e^{-\frac{\|y\|_2^2}{2 + c^2}}.
\end{equation}
and
\begin{equation}
\underset{\bw \sim N(0,1)^t}{\E}\left[ \bigg| R(\bw,\ol{y}) - \E[R] \bigg|^4 \right] = \Theta(1) \frac{1}{((2\pi)^2 c^3 \sqrt{4 + c^2})^t} \cdot e^{-\frac{ \|y\|_2^2 }{2 + c^2/2}}.
\label{eq:second}
\end{equation}
\end{corollary}
\begin{proof}
We have 
\[
\underset{\bw}{\E}\left[ \bigg| R(\bw,\ol{y}) - \E[R] \bigg|^2 \right] = \E[R(\bw,\ol{y})^2] - \E[R]^2.
\]
From the $j=1$ and $j=2$ cases of~\Cref{clm:bound_moments}, this is 
\[
\overbrace{
(1/\sqrt{2\pi})^{2t} \cdot \frac{1}{(c \sqrt{2 + c^2})^t} \cdot e^{-\frac{\|y\|_2^2}{2 + c^2}}}^{A} \ - \ 
\overbrace{(1/\sqrt{2\pi})^{2t} \cdot \frac{1}{(1+c^2)^t} \cdot e^{-\frac{ \|y\|_2^2 }{ 2(1+c^2) } \cdot 2}}^{B}.
\] 
To see that $A-B=\Theta(A)$ (which implies \Cref{eq:first}), we note that $\frac{1}{c \sqrt{2+c^2}} > 2 \cdot \frac{1}{1+c^2}$ and $- \frac{\|y\|_2^2}{2+c^2} > - \frac{\|y\|_2^2}{2(1+c^2)} \cdot 2$ for $0 < c < 0.1$. 

Turning to \Cref{eq:second}, we write $\underset{\bw}{\E}\left[ \bigg| R(\bw,\ol{y}) - \E[R] \bigg|^4 \right]$ as
$$
\E[R(\bw,\ol{y})^4] - 4 \E[R(\bw,\ol{y})^3] \E[R] + 6 \E[R(\bw,\ol{y})^2] \E[R]^2 - 4 \E[R]^4 + \E[R]^4.
$$

Using $j=1,2,3,4$ in \Cref{clm:bound_moments}, the above is equal to $(1/\sqrt{2 \pi})^{4t}$ times \ignore{\rnote{I ignore'd out a $(1/\sqrt{2\pi})^{2t}$ below, please check}}
\begin{align*}
& \overbrace{\frac{1}{(c^3 \sqrt{4 + c^2})^t} \cdot e^{-\frac{ \|y\|_2^2 }{2 + c^2/2}}}^{A'} \ - \ 4 \frac{1}{(c^2 \sqrt{3 + c^2})^t} \cdot e^{-\frac{\|y\|_2^2}{2 + 2c^2/3}} \cdot \frac{1}{(\sqrt{1+c^2})^t} \cdot e^{-\frac{ \|y\|_2^2 }{ 2(1+c^2) }}\\
& + 6 \ignore{(1/\sqrt{2\pi})^{2t}} \cdot \frac{1}{(c \sqrt{2 + c^2})^t} \cdot e^{-\frac{\|y\|_2^2}{2 + c^2}} \cdot \frac{1}{(\sqrt{1+c^2})^{2t}} \cdot e^{-\frac{ \|y\|_2^2 }{ 2(1+c^2) } \cdot 2} - 3 \frac{1}{(\sqrt{1+c^2})^{4t}} \cdot e^{-\frac{ \|y\|_2^2 }{ 2(1+c^2) } 
\cdot 4}.
\end{align*}
Similar to the above discussion, using $\frac{1}{c^3 \sqrt{4+c^2}} \ge 8 \cdot \bigg( \frac{1}{c^2 \sqrt{3+c^2} \cdot \sqrt{1+c^2}} + \frac{1}{(1+c^2)^4} \bigg)$ and $\frac{ \|y\|_2^2 }{2 + c^2/2} \le \min\{ \frac{\|y\|_2^2}{2 + 2c^2/3} + \frac{ \|y\|_2^2 }{ 2(1+c^2) }, \frac{\|y\|_2^2}{2 + c^2}+ \frac{ \|y\|_2^2 }{ 1+c^2 }, \frac{ 2\|y\|_2^2 }{ 1+c^2 } \}$ for $0 <c < 0.1$, the above is easily seen to be $\Theta(A')$, giving \Cref{eq:second}.
\end{proof}

To finish the proof of \Cref{clm:Kolmogorov_dist}, we recall the Berry-Esseen theorem which upper bounds the Kolmogorov-Smirnov distance between a normalized sum of independent random variables and the standard $N(0,1)$ Gaussian distribution:

\begin{theorem}[Berry-Esseen Theorem]
Let $\bX_1,\ldots,\bX_n$ be i.i.d.~random variables with $\E[\bX_i]=0$, $\E[\bX_i^2]=\sigma^2$, and $\E[|\bX_i|^3] = \rho$. Then the sum $\bY=\frac{\bX_1+\cdots+\bX_n}{\sigma \sqrt{n}}$ has $|\Pr[\bY \le \alpha]-\Pr[N(0,1) \le \alpha]| \leq O(1) \cdot \frac{\rho}{\sigma^3 \sqrt{n}}$ for any $\alpha$.
\end{theorem}

For the rest of this section let $\sigma$ denote the standard deviation of $\E[R(\bw,\ol{y})]$. Note that the Kolmogorov-Smirnov distance between $\frac{1}{n} \sum_{\ell \in [n]} R\left( (\bx^{(1)}_\ell,\ldots,\bx^{(t)}_\ell),\overline{y} \right)$ and the Gaussian distribution $N( \E[R(\bw,\ol{y})], \Var(R(\bw,\ol{y}))/n)$ is the same as the Kolmogorov-Smirnov distance between $\frac{1}{\sigma \sqrt{n}} \sum_{\ell \in [n]} \left( R\big( (\bx^{(1)}_\ell,\ldots,\bx^{(t)}_\ell),\overline{y} \big) - \E[R] \right)$ and $N(0,1)$.
By the Berry-Esseen theorem, this latter distance is
\begin{equation} \label{eq:be-output}
\frac{O(1)}{\sqrt{n}} \cdot \frac{\E\big[ \big|R(\bw,\ol{y})-\E[R] \big|^3 \big]}{\Var\big[ R(\bw,\ol{y}) \big]^{3/2}} = \frac{O(1)}{\sqrt{n}} \cdot \frac{\E\big[ \big|R(\bw,\ol{y})-\E[R] \big|^4 \big]^{1/2}}{\E\big[ \big|R(\bw,\ol{y})-\E[R] \big|^2 \big]},
\end{equation}
where we used the Cauchy-Schwartz inequality to bound $\E\big[ \big|R(\bw,\ol{y})-\E[R] \big|^3 \big]$ by $\E\big[ \big|R(\bw,\ol{y})-\E[R] \big|^2 \big]^{1/2} \cdot \E\big[ \big|R(\bw,\ol{y})-\E[R] \big|^4 \big]^{1/2}$. By~\Cref{cor:bound_centered_2nd_4th}, \Cref{eq:be-output} is at most 

\begin{equation}\label{eq:cal_distance}
\frac{O(1)}{\sqrt{n}} \cdot \frac{(1/\sqrt{2\pi})^{2t} \cdot (\frac{1}{c^3 \sqrt{4 + c^2}})^{t/2} \cdot e^{-\frac{ \|\ol{y}\|_2^2 }{2 + c^2/2}/2}}{(1/\sqrt{2\pi})^{2t} \cdot (\frac{1}{c \sqrt{2 + c^2}})^{t} \cdot e^{-\frac{\|\ol{y}\|_2^2}{2 + c^2}}} = \frac{O(1)}{\sqrt{n}} \cdot \left( \frac{2+c^2}{c \sqrt{4+c^2}} \right)^{t/2} \cdot e^{\|\ol{y}\|_2^2(\frac{1}{2+c^2} - \frac{1}{4+c^2})}.
\end{equation}

Recalling that $\ol{\by} \sim N(0,1+c^2)^t$ and that $t = \frac{\log n}{10 \log 1/c}$, it follows that with probability $1-\exp(-0.01 \cdot \frac{\log^2 n}{t})$ over the draw of $\ol{\by}$ we have $\|\ol{\by}\|_2^2 \le \frac{\log n}{2}$.
Hence we can bound \Cref{eq:cal_distance} by 
\[
\frac{O(1)}{\sqrt{n}} \cdot \left(\frac{2}{c}\right)^{t/2} \cdot e^{\|y\|_2^2/2},
\]
which is at most $1/n^{0.1}$ since $t = \frac{\log n}{10 \log 1/c}$ and $\|y\|_2^2 \le \frac{\log n}{2} $.
This concludes the proof of~\Cref{clm:Kolmogorov_dist}.


\subsection{Proof of \Cref{cor:stat_dist_k_sparse}}\label{sec:proof_k_sparse}

In this subsection we extend the proof of \Cref{lem:Gaussian_one_sparse} to finish the proof of \Cref{cor:stat_dist_k_sparse}.  The proof is very similar to the proof of \Cref{lem:Gaussian_one_sparse} with some minor differences.

The pdf of ${\cal S}_{\no}'$ at any point $(\overline{x},\overline{y})$ is the same as ${\cal S}_{\no}(\overline{x},\overline{y}).$  The pdf of ${\cal S}_{\yes}'$ is
\begin{equation}\label{eq:distribution_yes_prime}
{\cal S}_{\yes}'(\overline{x},\overline{y}) = \frac{k}{n} \sum_{\ell \in [0,n/k-1]} \left( \prod_{i \in [t]} \left[  \left( \prod_{j=1}^n \frac{1}{\sqrt{2 \pi}} e^{- (x^{(i)}_j)^2/2} \right) \cdot \frac{1}{\sqrt{2 \pi c^2}} e^{-\frac{\big(
\by^{(i)} - \frac{\bx^{(i)}_{\ell k+1}+\cdots+\bx^{(i)}_{\ell k+k}}{\sqrt{k}}
\big)^2}{2 c^2}} \right] \right).
\end{equation}

We begin by stating an analogue of \Cref{cor:probability_in_S}, which implies \Cref{cor:stat_dist_k_sparse} in the same way that \Cref{cor:probability_in_S} implies \Cref{lem:Gaussian_one_sparse}:

\begin{corollary} \label{cor:probability_in_Sprime}
Define the subset $S' \subseteq \R^{t \times n} \times \R^t$,
$S':=\{ (\overline{\bx},\overline{\by})  : {\cal S}_{\no}'(\ol{x},\ol{y})>
{\cal S}_{\yes}'(\ol{x},\ol{y})\}$.
Then both
$\underset{(\overline{\bx},\overline{\by}) \sim {\cal S}_{\no}' }{\Pr} \left[ (\overline{\bx},\overline{\by}) \in S' \right]$ and $\underset{(\overline{\bx},\overline{\by}) \sim {\cal S}_{\yes}' }{\Pr} \left[ (\overline{\bx},\overline{\by}) \in \overline{S'} \right]$ are at most ${\frac 1 2} + \frac{1}{(n/k)^{\Omega(1)}}$.\end{corollary}

It remains to establish \Cref{cor:probability_in_Sprime}.
As argued earlier, since $\overline{\bx}$ is generated the same way in ${\cal S}_{\no}'$ and ${\cal S}_{\yes}'$, we rewrite $\underset{(\overline{\bx},\overline{\by}) \sim {\cal S}_{\no}' }{\Pr}\left[{\cal S}_{\no}'(\overline{\bx},\overline{\by}) \ge {\cal S}_{\yes}'(\overline{\bx},\overline{\by})] \right]$ as
$$
\underset{(\overline{\bx}, \overline{\by})\sim {\cal S}_{\no}'}{\Pr}\left[ {\cal S}(\overline{\bx}) \cdot \prod_{i \in [t]} \frac{e^{-\frac{(\by^{(i)})^2}{2(1+c^2)}}}{\sqrt{2 \pi (1+c^2)}} \ge {\cal S}(\overline{\bx}) \cdot \frac{k}{n} \sum_{\ell \in [0,n/k-1]} \prod_{i \in [t]} \frac{e^{-\big(\by^{(i)} - \frac{\bx^{(i)}_{\ell k+1}+\cdots+\bx^{(i)}_{ \ell k+k}}{\sqrt{k}} \big)^2/(2 c^2)}}{\sqrt{2 \pi c^2}} \right].
$$

Thus it is sufficient to show that w.h.p. over $\overline{\by} \sim N(0,1+c^2)^t$,
\begin{equation}\label{eq:closeness_2}
\underset{\overline{\bx} \sim N(0,1)^{t\times n}}{\Pr} \left[ \prod_{i \in [t]} \frac{1}{\sqrt{2 \pi (1+c^2)}} e^{-\frac{(\by^{(i)})^2}{2(1+c^2)}} \ge \frac{k}{n} \sum_{\ell \in [0,n/k-1]} \prod_{i \in [t]} \frac{e^{-\big(\by^{(i)} - \frac{\bx^{(i)}_{\ell k+1}+\cdots+\bx^{(i)}_{\ell k+k}}{\sqrt{k}} \big)^2/(2 c^2)}}{\sqrt{2 \pi c^2}} \right] \approx 0.5.
\end{equation}

As before we use $R(\overline{w},\overline{y})$ to denote $\prod_{i \in [t]} \frac{1}{\sqrt{2 \pi c^2}} e^{-\frac{\big( y_i - w_i \big)^2}{2 c^2}}$ for $\overline{w} \in \mathbb{R}^t$ and $\overline{y} \in \mathbb{R}^t$, so the right hand side of \eqref{eq:closeness_2} becomes
$$
\frac{k}{n} \underset{\ell \in [0,n/k-1]}{\sum} R\left( \left(\frac{\bx^{(j)}_{\ell k}+\cdots+\bx^{(j)}_{\ell k+k-1}}{\sqrt{k}} \right)_{j \in [t]}, \overline{y} \right)
$$
by definition.

Finally, we bound the Kolmogorov-Smirnov distance.  Given the next claim, the proof of \Cref{cor:probability_in_Sprime} is concluded following the argument at the end of \Cref{sec:one_sparse} that proves \Cref{clm:adv_random_guess} using~\Cref{clm:Kolmogorov_dist}:

\begin{claim}\label{clm:Kolmogorov_dist_2}
With probability $1-1/\poly(n/k)$ over the outcome $\ol{y}$ of a random $\overline{\by} \sim N(0,1+c^2)^t$, the random variable 
$\frac{k}{n} \sum_{\ell=0}^{n/k-1} R\left( \big(\frac{\bx^{(j)}_{\ell k}+\cdots+\bx^{(j)}_{\ell k+k-1}}{\sqrt{k}} \big)_{j \in [t]}, \overline{y} \right)$
(where the randomness is over $\overline{\bx} \sim N(0,1)^{t\times n}$) is $1/(n/k)^{\Omega(1)}$-close in Kolmogorov-Smirnov distance to a Gaussian random variable with mean $(2\pi(1+c^2))^{-t/2} \cdot e^{-\frac{\|\overline{y}\|_2^2}{2(1+c^2)}}$ and variance $\frac{k}{n} \cdot \Theta \bigg( (2 \pi \cdot c \sqrt{2 + c^2})^{-t} \cdot e^{-\frac{\|\overline{y}\|_2^2}{2 + c^2}} \bigg)$ .
\end{claim}
\begin{proof}
Note that the random variable $\frac{k}{n} \underset{\ell \in [0,n/k-1]}{\sum} R\left( \big(\frac{\bx^{(j)}_{\ell k}+\cdots+\bx^{(j)}_{\ell k+k-1}}{\sqrt{k}} \big)_{j \in [t]}, \overline{y} \right)$ is distributed the same as $\frac{k}{n} \underset{\ell \in [n/k]}{\sum} R\left( \big( \bx'^{(j)}_\ell \big)_{j \in [t]}, \overline{y} \right)$ for a random $\bx' \sim N(0,1)^{t \times n/k}$, because each variable $\bx^{(j)}_i$ only appears once. Thus we can apply \Cref{clm:Kolmogorov_dist} to bound the Kolmogorov-Smirnov distance between $\frac{k}{n} \underset{\ell \in [n/k]}{\sum} R\left( \big( \bx'^{(j)}_\ell \big)_{j \in [t]}, \overline{y} \right)$ for $\bx' \sim N(0,1)^{t \times n/k}$ and \[
N\left(
(2\pi(1+c^2))^{-t/2} \cdot e^{-\frac{\|\overline{y}\|_2^2}{2(1+c^2)}} ,
\frac{k}{n} \cdot \Theta \bigg( (2 \pi \cdot c \sqrt{2 + c^2})^{-t} \cdot e^{-\frac{\|\overline{y}\|_2^2}{2 + c^2}} \bigg)
\right)
\] by $\frac{1}{(n/k)^{\Omega(1)}}$, and the claim is proved.
\end{proof}

%% file: Cumulant_lower_bound.tex
\section{Refinement of \Cref{thm:gaussian-lb}: a sample complexity lower bound based on cumulant size} \label{sec:refinement}

As mentioned in \Cref{sec:techniques-lb}, in this section we use \Cref{thm:gaussian-lb} to show, roughly speaking, that if the cumulants of a real random variable $\bX$ are small then many samples are required to test linearity under $\bX^n$ and Gaussian noise.

\begin{theorem} \label{thm:refinement}
For any $0 < \gamma<10^{-4}$ there exists a real random variable $\bX$ with mean 0, variance 1, and moments of all orders, such that
\begin{enumerate}
\item For any constant $\ell>2$, the $\ell$-th cumulant $\cum_\ell(\bX)$ is at most $O_\ell(\gamma)$;
\item For any constant $c>0$, any algorithm which is an $(\eps=0.99)$-tester of 1-linearity under ${\cal D} = \bX^n$ and Gaussian noise $N(0,c^2)$ must have sample complexity $\Omega \left(\min\{ \frac{1}{\gamma}, \log n \} \right)$.
\end{enumerate}
\end{theorem}

\begin{proof}
The real random variable $\bX$ is a mixture of the Gaussian $N(0,1)$ and the Bernoulli random variable $\{\pm 1\}$: a draw from $\bX$ is distributed as
\[ 
  \begin{cases}
    N(0,1)       & \quad \text{with probability } 1-\gamma \\
    \{ \pm 1 \}  & \quad \text{with probability } \gamma.
  \end{cases}
\]
It is straightforward to verify that $\bX$ has mean zero and unit variance.

We first show that for any constant $\ell>2$, the $\ell$-th cumulant $\cum_l(\bX)$ is $O(\gamma)$. Notice that the moment generating function $M(t)={\E}[e^{t\bX}]$ is $(1-\gamma)e^{t^2/2} + \frac{\gamma}{2} \cdot e^{-t} + \frac{\gamma}{2} \cdot e^{t}$. We rewrite this as
\begin{align*}
M(t) = (1-\gamma)e^{t^2/2} \left(1 + \frac{\gamma}{2(1-\gamma)} e^{-t-t^2/2} + \frac{\gamma}{2(1-\gamma)} e^{t-t^2/2} \right).
\end{align*}
Consequently 
\begin{align*}
\ln M(t) &= \ln \left[ (1-\gamma)e^{t^2/2} \left(1 + \frac{\gamma}{2(1-\gamma)} e^{-t-t^2/2} + \frac{\gamma}{2(1-\gamma)} e^{t-t^2/2} \right) \right] \\
& = \ln (1-\gamma) + t^2/2 + \ln \left(1 + \frac{\gamma}{2(1-\gamma)} e^{-t-t^2/2} + \frac{\gamma}{2(1-\gamma)} e^{t-t^2/2} \right)\\
& =  \ln (1-\gamma) + t^2/2 + \sum_{\ell=1}^\infty (-1)^{\ell-1} \cdot \left( \frac{\gamma}{2(1-\gamma)} e^{-t-t^2/2} + \frac{\gamma}{2(1-\gamma)} e^{t-t^2/2}\right)^\ell/\ell.
\end{align*}

For any monomial $t^j$ with $j>2$, its coefficient comes from the $\ell$th powers of the Taylor expansion of $\frac{\gamma}{2(1-\gamma)} e^{-t-t^2/2} + \frac{\gamma}{2(1-\gamma)} e^{t-t^2/2}$, where $\ell \le j$. From the definition $\ln M(t)=\sum_{\ell \geq 0} \frac{\cum_\ell(\bX)}{\ell!} t^\ell$, we have $\cum_\ell(\bX)=O(\gamma)$ for any constant $\ell>2$, giving the first part of the theorem.

Next we fix $t=\alpha \cdot \min\{\frac{1}{\gamma}, \frac{\log n}{\log 1/c} \}$ for a small constant $\alpha$ (say $\alpha \leq 10^{-4}$ for concreteness) and use \Cref{thm:gaussian-lb} to prove the second part.  For contradiction, we assume that there is a testing algorithm $A$ for the distribution ${\cal D}=\bX^n$ that has sample complexity $t$ under Gaussian noise $N(0,c^2)$.

We consider two different distributions for a random target vector $\bw$. The first distribution, denoted ${\cal W}_{\yes}$, is uniform over the $n$ canonical basis vectors $\{e_1,\ldots,e_n\}$. The second distribution, denoted ${\cal W}_{\no}$, is uniform over the following $n/100$ many 100-sparse vectors with disjoint supports: $\{  e_{100i+1} + \cdots + e_{100(i+1)} \}_{i \in \{0,1,\ldots,n/100-1}$. By the above assumption, $A$ with high probability successfully distinguishes  $\bw \sim {\cal W}_{\yes}$ from $\bw \sim {\cal W}_{\no}$ under ${\cal D}=\bX^n$ and noise $N(0,c^2).$

Any $w$ in the support of either ${\cal W}_{\yes}$  or ${\cal W}_{\no}$ is a 100-sparse vector.  Thus for any such $w$, the probability (over $t$ random vectors $\bx^{(1)},\ldots,\bx^{(t)}$ drawn i.i.d.~from $\bX^{n}$) that any $\bx^{(i)}$ has any $\pm 1$-entry in any coordinate where $w$ is nonzero, is at most $100 \gamma t \leq 100 \alpha \leq 0.01$.  In other words, with probability at least $0.99$ all entries of $\bx^{(1)},\ldots,\bx^{(t)}$ in coordinates corresponding to the support of $w$ are from $N(0,1)$.  This implies the existence of an algorithm $A'$ that can successfully distinguish $\bw \sim {\cal W}_{\yes}$ from $\bw \sim {\cal W}_{\no}$ \emph{under $N(0,1)^n$} with high probability:  algorithm $A'$ works simply by independently replacing each of the $n$ coordinates of each example with a draw from $\bits$ with probability $\gamma$ for each coordinate, and running $A$ on the resulting sample. By the argument above, the success probability of $A'$ is at most $0.01$ less than the success probability of $A$.  But the proof of~\Cref{thm:gaussian-lb} shows that no ($t' \le \frac{\log n/100}{10 \log 1/c}$)-sample algorithm can distinguish $\bw \sim {\cal W}_{\yes}$ from $\bw \sim {\cal W}_{\no}$ under $N(0,1)^n$ with advantage $\Omega(1)$. Taking the constant $\alpha$ in the definition of $t$ to be sufficiently small we have $t<t'$, which gives the desired contradiction that establishes the second part and concludes the proof.
\end{proof}

%% file: appen_proofs.tex

\section{Deferred Proofs}\label{sec:append_proofs}
We first finish the proof of~\Cref{clm:up_bound_cumulant} in \Cref{sec:up_bound_cumulant}. Then we bound the number of samples needed to estimate $m_{\ell}(w \cdot \bX^n)$ and $\kappa_{\ell}(w \cdot \bX^n)$ in \Cref{sec:moment_noiseless} and \Cref{sec:cumulant_noiseless} separately for the noiseless case. Finally we finish the proof of \Cref{lemma:number_samples_k_moments} and \Cref{thm:number_samples_Lk_w} in \Cref{sec:est_noise}.

We record two basic properties of moments that will be used in the next subsections when we are estimating moments and cumulants:
\begin{fact}\label{fact:properties_moments}
For any real random variable $\bX$, for any $n\geq k \geq 1$, we have
\begin{equation}\label{eq:monotone_root_moments}
\E[|\bX|^n]^{1/n} \ge \E[|\bX|^k]^{1/k}
\end{equation}
and
\begin{equation}\label{eq:monotone_moments}
\E[|\bX|^n] \ge \E[|\bX|^k] \E[|\bX|^{n-k}].
\end{equation}
In particular, for $\bX$ with mean zero and variance $\sigma^2$ and even $k$, we have $\E[\bX^k] \ge \E[\bX^2]^{k/2} = \sigma^k$.
\end{fact}

\begin{proof}
\Cref{eq:monotone_root_moments} follows by considering two random variables $\bA=\bX^{k}$, $\bB=1$, $p=n/k$ and $q=n/(n-k)$ and applying H\"{o}lder's inequality:
\[
\E[|\bA \bB|] \le \E[|\bA|^p]^{1/p} \E[|\bB|^q]^{1/q}.
\]
To prove \Cref{eq:monotone_moments} we apply \Cref{eq:monotone_root_moments}  twice with parameters $(n,k)$ and $(n,n-k)$, obtaining
\[
\E[|\bX|^n]^{k/n} \ge \E[|\bX|^k] \text{~and~} \E[|\bX|^n]^{(n-k)/n} \ge \E[|\bX|^{n-k}],
\]
which together imply $\E[|\bX|^n] \ge \E[|\bX|^k] \E[|\bX|^{n-k}]$.
\end{proof}

\subsection{Proof of~\Cref{clm:up_bound_cumulant}}\label{sec:up_bound_cumulant}

We fix the random variable $\bX$ in this proof and write simply $\cum_\ell$ for $\cum_\ell(\bX)$, $m_\ell$ for $m_\ell(\bX)$, and $\overline{m}_\ell$ for the $\ell$-th absolute moment $\E[|\bX|^\ell]$.

We apply induction to show that $|\cum_{\ell}| \leq \overline{m}_\ell \cdot {e^{\ell}} \cdot \ell!$ for every $\ell$, which directly implies the claim. \ignore{bound $\cum_\ell$ based on \Cref{eq:degree_l_cum}} 
In the base case $\ell=1$, we have $\cum_\ell=\E[X]=0$.
For the inductive step, we bound $\cum_{\ell+1}$ as follows:

\begin{align*}
|\cum_{\ell+1}| & \le |m_{\ell+1}| + \sum_{j=1}^{\ell} {\ell \choose j-1} |\cum_{j}| \cdot |m_{\ell+1-j}| \tag{\Cref{eq:degree_l_cum}}\\
& \le |m_{\ell+1}| + \sum_{j=1}^\ell \frac{\ell!}{(\ell-j+1)! (j-1)!} \cdot e^j \cdot j! \cdot \ol{m}_j  \cdot |m_{\ell+1-j}| \tag{inductive hypothesis}\\
& \le  \ol{m}_{\ell+1} + \sum_{j=1}^\ell \frac{\ell! \cdot j \cdot e^j}{(\ell-j+1)!} \cdot \ol{m}_j \cdot \ol{m}_{\ell+1-j} \tag{$|m_i| \leq \ol{m}_i$ for all $i$}\\
& \le  \ol{m}_{\ell+1} + \sum_{j=1}^\ell \frac{\ell! \cdot j \cdot e^j}{(\ell-j+1)!} \cdot \ol{m}_{\ell+1}^{j/(\ell+1)} \cdot \ol{m}_{\ell+1}^{(\ell+1-j)/(\ell+1)} \tag{\Cref{eq:monotone_root_moments}}\\
&= \ol{m}_{\ell+1} \parens*{1 + \ell! \sum_{j=1}^\ell {\frac{j \cdot e^{j}}{(\ell-j+1)!}}},\\
\end{align*}
where in the fourth line \Cref{eq:monotone_root_moments} was used to upper bound $\ol{m}_i \leq \ol{m}_{\ell+1}^{i/(\ell+1)} $ with $i=j$ and $i=\ell+1-j$.

To finish the proof we upper bound
\[
1 + \ell! \sum_{j=1}^\ell {\frac{j \cdot e^{j}}{(\ell-j+1)!}} \le 1 + \ell! \cdot \ell \cdot e^{\ell} \sum_{j=1}^\ell {\frac 1 {(\ell-j+1)!}} \le 1 + \ell! \cdot \ell \cdot e^{\ell} \cdot e \le 1 + \ell! \cdot \ell \cdot e^{\ell+1},
\]
which is less than $(l+1)! \cdot e^{\ell+1}$ for any $\ell \ge 1$.
\qed

\subsection{Estimating moments with no noise}\label{sec:moment_noiseless}

\Cref{thm:number_samples_Lk_w} is based on estimating cumulants.  
As a first step towards estimating cumulants, we start by considering how to estimate moments. The main result of this subsection, \Cref{lem:number_samples_moment_Z}, bounds the number of samples needed to estimate the $\ell$-th moment $m_\ell(\bY)$ when there is no noise.

The following upper and lower bounds on $m_\ell(\bY)$ in terms of the moments of $\bX$ will be useful:

\begin{claim}\label{clm:moment_Z}
Let $\E[\bX]=0,\Var[\bX]=1,|m_i(\bX)| < \infty$ for all $i$. For any even $\ell$ and any vector $w$,\ignore{\footnote{\color{brown} Xue: Change the statement for any vector $w$},} the random variable $\bY := w \cdot \bX^n$ satisfies
\[ 
\|w\|_2^{\ell} \leq m_\ell(\bY) \leq  \ell^{2\ell} \cdot m_{\ell}(\bX) \cdot \|w\|_2^{\ell}.
\]
\end{claim}

We defer the proof of~\Cref{clm:moment_Z} to~\Cref{sec:proof_moment_Z} and use~\Cref{clm:moment_Z}
to bound the number of samples which suffice to estimate $m_\ell(\bY)$:

\begin{lemma}\label{lem:number_samples_moment_Z}
Let $\E[\bX]=0,\Var[\bX]=1,|m_i(\bX)| < \infty$ for all $i$.
For any $\eps,\delta>0$, any even $\ell$, and any vector $w$ with $\|w\|_2 \in [1/C,C]$, let $\by^{(1)},\dots,\by^{(m)}$ be obtained according to $\by^{(i)}=w\cdot \bx^{(i)}$ where the $\bx^{(i)}$'s are i.i.d. according to $\bX^n$ and 
$m=O(\frac{(2\ell)^{4\ell} \cdot m_{2\ell}(\bX) \cdot C^{\ell}}{\delta \eps^2})$.
Then with probability at least $1-\delta$, we have that
\begin{equation} \label{eq:moment-deviation-bound}
\left|
{\frac {\sum_{i=1}^m (\by^{(i)})^\ell} m} - m_\ell(\bY)
\right| \leq \eps
\end{equation}
for the random variable $\bY = w \cdot \bX^n.$
\end{lemma}

\begin{proof}
It is clear that the expectation of ${\frac {\sum_{i=1}^m (\by^{(i)})^\ell} m}$ is $\E[\bY^\ell]=m_\ell(\bY)$.  We  bound the variance of ${\frac {\sum_{i=1}^m (\by^{(i)})^\ell} m}$ by first noting that for each $i$ we have
\[
\underset{\by^{(i)}}{\Var}[( \by^{(i)})^\ell]=
\underset{\bx^{(i)}}{\Var}[( w \cdot \bx^{(i)})^\ell] \le 
\underset{\bx^{(i)}}{\E}[ (w \cdot \bx^{(i)})^{2\ell}] = 
m_{2\ell}(\bY) \le (2\ell)^{4\ell} \cdot m_{2\ell}(\bX) \cdot C^{\ell}
\]
(where the last inequality is \Cref{clm:moment_Z}).
By independence, it follows that
\[
\Var\left[{\frac {\sum_{i=1}^m (\by^{(i)})^\ell} m}\right] \leq
(2\ell)^{4\ell} \cdot m_{2\ell}(\bX) \cdot C^{\ell}/m.
\]
By Chebyshev's inequality, it follows that \Cref{eq:moment-deviation-bound} holds with probability at least $1-\delta.$
\end{proof}

\subsection{Estimating cumulants with no noise}\label{sec:cumulant_noiseless}
Now we bound the number of samples which suffice to estimate $\cum_\ell(\bY)$ via its first $\ell$ moments (when there is no noise). The main result of this section is a preliminary result on estimating cumulants,   \Cref{lem:number_samples_cumulant_Z}, which will be used in the proof of \Cref{thm:number_samples_Lk_w}. 

\ignore{
}

\begin{lemma}\label{lem:number_samples_cumulant_Z}
Let $\bX$ be a symmetric real random variable with $\E[\bX]=0,\Var[\bX]=1$ and finite moments of all orders. 
There is an algorithm (depending on $\bX$) with the following property:  Let $w \in \R^n$ be any (unknown) vector with $\|w\|_2\in [1/C,C]$.
Given any $\eps,\delta>0$ and any even integer $\ell$, the algorithm takes as input $m=\poly\bigg( \ell!,m_{2\ell}(\bX),1/(\delta\eps), C^{\ell} \bigg)$ many independent random samples $(\bx^{(1)},\by^{(1)}),\dots,(\bx^{(m)},\by^{(m)})$ where each $\bx^{(i)} \sim \bX^n$ and each $\by^{(i)} = w \cdot \bx^{(i)}.$
It outputs an estimate $\wt{\cum}_\ell(\bY)$ of the $\ell$-th cumulant $\cum_\ell(\bY)$ of the random variable $\bY := w \cdot \bX^n$. With probability at least $1-\delta$, this estimate satisfies
\[ 
\bigg| \wt{\cum}_l(\bY) - \cum_\ell(\bY) \bigg| \le \eps.
\] 
\end{lemma}

\begin{proofof}{\Cref{lem:number_samples_cumulant_Z}}
We first apply~\Cref{lem:number_samples_moment_Z} to obtain, for each even $2 \leq k \leq \ell$, an estimate $\wt{m}_{{k}}(\bY):=\sum_{i=1}^m (w \cdot \bx^{(i)})^{{k}}/m$ of $m_{{k}}(\bY)$ such that with probability at least $\delta/\ell$, we have
\begin{equation} \label{eq:m-l-est}
\bigg| \wt{m}_{{k}}(\bY) - m_{{k}}(\bY) \bigg| \le \eps',
\end{equation}
where
$\eps'=\frac{\eps}{ (C \ell)^{C' \ell} \cdot m_\ell(\bX)}$ and $C'$ is a large absolute constant. 
For odd $1 \leq k \leq \ell$ we set $\tilde{m}_k(\bY)=0$ (recall that $\bX$ and $\bY= w \cdot \bX^n$ are symmetric and hence the actual value of $m_k(\bY)$ is $0$ for all odd $k$).  In the rest of the proof we assume that \Cref{eq:m-l-est} holds for all $k \in [\ell]$ (by a union bound, this is the case with probability at least $1-\delta$).

Then we apply \Cref{eq:cumulants_moments} to compute the value
\[
\wt{\cum}_\ell(\bY):=\sum_{k=1}^\ell (-1)^{k-1} (k-1)! B_{\ell,k}\bigg( \wt{m}_1(\bY),\ldots, \wt{m}_{\ell-k+1}(\bY) \bigg)
\]
which is our estimate of $\kappa_\ell(\bY).$  As in \Cref{fact:transfer_moments_cumulants}, $B_{\ell,k}$ are incomplete Bell polynomials 
\begin{equation} \label{eq:bell}
B_{\ell,k}(\wt{m}_1(\bY),\ldots, \wt{m}_{\ell-k+1}(\bY))=\sum \frac{\ell!}{j_1!\cdots j_{\ell-k+1}!} \left(\frac{\wt{m}_1(\bY)}{1} \right)^{j_1} \cdots \left(\frac{\wt{m}_{\ell-k+1}(\bY)}{(\ell-k+1)!}\right)^{j_{\ell-k+1}}
\end{equation}
whose summation is over all non-negative sequences $(j_1,\ldots,j_{\ell-k+1})$ with
\begin{equation}\label{eq:l_k_j}
j_1+\cdots+j_{\ell-k+1}=k \text{ and } j_1+2j_2 + \cdots + (\ell-k+1)j_{\ell-k+1}=\ell.
\end{equation}
Note that since $\tilde{m}_k(\bY)$ is zero for odd $k$, the only non-zero summands will correspond to sequences $j_1,\dots,j_{\ell-k+1}$ such that $j_1=j_3=\cdots=0$; thus we subsequently assume that all odd indices $j_1=j_3=\cdots=0$ and we need only consider the even indices $j_2,j_4,\dots$ in the subsequent analysis. 

We require upper and lower bounds on the individual summands appearing in the RHS of \Cref{eq:bell}, and to get such bounds the first step is to bound the corresponding expressions but with the actual moments $m_i(\bY)$ in place of the estimates $\tilde{m}_i(\bY)$.  So fix any $k \in [\ell]$ and given this $k$, let $h$ denote the largest even integer that is at most $\ell-k+1$. We note that for any sequence $j_2,\ldots,j_h$ satisfying \Cref{eq:l_k_j}, we have 
\begin{align*}
& \left(\frac{m_2(\bY)}{2} \right)^{j_2} \cdots \left(\frac{m_{h}(\bY)}{h!}\right)^{j_{h}} \\
& 
\le 2^{2 \cdot 2 j_2} \cdots h^{2h \cdot j_h} \cdot m_2(\bX)^{j_2} \cdots m_{h}(\bX)^{j_{h}} \cdot \|w\|_2^{2j_2 + \cdots + h j_h} \tag{\Cref{clm:moment_Z}}
\\
& \le \ell^{2(2j_2+\cdots+h \cdot j_{h})} \cdot m_2(\bX)^{j_2} \cdots m_{h}(\bX)^{j_{h}} \cdot \|w\|_2^{\ell} \\
&\le \ell^{2\ell} \cdot m_{\ell}(\bX) \cdot \|w\|_2^{\ell}. \tag{\Cref{eq:monotone_moments}}
\end{align*}
On the other hand, we also have
 \[
 \left(\frac{m_2(\bY)}{2} \right)^{j_2} \cdots \left(\frac{m_{h}(\bY)}{h!}\right)^{j_{h}} 
 \geq
\left(\frac{\|w\|_2^2}{2} \right)^{j_2} \cdots \left(\frac{\|w\|_2^h}{h!}\right)^{j_{h}} \ge \|w\|_2^{\ell}/\ell!,
\]
where the first inequality is again \Cref{clm:moment_Z}. The second inequality is based on $(k!)^{j_k} \le (k \cdot j_k)!$ for any $k \ge 2$ and $j_k \ge 0$ and $2 j_2 + \cdots + h j_h = \ell$ from \Cref{eq:l_k_j}.

Now we turn to bounding an actual summand on the RHS of \Cref{eq:bell} (which has the estimated rather than actual moments). For the upper bound, we have
\begin{align*}
& \left(\frac{\wt{m}_2(\bY)}{2} \right)^{j_2} \cdots \left(\frac{\wt{m}_{h}(\bY)}{h!}\right)^{j_{h}} \\
\le & \left(\frac{m_2(\bY)+\eps'}{2} \right)^{j_2} \cdots \left(\frac{m_{h}(\bY) +\eps'}{h!}\right)^{j_{h}} \tag{\Cref{eq:m-l-est}} \\
\le & \left(\frac{m_{2}(\bY)}{2} \right)^{j_2} \cdots \left(\frac{m_{h}(\bY)}{h!}\right)^{j_{h}} \cdot \exp\left( j_2 \cdot \frac{\eps'}{m_2(\bY)} + \cdots + j_{h} \cdot \frac{\eps'}{m_{h}(\bY)} \right)\\
\le & \left(\frac{m_2(\bY)}{2} \right)^{j_2} \cdots \left(\frac{m_{h}(\bY)}{h!}\right)^{j_{h}} \cdot \exp\left( \ell \cdot \eps' \cdot C^{\ell} \right),
\end{align*}
where we use the lower bound $m_i(\bY) \ge (\frac{1}{C})^{\ell}$ for any even $i \ge 2$ from \Cref{clm:moment_Z} in the last step. Since $\eps'<\frac{1}{2 \ell \cdot C^{\ell}}$, using $e^x \leq 1 + 2x$ for $x \in [0,1]$ we can further upper bound the above quantity as
\begin{align*}
& \left(\frac{m_2(\bY)}{2} \right)^{j_2} \cdots \left(\frac{m_{h}(\bY)}{h!}\right)^{j_{h}} \cdot \left(1 + 2\ell \cdot \eps' \cdot C^{\ell} \right) \\
\le & \left(\frac{m_2(\bY)}{2} \right)^{j_2} \cdots \left(\frac{m_{h}(\bY)}{h!}\right)^{j_{h}} + m_\ell(\bY)\cdot 2\ell \cdot \eps' \cdot C^{\ell} \tag{\Cref{eq:monotone_moments} }\\
\le & \left(\frac{m_2(\bY)}{2} \right)^{j_2} \cdots \left(\frac{m_{h}(\bY)}{h!}\right)^{j_{h}} + {\ell^{2\ell}} \cdot m_\ell(\bX) \cdot \|w\|_2^{\ell} \cdot 2\ell \cdot \eps' \cdot C^{\ell}. \tag{\Cref{clm:moment_Z}}
\end{align*}

For the lower bound, we have
\begin{align*}
& \left(\frac{\wt{m}_2(\bY)}{2} \right)^{j_2} \cdots \left(\frac{\wt{m}_{h}(\bY)}{h!}\right)^{j_{h}} \\
\ge & \left(\frac{m_2(\bY) - \eps'}{2} \right)^{j_2} \cdots \left(\frac{m_{h}(\bY) - \eps'}{h!}\right)^{j_{h}} \tag{\Cref{eq:m-l-est}}\\
\ge & \left(\frac{m_2(\bY)}{2} \right)^{j_2} \cdots \left(\frac{m_{h}(\bY)}{h!}\right)^{j_{h}} \cdot (1 - j_2 \eps' \cdot C^{j_2} - \cdots j_{h} \eps' \cdot C^{j_h})\\
\ge & \left(\frac{m_2(\bY)}{2} \right)^{j_2} \cdots \left(\frac{m_{h}(\bY)}{h!}\right)^{j_{h}} - {\ell^{2\ell}} \cdot m_\ell(\bX) \cdot \|w\|_2^{\ell} \cdot \ell \cdot \eps' \cdot C^{\ell}.
\end{align*}

Finally, we bound the number of possible sequences of $j_1,\ldots,j_{l-k+1}$ by $2^\ell$ and bound $\frac{\ell!}{j_1!\cdots j_{\ell-k+1}!}$ by $\ell!$. We also bound the sum of absolute coefficients in \Cref{eq:cumulants_moments} as $\sum_{k=1}^{\ell} (k-1)! \le \ell \cdot (\ell-1)! = \ell! $. These bounds imply that the error of $\wt{\cum}_\ell(\bY)$ is at most 
\[
2^\ell \cdot \ell! \cdot \ell! \cdot \big( \ell^{2\ell} \cdot m_\ell(\bX) \cdot C^{\ell} \cdot 2\ell \cdot \eps' \cdot C^{\ell} \big) \le \eps.
\]

\end{proofof}

\subsection{Proofs of \Cref{lemma:number_samples_k_moments} and \Cref{thm:number_samples_Lk_w}}\label{sec:est_noise}
We use the same framework as \Cref{lem:number_samples_cumulant_Z} to prove \Cref{thm:number_samples_Lk_w}. We first bound the moments of $\bY$ when there is noise present:

\begin{claim}\label{clm:moment_Z_noise}
\ignore{Let $\bX$ and $\ion$ be symmetric real random variables with mean zero, variance one, and finite moments of all orders.}
Let $\bX$ be a symmetric real-valued random variable with variance 1 and finite moments of all orders, and let $\ion$ be a symmetric real-valued random variable with finite moments of all orders.
For any even $\ell$ and any vector $w \in \R^n$, letting $\bY= w \cdot \bX^n + \ion$, the value of $m_\ell(\bY)=\E[\bY^\ell]$ satisfies
\[
m_\ell(\bY) \in \left[ \|w\|_2^{\ell}, 2^{\ell} \cdot \parens*{ (\ell)^{2 \ell} \cdot m_\ell(\bX) \cdot \|w\|_2^{\ell}+ m_\ell(\ion)} \right].
\]
\end{claim}
\begin{proof}
We have
\begin{align*}
m_{\ell}(\bY)
&=\underset{\bX^n,\boldeta \sim \ion}{\E}\left[  (w \cdot \bX^n + \boldeta)^{\ell} \right]\\
&=\E\left[ \sum_{i=0}^{\ell} {\ell \choose i} (w \cdot \bX^n )^i \boldeta^{\ell-i} \right]\\
& = \sum_{i=0}^\ell {\ell \choose i} \cdot m_i\bigg(  w \cdot \bX^n  \bigg) \cdot m_{\ell-i}\bigg(\ion \bigg) \tag{independence}\\
&\leq \sum_i {\ell \choose i} \cdot m_{\ell}\bigg(  w \cdot \bX^n \bigg)^{i/\ell} \cdot m_{\ell} \bigg(\ion \bigg)^{1-i/\ell} \tag{\Cref{eq:monotone_root_moments}}\\
&\le  \sum_i {\ell \choose i} \max \left\{ m_{\ell} \bigg(  w \cdot \bX^n \bigg), m_{\ell} \bigg(\ion \bigg) \right\}\\
&\le  2^{\ell} \cdot \left( m_{\ell}\bigg( w \cdot \bX^n \bigg) + m_{\ell} \bigg(\ion \bigg) \right).
\end{align*}
We obtain the claimed upper bound by plugging in the bound on $m_{\ell}\parens*{ w \cdot \bX^n }$ from \Cref{clm:moment_Z}.  Finally, the lower bound $m_\ell(\bY)$ follows from \Cref{clm:moment_Z} and $m_2(w \cdot \bX^n + \boldeta)=\E[(w \cdot \bX^n)^2] + \E[\boldeta^2] \ge \|w\|_2^2$:
\[
m_{\ell}(\bY) \ge m_2(\bY)^{\ell/2} \ge \|w\|_2^{\ell}.
\]
\end{proof}

By the same argument used to prove \Cref{lem:number_samples_moment_Z}, it holds that 
\[
m=O\parens*{\frac{\Var[ (\bY)^{\ell} ]}{\delta \eps^2}}=O\parens*{\frac{m_{2\ell}(\bY)}{\delta \eps^2}}=O\parens*{\frac{2^{2\ell} \cdot (2\ell)^{4\ell} \cdot \parens*{ m_{2\ell}(\bX) \cdot \|w\|_2^{\ell} + m_{2\ell}(\ion) }}{\delta \eps^2} }
\] random samples $\by^{(1)},\dots,
\by^{(m)}$ of $\bY= w \cdot \bX^n + \ion$ will with probability at least $1-\delta$ satisfy
\[ 
\bigg| \sum_{i=1}^m {\frac {(\by^{(i)})^{\ell}} m} - m_{\ell}(\bY) \bigg| \le \eps.
\]
This finishes the proof of \Cref{lemma:number_samples_k_moments}. \qed

\medskip
 
We apply the same approach of \Cref{lem:number_samples_cumulant_Z} to obtain the following corollary on estimating the cumulants of $w \cdot \bX^n + \ion$. The only difference compared to \Cref{lem:number_samples_cumulant_Z} is that now $m_2(w \cdot \bX^n + \ion)$ is in $[1/C^2,C^2+m_2(\ion)]$ rather than $[1/C^2,C^2]$. Since $(C^2+m_2(\ion))^{\ell} \le 2^{\ell}(C^{2\ell}+m_2(\ion)^{\ell})\le 2^{\ell} (C^{2\ell} + m_{2\ell}(\ion))$ by \Cref{fact:properties_moments} the running time is still a polynomial in $C^{\ell},2^{\ell}$, and $m_{2\ell}(\ion)$, and we obtain the following:

\begin{corollary}\label{cor:number_samples_noisy_cumulant}
\ignore{Let $\bX$ and $\ion$ be symmetric real-valued random variables with mean zero, variance 1, and finite moments of all orders.}
Let $\bX$ be a symmetric real-valued random variable with variance 1 and finite moments of all orders, and let $\ion$ be a symmetric real-valued random variable with finite moments of all orders.
There is an algorithm with the following property:  Let $w \in \R^n$ be any (unknown) vector with $\|w\|_2 \in [1/C,C]$.
Given any $\eps,\delta>0$ and any even integer $\ell$, the algorithm takes as input $m=\poly\parens*{ \ell!, m_{2\ell}(\bX) +m_{2\ell}(\ion),1/(\delta\eps), C^{\ell}  }$ many independent random samples of $\bY = w \cdot \bX^n + \ion$ and outputs an estimate 
\ignore{
}
$\wt{\cum}_{\ell}(\bY)$ of $\cum_{\ell}(\bY)$ that with probability at least $1-\delta$ satisfies
\[ 
\bigg| \wt{\cum}_{\ell}(\bY) - \cum_{\ell}(\bY) \bigg| \le \eps.
\]
\end{corollary}

Finally, it follows directly from \Cref{fact:additive} and \Cref{fact:homogeneous} that $\cum_{\ell}(\bY)=\cum_{\ell}(\bX) \cdot \|w\|_{\ell}^{\ell}+\cum_{\ell}(\ion)$.  By \Cref{cor:number_samples_noisy_cumulant}, with the claimed number of samples we can obtain an estimate $\wt{\cum}_{\ell}(\bY)$ of $\cum_{\ell}(\bY)$ which is additively accurate to within $\pm \eps \cdot \tau$. It follows that $|\frac{\wt{\cum}_{\ell}(\bY) - \cum_{\ell}(\ion)}{\cum_{\ell}(\bX)}-\|w\|_{\ell}^{\ell}| \le \eps$, and \Cref{thm:number_samples_Lk_w} is proved. \qed

\subsection{Proof of~\Cref{clm:moment_Z}}\label{sec:proof_moment_Z}

The lower bound is simple: since $\E[\bX]=0$ and $\E[\bX^2]=1$, the random variable $\bY = w \cdot \bX^n$ satisfies $\E[\bY^2]=\|w\|_2^2=1$, and hence $m_\ell(\bY)=\E[\bY^\ell] \ge \E[\bY^2]^{\ell/2}$ for any even $\ell$ by \Cref{fact:properties_moments}.

We use the following result to upper bound the moments of $\bY$:
\begin{lemma}[Marcinkiewicz-Zygmund inequality \cite{Zygmund77,Rio2009}]\label{lem:generalized_M_Zygmund}
If $\bY_1,\ldots,\bY_n$ are independent random variables such that $\E[\bY_i]=0$ and $\E[|\bY_i|^\ell] < \infty$ for all $1 \le \ell < \infty$, we have
$$
\E\left[ \left |\sum_{i=1}^n \bY_i \right|^\ell \right] \le B^\ell \E\left[ \left(\sum_{i=1}^n |\bY_i|^2\right)^{\ell/2} \right]
$$
for $B=\ell-1$. 
\ignore{
}
\end{lemma}

We apply Lemma~\ref{lem:generalized_M_Zygmund} with random variables $\bY_i = w_i \bX_i$:
\[
\E\left[ |\bY|^\ell \right] \le B^\ell \E\left[ \left( \sum_{i=1}^n w^2_i \bX_i^2 \right)^{\ell/2} \right].
\]
Then we expand the right hand side as
\begin{align*}
\E\left[ \left( \sum_{i=1}^n w^2_i \bX_i^2 \right)^{\ell/2} \right] & = \E \left[ \sum_{j \le \ell/2: i_1< i_2 < \cdots < i_j} \sum_{p_1+\cdots+p_j=\ell/2} {\ell/2 \choose p_1,\ldots,p_j} (w_{i_1} \bX_{i_1})^{2p_1} \cdots (w_{i_j} \bX_{i_j})^{2p_j} \right]\\
& = \E\left[ \sum_{p_1+\cdots+p_j=\ell/2} {\ell/2 \choose p_1,\ldots,p_j} \sum_{i_1=1}^n (w_{i_1} \bX_{i_1})^{2p_1} \sum_{i_2=i_1+1}^n (w_{i_2} \bX_{i_2})^{2p_2} \cdots \sum_{i_j=i_{j-1}+1}^n  (w_{i_j} \bX_{i_j})^{2p_j} \right]\\
& \le \sum_{p_1+\cdots+p_j=\ell/2} {\ell/2 \choose p_1,\ldots,p_j} \E\bracks*{\sum_{i_1=1}^n (w_{i_1} \bX_{i_1})^{2p_1}} \cdots \E\bracks*{\sum_{i_j=1}^n (w_{i_j} \bX_{i_j})^{2p_j} }\\
& \le \sum_{p_1+\cdots+p_j=\ell/2} {\ell/2 \choose p_1,\ldots,p_j} \left( \sum_{i_1=1}^n w_{i_1}^{2p_1} \E[ \bX^{2p_1}] \right) \cdots \left( \sum_{i_j=1}^n w_{i_j}^{2p_j} \E[ \bX^{2p_j} ] \right).
\end{align*}
We know $\sum_{i=1}^n w_i^\ell \le (\sum_{i=1}^n w_i^2)^{\ell/2} \le \|w\|_2^{\ell}$ for any $\ell \ge 4$. Thus we can upper bound the above by
\[
\sum_{p_1+\cdots+p_j=\ell/2} {\ell/2 \choose p_1,\ldots,p_j} \cdot \|w\|_2^{l} \E[\bX^{2p_1}] \cdots \E[\bX^{2p_j}] \le (\ell/2)! \cdot \|w\|_2^{\ell} \cdot \sum_{p_1+\cdots+p_j=\ell/2} \E[\bX^{2p_1}] \cdots \E[\bX^{2p_j}].
\]
\Cref{eq:monotone_moments} implies that $\E[\bX^p] \E[\bX^q] \le \E[\bX^{p+q}]$ for any even $p$ and $q$, which lets us upper bound the above by $(\ell/2)! \|w\|_2^{\ell} \cdot 2^{\ell/2} m_{\ell}(\bX)$. Finally we bound the constant $B^{\ell} \cdot (\ell/2)! \cdot 2^{\ell/2}$ by $\ell^{2 \ell}$, and the proof is complete.
\qed

%% file: noise-is-necessary.tex
\section{On the role of noise in our model}~\label{sec:noise-is-necc}

We now note that in the setting of Theorem~\ref{thm:main-informal}, if there is no  noise  in the labels, then there is a trivial algorithm which can in fact perform sparse recovery (not just testing) when the distribution ${\cal D} = ({\cal D'})^n$ is continuous. 
To see this, observe that given any candidate set $S=\{i_1,\dots,i_k\} \subset [n]$ of relevant coordinates, a noise-free example $(x,y=w \cdot x) \in \R^n \times \R$ provides a linear equation over the unknowns $\{w_i\}_{i \in S}$ in the obvious way, namely
\[
\sum_{i \in S} w_i x_i  = y.
\]
Since draws of $\bx \sim ({\cal D'})^n$ are in general position with probability 1, given $k+1$ noise-free examples from this distribution, 
\begin{itemize}
\item if $w$ is $k$-sparse then with probability 1 there is a unique subset $S$ of size at most $k$ (the support of $w$) for which the resulting system of equations has a solution in which each $w_i$ is nonzero, and
\item if $w$ is not $k$-sparse then with probability 1 no subset $S$ of size at most $k$ will be such that the resulting system has a solution.
\end{itemize}
